\documentclass[reqno]{tran-l}
\usepackage{a4wide}
\usepackage{amsmath,amsfonts,amssymb,
amsthm,dsfont,color, MnSymbol,relsize}
\usepackage{hyperref}
\usepackage{csquotes}

\usepackage[active]{srcltx}
\usepackage{dutchcal}

\usepackage{mathrsfs}
\usepackage{amscd,graphicx}

\newtheorem{theorem}{Theorem}[section]
\newtheorem{definition}[theorem]{Definition}
\newtheorem{proposition}[theorem]{Proposition}
\newtheorem{corollary}[theorem]{Corollary}
\newtheorem{lemma}[theorem]{Lemma}

\theoremstyle{definition}
\newtheorem{remark}[theorem]{Remark}

\newtheorem{hypothesis}[theorem]{Hypothesis}

\def\R{\mathbb{R}}

 \def\Op{\mathfrak{Op}} 

\def\X{\mathcal X}

\def\h{\mathcal{h}}

\def\z{\mathfrak{z}}

\def\bb1{{\rm{1}\hspace{-3pt}\mathbf{l}}}
\def\Ie0{[-\epsilon_0,\epsilon_0]}

\def\supp{\text{\sf supp}}

\def\dist{{\rm dist}}
\def\Tr{\mathbb{T}{\rm r}}
\def\Int{\mathfrak{I}\mathit{nt}}
\def\dist{\text{\sf dist}}
\def\supp{\mathop{\rm supp} \nolimits} % Support

\def\BC2{\mathbb{B}\big(\mathbb{C}^2\big)}

\def\Hes{\mathcal{Hess}}

\def\beq{\begin{equation}}
\def\eeq{\end{equation}}

\numberwithin{equation}{section}

\begin{document}

%opening
\title[Dirac crossings and magnetic fields]{Spectral analysis near a Dirac type crossing in a weak non-constant magnetic field}


%    author one information
\author{Horia D. Cornean}
\address{Department of Mathematical Sciences, Aalborg University, DK-9220 Aalborg, Denmark}
\email{cornean@math.aau.dk}

%    author two information
\author{Bernard Helffer}
\address{Laboratoire de Math{\'e}matiques Jean Leray, Universit{\'e} de Nantes and CNRS, Nantex, France; 
Laboratoire de Math{\'e}matiques d’Orsay, L'Universit{\'e} Paris-Sud, Universit{\'e} Paris-Saclay, Orsay, France}
\email{Bernard.Helffer@univ-nantes.fr}

\author{Radu Purice}
\address{\enquote{Simion Stoilow} Institute of Mathematics of the Romanian Academy, P.O. Box 1-764, 014700 Bucharest, Romania; 
Centre Francophone en Math{\'e}matiques de Bucarest, 21 Calea Grivitei Street, 010702 Bucharest, Romania}
\email{Radu.Purice@imar.ro}
\thanks{All three authors acknowledge support from Grant 8021-00084B of the Independent Research Fund Denmark $|$ Natural Sciences. H.C. was also supported by a Bitdefender Invited Professor Scholarship with IMAR, Bucharest. B.H. and R.P. acknowledge support from the International Research Network (GDRE) ECO-Math financed by CNRS and the Romanian Academy.}

%    \subjclass is required.
\subjclass[2020]{Primary: 81Q10, 81Q15. Secondary: 35S05}

\date{January 4th, 2021}

\dedicatory{We dedicate this paper to Viorel Iftimie for his 80-th anniversary.}

%    Abstract is required.
\begin{abstract}
This is the last paper in a series of three in which we have studied the Peierls substitution in the case of a weak magnetic field. Here we deal with two $2d$ Bloch eigenvalues which have a conical crossing. It turns out that in the presence of an almost constant weak magnetic field, the spectrum near the crossing develops gaps which remind of the Landau levels of an effective mass-less magnetic Dirac operator. 
\end{abstract}

\maketitle

\section{Introduction.}

This paper concludes our work concerned with the rigorous mathematical theory of the  Peierls-Onsager effective Hamiltonian for some magnetic spectral problems in $2$ dimensions (\cite{CHP-1}, \cite{CHP-2}).  This time we focus on the conical crossing case, which is typical for graphene-like systems. When subjected to external magnetic fields, such systems have been playing an important r\^ole in making the Quantum Hall effect possible at room temperatures \cite{Novo}. 
\subsection{On Peierls-Onsager substitution}
Let us briefly recall that the \textit{Peierls-Onsager substitution} is used by physicists (\cite{Pe}, \cite{Lu}) in the study of non-interacting electrons in a periodic potential (describing the lattice of atoms in the solid) and subjected to  a magnetic field.

In the absence of  a long range magnetic field, the periodic Hamiltonian is described in the Floquet representation as the sum of a countable family of multiplication operators living in some finite dimensional sub-spaces and given by some real functions  $\big\{\lambda_n:\mathcal{B}\rightarrow\mathbb{R}\big\}_{n\in\mathbb{N}}$ defined on the Brillouin domain $\mathcal B$;  
%and living in some finite dimensional sub-spaces;
these are the Bloch functions. We recall that the Brillouin domain may be considered, modulo some topological subtleties, as the unit cell in the momentum space with respect to the dual of the lattice defined by the periodic potential. The Peierls-Onsager substitution consists in replacing  the complete Hamiltonian in a magnetic field $B=dA$ by the effective Hamiltonian obtained by replacing 
the functions $\lambda_n(\theta)$ with $\theta\in\mathcal{B}$ by $\lambda_n\big(\theta-A(x)\big)$. As one can see even after this brief presentation, giving a sound mathematical meaning to these operators  is not quite evident and a rich literature has been devoted to this subject. We indicate here only a very subjective selection: \cite{HS}, \cite{Sj}, \cite{Ne-RMP}, \cite{Ne-LMP}, \cite{Be1}, \cite{Be2}, \cite{PST}, \cite{dNL}, \cite{CIP}, \cite{FT}, \cite{RB}, guided by our further developments and having in view all the references therein. Some important restrictive hypothesis imposed in these studies have been the existence of isolated Bloch bands (i.e. some function $\lambda_{n_0}$ that does not intersect with any other), the existence of Wannier  bases for such isolated Bloch bands and the constancy of the magnetic field. An important difficulty in using the Peierls-Onsager effective Hamiltonians for obtaining a detailed spectral information comes from the presence of the Bloch eigenprojections and the fact that they live on sub-spaces that depend on the magnetic field.

In our previous work \cite{CHP-1, CHP-2} we have  considered a 2-dimensional situation in which we could allow for some slow variation of the intensity of the magnetic field and prove a rather detailed spectral analysis of the effective Hamiltonians. First, in \cite{CHP-1} we studied the bottom of the magnetically perturbed spectrum in a narrow window  around the non-degenerate minimum of an isolated Bloch energy whose corresponding spectral projection had a zero Chern number and admitted an exponentially localized Wannier basis. Second, in \cite{CHP-2} we generalized these results to situations in which the unperturbed bottom of the spectrum comes from a single Bloch eigenvalue which either might cross with others outside the narrow window, or its corresponding spectral subspace has a non-trivial topology. 

Our general strategy is to isolate some simple effective Hamiltonian that on a small neighborhood of some point in the Brillouin domain approximates well the exact one in the absence of magnetic field. We have in view either a minimum of a Bloch eigenvalue (in \cite{CHP-1} and \cite{CHP-2}), where we use the quadratic form given by the Hessian of the given Bloch function, or a conical crossing point (in the present paper), where we use a $2\times 2$-matrix valued Dirac type Hamiltonian defined by the two crossing Bloch functions and their $1$-dimensional eigenprojections. The magnetic field that we considered in our papers \cite{CHP-1} and \cite{CHP-2} is of the form $B_{\epsilon,\kappa}(x)=\epsilon B^\circ+\epsilon\kappa B(\epsilon x)$  where $B^\circ$ is a constant magnetic field producing some spectral gaps controlled by $\epsilon\in[0,\epsilon_0]$ for some $\epsilon_0>0$ small enough and $\epsilon\kappa B(\epsilon x)$ is a slowly varying magnetic field considered as a perturbation controlled by $\kappa\in[0,\kappa_0]$ for some $\kappa_0\in(0,1]$. 

In this paper we will not need the requirement of \enquote{slow variation} and consider a magnetic field of the form $B_{\epsilon,\kappa}(x)=\epsilon B^\circ+\epsilon\kappa B(x)$ with some magnetic field $B$ with smooth components bounded together with all their derivatives. Our aim is to show that in a neighborhood of the \textit{special spectral point} corresponding to a  \enquote{conical crossing} of two Bloch  energy bands, the above magnetic field produces a family of spectral gaps with widths and separation controlled by $\epsilon$ and $\kappa$.

\subsection{The framework.}
\subsubsection{The periodic Hamiltonian.}\label{SS-per-Ham}

We work in \textit{a 2-dimensional configuration space} $\X\cong\R^2$ in which \textit{a regular lattice} $\Gamma$ is given, i.e. a lattice generated by two linearly independent vectors. For any $k\in\mathbb{N}\setminus\{0\}$ we denote by $BC^\infty(\X;\mathbb{R}^k)$ the linear space of vector valued smooth bounded functions on $\X$ having bounded derivatives of all orders. We denote by $\mathscr{S}(\X)$ the space of Schwartz test functions (smooth complex functions having rapid decay together with all their derivatives). We shall also denote by $L^2(\X)$ the Hilbert space of square integrable classes of complex functions with respect to the Lebesgue measure on $\X$.

\begin{definition}\label{D-HGamma}
	Let the $\Gamma$-periodic functions $V^\Gamma\in BC^\infty(\X;\mathbb{R})$ and $  A^\Gamma\in BC^\infty(\X;\mathbb{R}^2)$ and define the
	2-dimensional $\Gamma$-periodic Hamiltonian $H_\Gamma$ as the self-adjoint extension in  $L^2(\X)$ of the symmetric operator
	\beq\label{hc100} -\Delta_{A^\Gamma} \,+\,V^\Gamma:\,\mathscr{S}(\X)\rightarrow\mathscr{S}(\X)\,,
	\eeq
	with
	$$
	-\Delta_{A^\Gamma} :=
	\underset{j=1,2}{\sum}\big(-i\partial_{x_j}-A^\Gamma_j(x)\big)^2\,.
	$$
\end{definition}
In section XIII.16 of \cite{RS-4} it is proven that this extension exists, is unique, lower semi-bounded and has as domain in  $L^2(\X)$ the usual Sobolev space
$\mathscr{H}^2(\X)$ of square integrable functions with square integrable Laplacian (in the sense of distributions).
We recall that the periodic operators on $L^2(\mathbb{R}^d)$ admit a kind of 'partial diagonalization' given by the
Bloch-Floquet unitary map (see also section XIII.16 of \cite{RS-4}). We shall briefly present in Subsection \ref{SSS-BFZ-repr} some basic facts concerning the Bloch-Floquet Transformation. Here we shall recall only those notions that are necessary for the formulation of our main results.

We can define the quotient $\X/\Gamma\cong\mathbb{R}^2/\mathbb{Z}^2$ that will be identified as a topological group with \textit{the 2-dimensional torus} $\mathbb{T}\cong\mathbb{S}^1\times\mathbb{S}^1$. Here $\mathbb{S}^1$ is the unit circle as subset of elements of modulus 1 in $\mathbb{C}$ with the multiplication and topology induced from $\mathbb{C}\setminus\{0\}$.
We associate to the inclusion $\mathbb{Z}\subset\mathbb{R}$ the following decomposition: 
\beq\label{F-ent-dec}
t=[t+1/2]+\{t\}_2; \ [t]:=\max\{k\in\mathbb{Z}\,,\,k\leq t\big\}\in\mathbb{Z},\ \{t\}_2:=t-[t+1/2]\in  [-1/2,1/2),\quad\forall t\in\mathbb{R},
\eeq
and we define the elementary cell of $\X$ associated to it $\mathcal{E}:=\big\{x\in\X\,,\,-1/2\leq x_j<1/2,\,j=1,2\big\}\subset\X$.
We denote by $<\cdot,\cdot>:\X\times\X^*\rightarrow\mathbb{R}$ the duality map on $\X\times\X^*$ and define the following dual objects: the dual lattice
\beq\label{DF-Gamma-star}
\Gamma_*:=\big\{\gamma^*\in\X^*,\ <\gamma^*,\gamma>\in 2\pi\mathbb{Z}\big\},
\eeq
the dual elementary cell associated with the same decomposition \eqref{F-ent-dec} denoted by $\mathcal{B}$ and called \textit{the Brillouin zone} and the quotient of the duals $\mathbb T_{\star}:= \X^*/\Gamma_*$. Let us emphasize that although $\X$ and $\X^*$ will be just two distinct copies of $\mathbb{R}^2$ we shall keep this notation that allows us to have a clear distinction between the configuration and the momentum spaces.

We define the Bloch-Floquet-Zak transform of a test function $\phi\in\mathscr{S}(\X)$ by the formula
\beq\label{b-f-z}
\big(\check{\mathcal{U}}_{\Gamma}\phi\big)
(x,\theta):=(2\pi)^{-1}\sum\limits_{\gamma\in\Gamma}e^{- i[(x_1-\gamma_1)\theta_1+(x_2-\gamma_2)\theta_2]}\phi(x-\gamma),\qquad\forall x\in\X,\ \forall\theta\in\mathbb{R}^2.
\eeq
We notice that for any $\phi\in\mathscr{S}(\X)$ we have the following behavior:
\begin{align}\label{F-defF-1}
\forall\alpha\in\Gamma:&\quad\big(\check{\mathcal{U}}_{\Gamma}\phi\big)
(x+\alpha,\theta)=\big(\check{\mathcal{U}}_{\Gamma}\phi\big)
(x,\theta),\quad\forall(x,\theta)\in\X\times\mathbb{R}^2, \\ \label{F-defF-2}
\forall\nu\in\mathbb{Z}^2:&\quad\big(\check{\mathcal{U}}_{\Gamma}\phi\big)
(x,\theta+2\pi \nu)=e^{-2\pi i(x_1\nu_1+x_2\nu_2)}\big(\check{\mathcal{U}}_{\Gamma}\phi\big)
(x,\theta),\quad\forall(x,\theta)\in\X\times\mathbb{R}^2.
\end{align}
Due to the periodicity in the $x$-variable \eqref{F-defF-1}, we can project this variable on the 2-dimensional torus, that we denote by $\mathbb{T}$ and consider functions defined on $\mathbb{T}\times\mathbb{R}^2$ satisfying condition \eqref{F-defF-2}. Property \eqref{F-defF-2} suggests to restrict the variable $\theta\in\mathbb{R}^2$ to the square $(-\pi,\pi )^2$. One can prove that the transformation $\check{\mathcal{U}}_\Gamma$ defines a unitary operator $L^2(\X)\overset{\sim}{\rightarrow}L^2(\mathbb{T})\otimes L^2\big ((-\pi,\pi)^2)\big )$.
In this representation the periodic Hamiltonian $H_\Gamma$ defined in Definition \ref{D-HGamma} becomes the operator of multiplication with an operator-valued function of $\theta\in\mathbb{R}^2$ taking values self-adjoint operators $\check{H}(\theta)$ 
acting in $L^2(\mathbb{T})$. The following result is well known:

\begin{proposition}\label{P-BF}
The operators $\check{H}(\theta)$ with $\theta\in\mathbb{R}^2$ are self-adjoint, lower semi-bounded with compact resolvent in $L^2(\mathbb{T})$ and we shall choose their eigenvalues $\{\lambda_k(\theta)\}_{k\in\mathbb{N}}$ (called \textit{Bloch eigenvalues}) in increasing order  taking into account their multiplicity. The spectrum of $\check{H}(\theta)$ is $(2\pi\mathbb {Z})^2$ periodic.
\end{proposition}

\subsubsection{The magnetic field.}\label{SSS-MField}

We are interested in exhibiting  a structure of gaps created in the band spectrum of $H_\Gamma$ given by Definition \ref{D-HGamma} and Proposition \ref{P-BF} by a weak constant magnetic field and in studying their stability when perturbing 
the magnetic field by a smaller bounded smooth magnetic field that is not supposed to be constant or slowly varying.  We note that even proving continuity of the spectrum as a set when long range magnetic perturbations are involved has been a challenging problem \cite{Bec-Bel, BBdN, BC, Corn, CN, CNP, CP-1, CP-2}. 

Given $(\epsilon,\kappa)\in[0,1]\times[0,1]$, we shall consider a magnetic field of the form
\begin{equation}\label{Bek}
B_{\epsilon,\kappa}(x)=\epsilon B^\circ+\epsilon\kappa B(x)\,,
\end{equation}
where $B^\circ$ is a constant magnetic field that we shall take to be positive 
and
$\kappa B(x)$  is a weak magnetic field considered as a perturbation of $B^\circ$.  
$B(x)$ is of class $BC^\infty(\X;\mathbb{R})$.
Let us choose some smooth  {\it vector potentials} $A^\circ:\X\rightarrow\mathbb{R}^2$ and
$A:\X\rightarrow\mathbb{R}^2$ such that:
\begin{equation}\label{defAepsilon0}
B^\circ=\partial_1A^\circ_2-\partial_2A^\circ_1\, ,\quad B=\partial_1A_2-\partial_2A_1\,,
\end{equation}
and 
\begin{equation}\label{defAepsilon}
  A^{\epsilon,\kappa}(x):=\epsilon A^\circ(x)+\kappa\epsilon  A(x)\,, \quad B_{\epsilon,\kappa}=\partial_1A^{\epsilon,\kappa}_2-\partial_2A^{\epsilon,
	\kappa}_1\,.
\end{equation}
The vector potential $A^\circ$ is considered in the {\it transverse gauge}, i.e.
\beq \label{defA0}
A^\circ(x)\,=\,(1/2)\big(-B^\circ x_2,B^\circ x_1\big).
\eeq
We consider the following magnetic Schr\"{o}dinger operator, that is essentially self-adjoint on $\mathscr{S}(\X)$:
\begin{equation}\label{mainH}
H^{\epsilon,\kappa}_\Gamma:=\Big(-i\partial_{x_1} - A^\Gamma_1(x)-A^{\epsilon,\kappa}_1(x)\Big)^2 + \Big(-i\partial_{x_2} - A^\Gamma_2(x)-A^{\epsilon,\kappa}_2(x)\Big)^2 + V^\Gamma (x)
\end{equation} 
and treat it as a perturbation (controlled by the small parameter $\kappa\in[0,1]$) of the following operator
\beq\label{FD-Hepsilon}
H^{\epsilon}_\Gamma:=\Big(-i\partial_{x_1} - A^\Gamma_1(x) +\epsilon B^\circ x_2/2\Big)^2 + \Big(-i\partial_{x_2} - A^\Gamma_2(x)-\epsilon B^\circ x_1/2\Big)^2 + V^\Gamma (x)
\eeq
that is also essentially self-adjoint on $\mathscr{S}(\X)$. 

\subsection{Formulation of the main result.}

\begin{figure}[h]
\centering
\includegraphics[width=16 cm, height=11 cm]{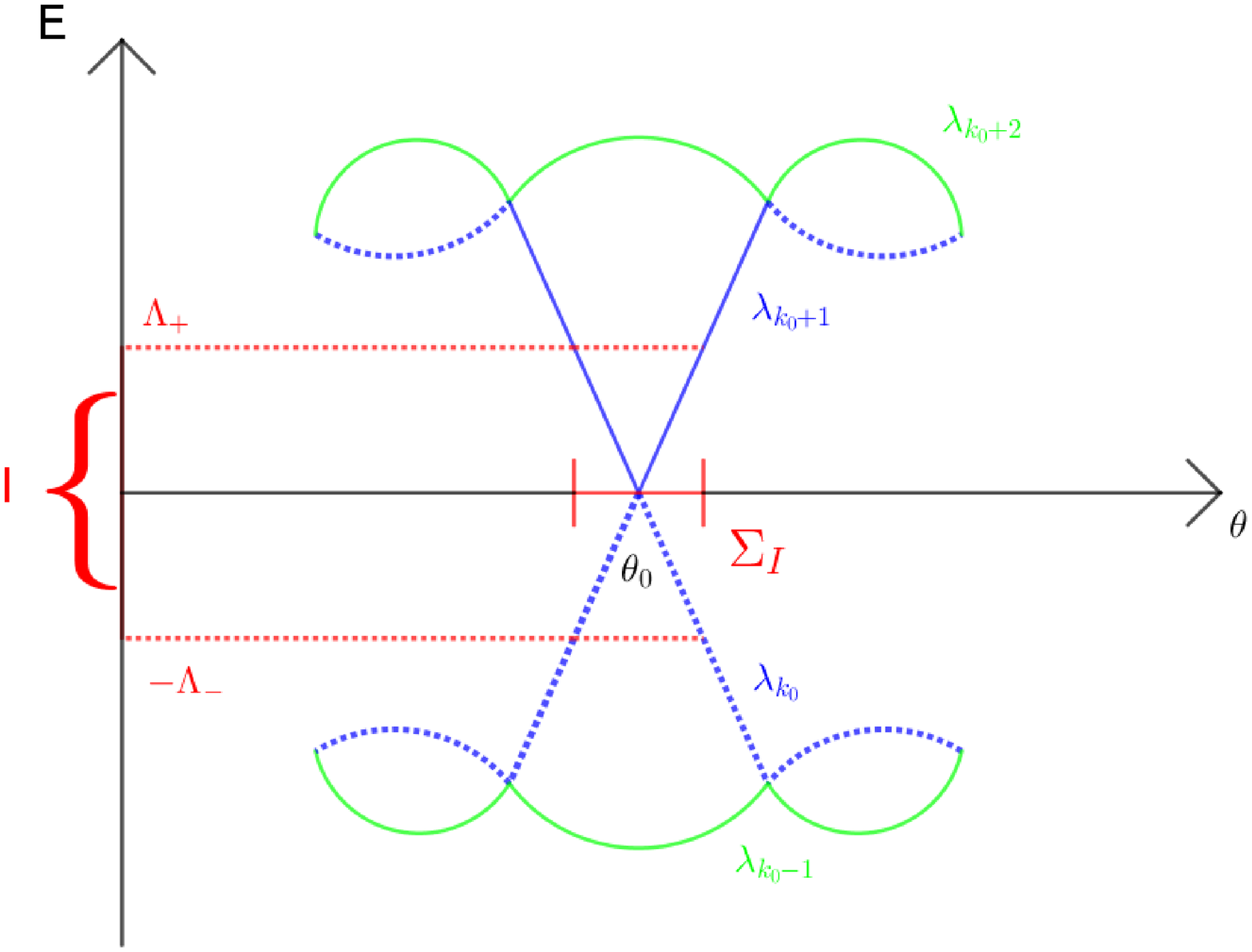}
\caption{The unperturbed spectrum}
\label{Fig:1}
\end{figure}

We denote by $d:\mathbb{T}_*\times\mathbb{T}_*\rightarrow\mathbb{R}_+$ the geodesic distance on the space $\mathbb{T}_*$.
For any $\theta_0 \in\mathbb{T}_*$ and  $r>0$ we define the closed ball
$$
B_r(\theta_0 ):=\big\{\theta\in\mathbb{T}_*,\,d(\theta,\theta_0) \leq r\big\}\,,
$$
\begin{hypothesis}\label{H-1} 
	There exists a compact interval $I:=[-\Lambda_-,\Lambda_+]\subset\mathbb{R}$ containing $0$ in its interior, an index $ k_0\in\mathbb{N}\setminus\{0\}$, a point $\theta_0\in  (-\pi,\pi)^2$ and a compact neighborhood $\Sigma_I\subset  (-\pi,\pi)^2$ of $\theta_0$, diffeomorphic to the unit disk, such that:
	\begin{align}
	&I\cap\lambda_k(\mathbb{T}_*)\neq\emptyset\Rightarrow k\in\{k_0,k_0+1\}, \nonumber \\
	& [-\Lambda_-,0]=\lambda_{k_0}\big(\Sigma_I\big),\quad [0,\Lambda_+]=\lambda_{k_0+1}\big(\Sigma_I\big), \nonumber \\ \label{F-unique-int}
	& \lambda_{k_0}(\theta)=\lambda_{k_0+1}(\theta)\Rightarrow\theta=\theta_0.
	\end{align}
\end{hypothesis}

For $\theta\in\Sigma_I$ we shall denote by
\beq\label{DF-lambda-pm}
\lambda_-(\theta):=\lambda_{k_0}(\theta),\qquad\lambda_+(\theta):=\lambda_{k_0+1}(\theta).
\eeq
We now express the nature of the touching of  $\lambda_-$ and $\lambda_+$ at $\theta_0\,$, the so-called conical crossing type.
\begin{hypothesis}\label{H-2}~\\
The map $\Sigma_I\ni \theta \mapsto\mathcal{d}(\theta):= \lambda_-(\theta) \lambda_+(\theta)$ has a non-degenerate maximum value equal to zero at $\theta_0$.
\end{hypothesis}

We need one more notation before stating the main result. For any two subsets $M_1,M_2$ in a metric space $\big(\mathcal{M},d\big)$ we denote by
\begin{equation}\label{apr-1}
d_{\text{\tt H}}\big(M_1,M_2\big):=\max\big\{\underset{x\in M_1}{\sup}\underset{y\in M_2}{\inf}d(x,y)\,,\,\underset{x\in M_2}{\sup}\underset{y\in M_1}{\inf}d(x,y)\big\}
\end{equation}
their \textit{Hausdorff distance}.	

\begin{theorem}\label{T-Main-2}
	Let us 
	assume that Hypotheses \ref{H-1}  and~\ref{H-2} hold true.
	Let $H_\Gamma^{\epsilon,\kappa}$ be  the magnetic Hamiltonian	in \eqref{mainH}
	with a magnetic field $B^{\epsilon,\kappa}$ satisfying  \eqref{Bek}. Then there exists a self-adjoint operator $\mathfrak{L}$ acting on $L^2(\mathbb{R})$ with discrete spectrum $\sigma(\mathfrak{L})$ symmetric with respect to the origin, containing $0$ and with all the eigenvalues of multiplicity $1$, such that for any $L>0$ situated in the middle of a gap of $\sqrt{B^\circ}\sigma(\mathfrak{L})$, there exist   positive  $\epsilon_L$, $\kappa_L$, and  $C_L$ such that for  $0 < \epsilon \leq \epsilon_L$ and $\kappa \in [0,\kappa_L]$, we have  
	\begin{equation*}
	d_{\text{\tt H}}\left( \sigma\big(H^{\epsilon,\kappa}_\Gamma\big)  \cap \big(- L\epsilon^\frac 12, L \epsilon^\frac 12\big) \,,\, (\epsilon B^\circ)^\frac 12 \sigma (\mathfrak{L}) \cap  \big(- L\epsilon^\frac 12, L \epsilon^\frac 12\big) \right)\ \leq\  C_L\, (\sqrt{\kappa\epsilon} + \epsilon) .
	\end{equation*}
\end{theorem}

\begin{figure}[h]
\centering
\includegraphics[width=14 cm, height=11 cm]{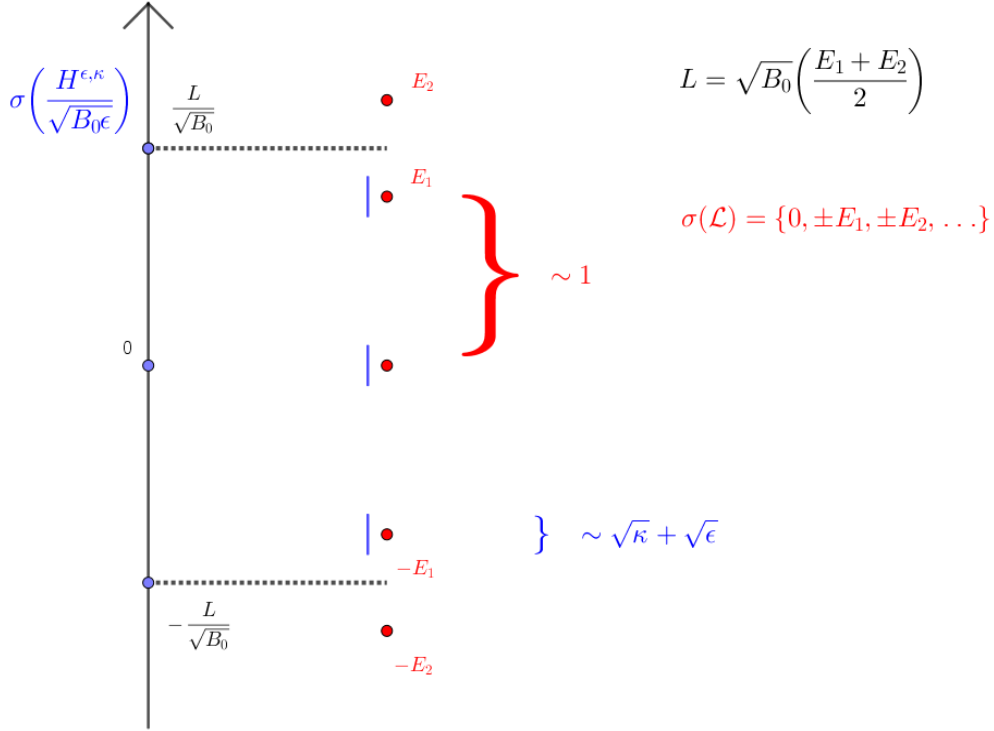}
\caption{The perturbed spectrum}
\label{Fig:2}
\end{figure}

\begin{remark}
 The set $(\epsilon B^\circ)^\frac 12 \sigma (\mathfrak{L})\cap  \big(- L\epsilon^\frac 12, L \epsilon^\frac 12\big)$ consists of finitely many  isolated points (of the order of the integer part of $L$) situated at a distance of order $\sqrt{\epsilon}$ from each other. Thus when both $\epsilon_L$ and $\kappa_L$ are small enough, the set $\sigma\big(H^{\epsilon,\kappa}_\Gamma\big)  \cap \big(- L\epsilon^\frac 12, L \epsilon^\frac 12\big)$ develops gaps of order $\sqrt{\epsilon}$, uniformly in $\kappa$. 
\end{remark}

\begin{remark}\label{P-CP1} It turns out that we only have to prove Theorem \ref{T-Main-2} for $\kappa=0$. This is because of the following statement which is a direct corollary of Theorem 3.1 in \cite{CP-1}:
\begin{quote}
	\textit{Suppose that $L$ is chosen as in Theorem \ref{T-Main-2}. Then there exists a constant $C>0$ such that}
	$$
	d_{\text{\tt H}}\Big(\sigma\big(H^{\epsilon,\kappa}_\Gamma\big)\cap (-L\sqrt{\epsilon},L\sqrt{\epsilon})\,,\,\sigma\big(H^{\epsilon}_\Gamma\big)\cap (-L\sqrt{\epsilon},L\sqrt{\epsilon})\Big)\,\leq\,C\sqrt{\kappa\epsilon}.
	$$
\end{quote}

 The result in \cite{CP-1} is very robust and in some sense optimal: it states that a small magnetic perturbation of order $\kappa\epsilon $ could create (large) gaps which are of order $\sqrt{\kappa \epsilon}$, but not larger. Its proof is based on geometric perturbation theory which uses a $\kappa\epsilon$ dependent partition of unity, combined with a use of magnetic gauge covariance. 
\end{remark}

An important difficulty in the proof of Theorem \ref{T-Main-2} comes from the fact that while in \cite{CHP-2} we were working near the bottom of the spectrum, using positivity conditions for invertibility,  we are now working somewhere in the \enquote{bulk} of the spectrum, in the interval $I\subset\sigma\big(H_\Gamma\big)$. A whole procedure for creating a spectral gap in the studied region inside the interval $I$ with some stability with respect to the magnetic field perturbation has to be elaborated. Moreover we have to replace the $1$-dimensional smooth unit norm global section defining the \enquote{quasi-band} associated with $\lambda_0$ near its minimum in \cite{CHP-2} with a smooth global orthonormal pair of sections having some \enquote{good behavior} with respect to the induced spectral gap, in order to define a kind of quasi-band associated with the two Bloch levels near their crossing point. 

\subsection{Some preliminaries and notations.}

\subsubsection{Notations.}

Given any finite dimensional real vector space $\mathcal{V}$ we denote by $C^\infty(\mathcal{V})$ the  space of complex valued smooth functions $\mathcal{V}\rightarrow\mathbb{C}$ and consider its sub-spaces $C_0^\infty(\mathcal{V})$ of functions with compact support, $BC^\infty(\mathcal{V})$ of bounded functions with bounded derivatives of all orders, $C^\infty_{\text{\sf pol}}(\mathcal{V})$ (resp. $C^\infty_{\text{\sf pol,u}}(\mathcal{V})$) of polynomially bounded (resp. uniformly polynomially bounded) functions.
When restricting to real functions we shall use the notation $C^\infty\big(\mathcal{V};\mathbb{R}\big)$ and similar ones for the above specified sub-spaces. 
We shall also consider $\mathscr{S}(\mathcal{V})$ and $\mathscr{S}^\prime(\mathcal{V})$ the usual dual pair of Schwartz test functions and tempered distributions.
We shall also use test functions with values in some finite dimensional vector space $\mathbb{V}$ and denote them by $\mathscr{S}\big(\mathcal{V};\mathbb{V}\big)$. We denote by $\tau_v$ the translation by $-v\in\mathcal{V}$ acting on various classes of functions (and distributions) on $\mathcal{V}$. We use the notation $<v>:=\sqrt{1+|v|^2}$ for any $v\in\mathcal{V}$.

For any Banach space $\mathcal{B}$ we denote by $\mathbb{B}(\mathcal{B})$ the algebra of continuous linear operators in $\mathcal{B}$. For any Hilbert space $\mathcal{H}$ we denote by $\mathbb{U}(\mathcal{H})$ its subset of unitary operators, by $\mathbb{P}(\mathcal{H})$ the family of $1$-dimensional orthogonal projections and by $\mathbb{LP}(\mathcal{H})$ the family of orthogonal projections. Given two Hilbert spaces $\mathcal{H}_1$ and $\mathcal{H}_2$ let $\mathbb{U}\big(\mathcal{H}_1;\mathcal{H}_2\big)$ be the group of unitary operators from $\mathcal{H}_1$ to $\mathcal{H}_2$.  We shall denote by $\mathbb{L}\big(\mathscr{V}_1;\mathscr{V}_2\big)$ the space of continuous linear operators from the topological vector space $\mathscr{V}_1$ to the topological vector space $\mathscr{V}_2$, endowed with the topology of uniform convergence on bounded sets.

Given any family of vectors $F$ in a vector space $\mathcal{K}$ over a field $\mathbb{K}$ (either $\mathbb{R}$ or $\mathbb{C}$) we denote by $\mathcal{L}_{\mathbb{K}}[F]$ the linear space they generate over $\mathbb{K}$.

For any subset $M\subset\mathcal{T}$ in a topological space $\mathcal{T}$ we denote by $\mathring{M}$ its \textit{interior}.

Once we are given a $2$-dimensional regular lattice $\Gamma\subset\X$  in a 2-dimensional real affine space $\X$, if we fix an origin for it, we have a precise realization of the configuration space as 2-dimensional real linear space and of $\Gamma$ as a copy of $\mathbb{Z}^2$.
{Once we have fixed the linear structure on $\X\cong\mathbb{R}^2$ induced by the regular lattice we denote by $\X^*\cong\mathbb{R}^2$ its dual}. We use the standard multi-index notation $\partial_x^\alpha:=\partial_{x_1}^{\alpha_1}
\partial_{x_2}^{\alpha_2}$ for any $\alpha=(\alpha_1,\alpha_2)\in\mathbb{N}^2$, with $|\alpha|:=\alpha_1+\alpha_2$ and similar ones for the dual $\X^*\cong\mathbb{R}^2$. The phase space is denoted by $\Xi:=\X\times\X^*$ and is endowed with the canonical symplectic form.

\subsubsection{The pseudo-differential calculus.}

 For pseudo-differential operators on $L^2(\X)$ we shall use the following H\"{o}rmander type classes (see \cite{H-3}): for $s\in\mathbb{R}$ and $\rho\in[0,1]$,
\begin{equation*}
S^s_\rho(\X\times\X^*):=\Big\{F\in\ C^\infty_{\text{\sf pol,u}}(\Xi)\,,\,\underset{(x,\xi)\in\Xi}{\sup}<\xi>^{-s+\rho|\beta|}\big|\big(\partial^\alpha_x\partial^\beta_\xi\,F\big)(x,\xi)\big|\leq C_{\alpha,\beta},\,\forall(\alpha,\beta)\in\mathbb{N}^2\times\mathbb{N}^2\Big\},
\end{equation*}
$$
S^{-\infty}(\X\times\X^*):=\underset{s\in\mathbb{R}}{\bigcap}S^s_0(\X\times\X^*).
$$

All these vector spaces are endowed with  locally convex topologies defined by countable families of semi-norms that we shall denote generically by $\nu: S^s_\rho(\X\times\X^*)\rightarrow\mathbb{R}_+$. We shall also use the notation:
$$
S^\infty_\rho(\X\times\X^*):=\underset{s\in\mathbb{R}}{\bigcup}S^s_\rho(\X\times\X^*).
$$
A specific feature of the arguments which are developed  in this paper is the use of a calculus with $2\times2$-matrix valued pseudodifferential operators. Thus we shall consider H\"{o}rmander type symbols with values in the algebra of $2\times2$ complex matrices endowed with the matrix norm, that we denote by $\mathcal{M}_{2\times2}(\mathbb{C})$. 
We shall consider the following $2\times2$ complex matrix valued symbols:
\begin{align}\label{N-matrix-symbols}
	&S^s_\rho(\X\times\X^*)_{2\times2}:=\\
	&\hspace*{0.2cm}\Big\{F\in C^\infty_{\text{\sf pol,u}}\big(\Xi;\mathcal{M}_{2\times2}(\mathbb{C})\big)\,,\,\underset{(x,\xi)\in\Xi}{\sup}<\xi>^{-s+\rho|\beta|}\big\|\big(\partial^\alpha_x\partial^\beta_\xi\,F\big)(x,\xi)\big\|_{\mathcal{M}_{2\times2}(\mathbb{C})}\leq C_{\alpha,\beta},\,\forall(\alpha,\beta)\in\mathbb{N}^2\times\mathbb{N}^2\Big\},\nonumber 
\end{align}

$$
S^{-\infty}(\X\times\X^*)_{2\times2}:=
\underset{s\in\mathbb{R}}{\bigcap}S^s_0(\X\times\X^*)_{2\times2}.
$$

The topologies are defined by the same type of semi-norms in which the modulus of complex numbers has been replaced by the matrix norm (we shall work with $\ell^\infty$ type matrix norm $\|M\|_{\mathcal{M}_{2\times2}(\mathbb{C})}=\max\big\{\big|M_{jk}\big|,\,(j,k)\in\{1,2\}^2\big\}$). Due to the specificity of our problem, on $\mathcal{M}_{2\times2}(\mathbb{C})$ we shall  also work with matrix indices $(j,k)$ with $j=\pm$ and $k=\pm$. Formulas \eqref{D-PsiDO} and \eqref{D-MPsiDO} with matrix-valued symbols are well defined as linear operators on $\mathscr{S}\big(\X;\mathbb{C}^2\big)$ for any $\Phi\in S^s_\rho(\X\times\X^*)_{2\times2}$. An inspection of the results on the magnetic pseudodifferential calculus summarized in appendix B of \cite{CHP-1} shows that they remain true in our new setting (of matrix-valued symbols). 

We shall use the Weyl quantization of symbols (see \cite{H-3}):
\beq\label{D-PsiDO}
\big(\Op(\Phi)\phi\big)(x):=(2\pi)^{-2}\int_{\X}dy\int_{\X^*}d\xi\,
e^{i<\xi,(x-y)>}\Phi\big((x+y)/2,\xi\big)\phi(y),\quad\forall(\Phi,\phi)\in\mathscr{S}(\Xi)\times\mathscr{S}(\X).
\eeq
It is well known that the composition of linear operators induces a non-commutative product, \textit{the Moyal product} on test functions on $\Xi$ and can be extended   to a large class of tempered distributions (see \cite{G-BV-1, G-BV-2, 
F, GLS}). We shall denote it by $\Phi\sharp\Psi$, i.e.
$$
\Op(\Phi\,\sharp\, \Psi)\,=\,\Op(\Phi)\circ\Op(\Psi).
$$
With these definitions it is straightforward to notice that $H_\Gamma=\Op(h)$ for the symbol
\beq \label{R-HGammaOp}
h(x,\xi):=\big(\xi_1-A^\Gamma_1(x)\big)^2+\big(\xi_2-A^\Gamma_2(x)\big)^2+V^\Gamma(x)
\eeq
which  is an elliptic symbol of H\"{o}rmander class $S^2_1(\X\times\X^*)$ whose principal symbol is $|\xi|^2$.

Given some magnetic field $B$ with components of class $BC^\infty(\X;\mathbb{R})$ with an associated vector potential $A$ with components of class $C^\infty_{\text{\sf pol}}(\X;\mathbb{R})$,  we can also define a magnetic pseudodifferential calculus (\cite{MP-1}, \cite{IMP-1}, \cite{IMP-2}, \cite{AMP}):
\begin{align}\label{D-MPsiDO}
\big(\Op^A(\Phi)\phi\big)(x)&:=(2\pi)^{-2}\int_{\X}dy\int_{\X^*}d\xi\,
e^{i<\xi,(x-y)>}e^{-i\int_{[y,x]}A}\,\Phi\big((x+y)/2,\xi\big)\phi(y),\\
&\quad \forall(\Phi,\phi)\in\mathscr{S}(\Xi)\times\mathscr{S}(\X),\nonumber 
\end{align}
where the quantity $\int_{[y,x]}A$ is defined and commented on at the beginning of Section \ref{S-MLB}.  
Similarly we can define a \textit{'magnetic' Moyal product} $\Phi \,\sharp^B\, \Psi$ such that
\beq\label{DF-magn-MoyalPr}
\Op^A(\Phi\,\sharp^B\,\Psi)\,=\,\Op^A(\Phi)\circ\Op^A(\Psi).
\eeq
One can prove that it only depends on the magnetic field $B$ and not on the vector potential one has chosen. Its properties, rather similar with those of the usual Moyal product are studied in \cite{MP-1}, \cite{IMP-1} and \cite{IMP-2}.

Given any continuous linear operator $T:\mathscr{S}(\X)\rightarrow\mathscr{S}^\prime(\X)$ we denote by $\mathfrak{S}_T\in\mathscr{S}^\prime(\Xi)$ its Weyl symbol and by $\mathfrak{S}^A_T\in\mathscr{S}^\prime(\Xi)$ its magnetic symbol, i.e. the tempered distributions satisfying 
\beq\label{N-magn-symb}
\Op\big(\mathfrak{S}_T\big)=T,\quad\Op^A\big(\mathfrak{S}^A_T\big)=T;
\eeq
we denote by $\mathfrak{K}_T\in\mathscr{S}^\prime(\X\times\X)$ its distribution kernel given by the Schwartz Kernels Theorem. Moreover given any distribution kernel $\mathfrak{K}\in\mathscr{S}^\prime(\X\times\X)$ we denote by $ \Int\,\mathfrak{K}\in\mathbb{L}\big(\mathscr{S}(\X);\mathscr{S}^\prime(\X^*)\big)$ the integral operator defined by
\beq\label{DF-int-op}
\big(\Int\,\mathfrak{K}(\phi)\big)(\psi)\,:=\,\mathfrak{K}\big(\phi\otimes\psi\big),\quad\forall(\phi,\psi)\in\mathscr{S}(\X)\times\mathscr{S}(\X).
\eeq

\subsubsection{The Bloch-Floquet representation.}
\label{SSS-BFZ-repr}

Let us make a more detailed analysis of the image of the Bloch-Floquet-Zak transformation  (\cite{Z} or \cite{Ku} for a general discussion of the Bloch-Floquet theory). The formulas \eqref{F-defF-1} and \eqref{F-defF-2} imply that we may describe the image $\check{\mathcal{U}}_\Gamma\big[L^2(\X)\big]$ as the space
\beq\label{DF-spG}
\mathscr{G}:=\big\{F\in L^2_{\text{\sf loc}}\big(\X^*;L^2(\mathbb{T})\big)\,,\,F(\xi+\gamma^*)=e^{-i<\gamma^*,x>}F(\xi),\ \forall\gamma^*\in\Gamma_*\big\}.
\eeq

Defining $\mathscr{V}(x,\xi):=e^{i<\xi,x>}$ as a function of class $C^\infty_{\text{\sf pol}}(\Xi)$ and embedding $\mathscr{G}$ in $L^2_{\text{\sf loc}}(\Xi)$ we can consider its image under the operator of multiplication with the function $\mathscr{V}$ that will take us in a space of $\Gamma_*$-periodic functions
$$
\mathscr{V}\mathscr{G}=\big\{F\in L^2_{\text{\sf loc}}\big(\Xi\big)\,,\,F(x,\xi+\gamma^*)=F(x,\xi),\ \forall\gamma^*\in\Gamma_*;\ F(x+\gamma,\xi)=e^{i<\xi,\gamma>}F(x,\xi),\,\forall\gamma\in\Gamma\big\}.
$$

We would like to see this space of $\Gamma_*$-periodic functions as a space of functions defined on $\mathbb{T}_*$. For that let us define for any $\theta\in\mathbb{T}_*$ the space of functions
$$
\mathscr{F}_\theta:=\big\{F\in L^2_{\text{\sf loc}}(\X)\,,\,F(x+\gamma)=e^{i<\theta,\gamma>} F(x),\,\forall\gamma\in\Gamma\big\}
$$
endowed with the quadratic norm
$$
\big\|F\big\| ^2_{\mathscr{F}_\theta}:=\int_{\mathcal{E}}dx\,\big|F(x)\big|^2.
$$
We notice that embedding $L^2(\mathbb{T})$ in $L^2_{\text{\sf loc}}(\X)$ (as periodic functions) we can write that 
$$ \mathscr{F}_\theta=\mathscr{V}(\cdot.\theta)L^2(\mathbb{T})\equiv\mathcal{V}_\theta L^2(\mathbb{T}).
$$
This allows us to use the notion of measurable field of Hilbert spaces as developed in §II.1.5 in \cite{Di} and notice that our space $\mathscr{V}\mathscr{G}$ is in fact \textit{the direct integral of Hilbert spaces} 
\beq\label{DF-BF-repr}
\mathscr{V}\mathscr{G}\cong\mathscr{F}\,:=\,\int^\oplus_{\mathbb{T}_*}d\theta\,\mathscr{F}_\theta
\eeq
defined by the measurable vector fields $ \big\{\mathcal{V}_\theta u\mid\theta\in\mathbb{T}_*,\,u\in L^2(\mathbb{T})\big\}$.  Using the already introduced formula \eqref{b-f-z}, we denote by $\mathcal{U}_{\Gamma}:=\mathscr{V}\check{\mathcal{U}}_\Gamma:L^2(\X)\overset{\sim}{\rightarrow}\mathscr{F}$ and notice that it has the following explicit form:
\begin{equation}\label{UGamma}
\big[\mathcal{U}_{\Gamma}\phi\big]
(x,\theta):= (2\pi)^{-1}\sum\limits_{\gamma\in\Gamma}e^{i<\theta,\gamma>}\phi(x-\gamma),\qquad\forall\phi\in\mathscr{S}(\X).
\end{equation}
Then $\mathcal{U}_{\Gamma}$ defines a unitary
operator 
$L^2\big(\X\big)\rightarrow\mathscr{F}$ that we call \textit{the Bloch-Floquet representation}.

 An essential tool in working with direct integrals of Hilbert spaces are the \textit{measurable fields of operators}, the \textit{decomposable operators} (see Definition 1 and 2 in Ch. II.2.1 in \cite{Di}) and the \textit{measurable fields of von Neumann algebras} (see Ch. II.3.2 in \cite{Di}). 
In fact, in the Bloch-Floquet representation the periodic Hamiltonian $H_\Gamma$ is \textit{decomposable}, being defined by a measurable field of self-adjoint operators $\widehat{H}(\theta)$ acting in each space $\mathscr{F}_\theta$ and we may write:
\beq \label{hc1}
\mathscr{U}_\Gamma H_\Gamma\mathscr{U}_\Gamma^{-1}\,=\,\int_{\mathcal{B}}^\oplus d\theta\,\widehat{H}(\theta), \quad \widehat{H}(\theta)=\underset{k\in\mathbb{N}}{\sum}\lambda_k(\theta)\, \widehat{\pi}_k(\theta).
\eeq

In order to define the eigenprojections $\widehat{\pi}_k$ when the multiplicity is bigger then 1, we apply the following procedure. To each $k\in \mathbb N$ and $\theta \in \mathbb T_*$, we introduce the minimal labelling of $\lambda_k(\theta)$:  $$\nu (k,\theta)=\inf_{\lambda_j(\theta) =\lambda_k(\theta)}j\,.$$
We define the eigenprojections by the following Dunford contour integral (Chapter VII in \cite{DS-1}):
\beq\label{DF-pi-k}
\begin{split}
	\widehat{\pi}_k(\theta):=&\left\{
	\begin{array}{ll}
		=\frac{1}{2\pi i}\oint_{\mathscr{C}_k(\theta)}d\z\,\big(\widehat{H}(\theta)-\z\bb1\big)^{-1}\in  \mathbb{LP}(\mathscr{F}_
		\theta) & \mbox{ if } \nu(k,\theta)=k\,,\\
		=0 & \mbox{ if }  \nu(k,\theta) < k\,,
	\end{array}
	\right.
\end{split}
\eeq 
where $\mathscr{C}_k(\theta)\subset\mathbb{C}$ is a circle surrounding $\lambda_k(\theta)$ and no other point from $\sigma\big(\widehat{H}(\theta)\big)$. They define measurable functions on $\mathbb{T}_*$. Moreover, by elliptic regularity  and noticing that $\mathcal{V}_\theta\big[C^\infty(\mathbb{T})\big]=\mathscr{F}_\theta\cap C^\infty(\mathbb{R}^2) =:\mathscr{F}_\theta^\infty$, we deduce that the finite number of eigenfunctions $\varphi_k(\theta)$ associated to any Bloch eigenvalue $\lambda_k(\theta)$ are smooth functions and thus the projection valued sections $\widehat{\pi}_k(\theta)$ project on the subspace of smooth functions in $\mathscr{F}_\theta$ for any $\theta\in\mathbb{T}_*$, i.e. 
\begin{equation}\label{aprilie6}
\widehat{\pi}_k(\theta)\mathscr{F}_\theta\subset\mathscr{F}_\theta^\infty,\quad\forall\theta\in\mathbb{T}_*.
\end{equation}

Let us notice that if a group of eigenvalues remains isolated from the rest of the spectrum while $\theta$ varies in some open set, then the total Riesz projection associated with this group is locally smooth in $\theta$ on that open set. 

We denote by $P_K(\widehat{H}(\theta))$ the spectral projection of $\widehat{H}(\theta)$ corresponding to a set $K\subset \R$. Due to the unboundedness of the operators $\widehat{H}(\theta)$ we shall also use the direct integral decomposition of the domain of $H_\Gamma$ considered as Hilbert space for the graph-norm:
$$
\mathcal{U}_\Gamma\mathscr{H}^2(\X)=\int^\oplus_{\mathbb{T}_*}d\theta\,\mathscr{F}^2_\theta,\quad\mathscr{F}^2_\theta:=\mathcal{V}_{\theta}\mathscr{H}^2(\mathbb{T})
$$
with $\mathscr{H}^2(\mathbb{T})$ the Sobolev space of order 2 on the 2-dimensional torus.

In dealing with direct integrals of Hilbert spaces over the dual torus $\mathbb{T}_*$ and decomposable operators, we shall usually call \textit{global sections} the measurable fields of vectors: 
$$
f:\mathbb{T}_*\rightarrow\int^\oplus_{\mathbb{T}_*}d\theta\,\mathscr{F}_\theta,\quad\,f(\theta)\in\mathscr{F}_\theta,
$$
or of operators 
$$
T:\mathbb{T}_*\rightarrow\int^\oplus_{\mathbb{T}_*}d\theta\,\big[\mathbb{B}(\mathscr{F}_\theta)\big],\quad\,T(\theta)\in\mathbb{B}(\mathscr{F}_\theta).
$$
One can also consider local sections defined only for some open subsets $\mathcal{O}\subset\mathbb{T}_*$. The topology of the torus $\mathbb{T}_*$ has as consequence the non-triviality of the extension problem of a regular (continuous or smooth) local section to a regular global section and we shall have to deal with this problem in our analysis. Let us briefly mention here that one may approach these problems in the framework of sections in vector bundles but we shall avoid these aspects in our work in order to make a rather self-contained presentation; nevertheless, we shall frequently use the term \textit{fiber} for the Hilbert spaces or the spaces of operators in the direct integral.

\subsection{The main steps of the proof of Theorem \ref{T-Main-2}.}
\label{SS-ProofStr}

Let us present here  the main steps of the proof of Theorem \ref{T-Main-2}. 
We start by mentioning that even though the eigenprojections associated to $\{\lambda_\pm(\theta)\}_{\theta\in\Sigma_I}\}$ by \eqref{DF-pi-k} have a singularity in $\theta_0$, their orthogonal sum (denoted by $\widehat{\Pi}_I(\theta)$) is smooth on $\Sigma_I$ due to Hypothesis \ref{H-1}. In \eqref{eq:1.19} we define the rank 2 \enquote{window} Hamiltonian $\widehat{H}_I(\theta)$ associated to the the spectrum $\{\lambda_\pm(\theta)\}_{\theta\in\Sigma_I}\}$.
\begin{description}
	\item[{\bf Step 1.}] Choosing a local basis in the Bloch-Floquet representation we define the unitary-valued map $\Sigma_I\ni\theta\rightarrow\widehat{\Upsilon}(\theta)\in\mathbb{U}
	\big(\widehat{\Pi}_I(\theta) \mathscr{F}^\infty_\theta;\mathbb{C}^2\big)$ and we decompose the corresponding $2\times2$ Hermitian matrix $M_I(\theta):=\widehat{\Upsilon}(\theta)\widehat{H}_I
	(\theta)\widehat{\Upsilon}
	(\theta)^{-1}$ with respect to the basis of Pauli matrices in $\mathcal{M}_{2\times2}(\mathbb{C})$ \eqref{F-HIcirc}. This is the object of Subsection \ref{SS-LH}.
	\item[{\bf Step 2.}] Using a perturbation of the matrix $M_I$ we introduce a spectral gap around $\theta_0$ and replace the two eigenprojections $\widehat{\pi}_\pm(\theta)$ for $\theta\in\Sigma_I$ by a pair of smooth orthogonal projections in $\widehat{\Pi}_I(\theta)\mathscr{F}_\theta$ for any $\theta\in\Sigma_I$. Then, using a rather explicit construction explained in  Appendix \ref{A-BBdl}, we extend these local projection-valued sections to $\mathbb{~T}_*$ and define a pair of smooth global orthonormal sections $\{\widehat \Phi_\pm(\theta)\}_{\theta\in\mathbb{T}_*}$ and a neighborhood $U\subset\Sigma_I$ of $\theta_0$, such that (see Proposition \ref{P-smooth-glob-frame})
	\begin{itemize}
		\item for $\theta\in U$ we have that
		\beq\label{F-Psi-PiI}
		\mathcal{L}_{\mathbb{C}}\{\widehat{\Phi}_-
		(\theta)\,,\,\widehat{\Phi}_+
		(\theta)\}\,=\,\widehat{\Pi}_I(\theta)\mathscr{F}_\theta,
		\eeq
		\item for $\theta\notin U$ we have that:
		\begin{align} \label{F-Psi-delta-}
		&\widehat{\Phi}_-(\theta)\in P_{(-\infty,0]}(\widehat{H}(\theta))\mathscr{F}_\theta 
		\\
		\label{F-Psi-delta+}
		&\widehat{\Phi}_+(\theta)\in P_{(0,\infty)}(\widehat{H}(\theta))\mathscr{F}^2_\theta.
		\end{align} 
	\end{itemize}
	This smooth global basis allows us to define a \textit{\enquote{local band} projector} $\{\widehat{P}_I(\theta)\}_{\theta\in\mathbb{T}_*}$ (see Definition \ref{D-Pdelta}), which is not a spectral projector for  $\widehat{H}(\theta)$ on the whole Brillouin zone, but coincides with such a spectral projector on the neighborhood $U$ of $\theta_0$. This procedure will  firstly allow us to circumvent the possible non-existence of a localized Wannier basis for a spectral projection of $H_\Gamma$ which includes the spectral window $I$ near the crossing, and secondly to reduce our analysis to a possibly smaller spectral window inside $I\subset\mathbb{R}$. These constructions are presented in Subsection \ref{SS-LB-frame} and they depend on the choice of a parameter $\delta>0$ controlling the width of the spectral gap created around $\theta_0$. In fact we shall fix the value of this parameter depending on the properties of the periodic Hamiltonian $H_\Gamma$ and develop all our arguments keeping this value fixed. Although the conclusion of our Theorem does not depend on this fixed value, the accuracy of the estimations (the constants appearing in the estimations) may in fact depend on the exact value we fixed for $\delta>0$.  We note that  \cite{M-P} contains a detailed analysis of what can happen near a \enquote{closely avoided} conical crossing and its topological consequences. 
	
	\item[{\bf Step 3.}] A general procedure having its roots in (\cite{HS}, \cite{Ne-RMP}) and developed in \cite{CHN} and \cite{CIP} allows us  to define some \textit{magnetic \enquote{local band} projectors} $P_I^\epsilon$ associated with the constant magnetic field $\epsilon B^\circ$. While these operators are not expected to be norm-continuous with respect to $\epsilon\in[0,\epsilon_0]$, their \textit{symbols as magnetic pseudodifferential operators} are much better behaved (see Proposition \ref{F-est-dif-bdproj}). We can define now a \textit{magnetic \enquote{local band} Hamiltonian}: $H^\epsilon_I:=P_I^\epsilon H^\epsilon_\Gamma P_I^\epsilon$. We present the details of this construction in Section \ref{S-MLB}.
		\item[{\bf Step 4.}] In order to compare the magnetic periodic Hamiltonian $H^{\epsilon}_\Gamma$ with the \textit{magnetic \enquote{local band} Hamiltonian} $H^\epsilon_I$ we have to modify our version of the Feshbach-Schur procedure elaborated in \cite{CHP-2} in order to apply it for energies in a spectral gap of the studied operator. The abstract procedure is explained in Appendix \ref{A-FS-arg} and we prove its applicability to our given situation in Proposition \ref{P-Hyp-H-magn}. This allows us to obtain Corollary \ref{C-magn-FS-est} estimating the Hausdorff distance between the spectra of the magnetic periodic Hamiltonian $H^{\epsilon}_\Gamma$ and a \enquote{dressed} modified version $\widetilde{H}^\epsilon_I$ of the magnetic \enquote{quasi band} Hamiltonian $H^\epsilon_I$ in an interval $ (-L\sqrt{\epsilon},L\sqrt{\epsilon})$.
		
		\item[{\bf Step 5.}] For the \textit{\enquote{dressed} modified magnetic \enquote{quasi band} Hamiltonian} $\widetilde{H}^\epsilon_I$ we apply the procedure of \cite{CHN} as developed in \cite{CHP-2} and using \textit{magnetic matrices} we prove its unitary equivalence with an effective Hamiltonian of the form $\Op^{\epsilon}(\mathcal  k^\epsilon_I)$, with $\Gamma_*$-periodic symbol $\mathcal  k^\epsilon_I\in C^\infty_{\text{\sf pol,u}}\big(\X^*;\mathcal{M}_{2\times2}(\mathbb{C})\big)$ (see Proposition \ref{P-magn-matrix-ke}). We also notice that this periodic symbol is close of order $\epsilon$ (with respect to the semi-norms defining the topology on the space of $S^0$ symbols)  to a periodic symbol $\mathcal{k}_I$ that no longer depends on the magnetic field (Proposition \ref{P-kepsilon-k}). This non-magnetic symbol coincides near the crossing point with our initial matrix $M_I$ \eqref{F-HIcirc}. We can now define the \textit{Peierls-Onsager effective Hamiltonian} $\Op^{\epsilon}(\mathcal  k_I)$ for our spectral region $I\subset\mathbb{R}$.
		
		\item[{\bf Step 6.}] Once we obtained this Peierls-Onsager effective Hamiltonian, we could repeat the procedure developed by us in \cite{CHP-1} and use the magnetic pseudodifferential calculus to construct a \enquote{quasi-resolvent}. Nevertheless, as in our situation we can restrict to constant magnetic fields due to Remark \ref{P-CP1}, we prefer to use some existing results obtained for this situation in \cite{HS1} and \cite{HS2}. The idea is that for constant magnetic fields, our 2-dimensional problem is in fact isospectral with an $1$-dimensional operator, whose spectrum can be completely understood in the framework of the semi-classical analysis that has been developed in the cited references. The intensity of the constant magnetic field $\epsilon B^\circ$ plays the role of the \textit{semi-classical \enquote{small parameter}} being controlled by $\epsilon\in[0,\epsilon_0]$. We present these ideas and details in Subsection~\ref{ss7.1new}.

		\item[{\bf Step 7.}] The semi-classical problem involves a $\Gamma_*$-periodic $2\times2$-matrix valued symbol in one dimension, with values Hermitian matrices having two real eigenvalues that remain well separated on the entire 2-dimensional phase space with the exception of some small neighborhoods of the points of the lattice $\Gamma_*$. We are interested in locating the spectrum in a small interval around $0\in\mathbb{R}$ for its Weyl quantization acting on $L^2(\mathbb{R})$. Thus we are interested only in the behavior near the points  of the lattice $\Gamma_*$. Lemma 2.1 in \cite{HS1} allows us to decompose our symbol as a superposition of $2\times2$-matrix valued symbols having only one \enquote{crossing point} for their values at one of the vertices of the lattice $\Gamma_*$. For these symbols we approximate their spectrum close to $0$ by the spectrum of the linear term in their Taylor expansion (see Subsection \ref{SS-1-p-scls}).
		
		\item[{\bf Step 8.}] Finally, in Section \ref{S-final} we put together the results that lead us to the conclusion of Theorem~\ref{T-Main-2} with $\kappa=0$.
\end{description}

\section{The \enquote{quasi-band} associated to the spectral window {\it I}.}\label{S-qB-frame}

This section is devoted to the first two steps of our strategy for proving Theorem \ref{T-Main-2}. We prove the existence of two smooth global sections $ \hat{\Psi}_\pm:\mathbb{T}_*\rightarrow\mathscr{F}$ satisfying (\ref{F-Psi-PiI}) -( \ref{F-Psi-delta+}). This allows us to define the \enquote{quasi-band} associated with the spectral window $I\subset\mathbb{R}$ (Definition \ref{D-Pdelta}). Mimicking the Wannier basis description of the Bloch spectral bands, it is defined as the linear subspace generated by the inverse Floquet transformed sections $\hat{\Psi}_\pm$ translated by vectors in $\Gamma$. 
In the whole section, we assume Hypotheses \ref{H-1} and \ref{H-2}. We shall use the notation: 
\beq\label{FD-par-sp}
E_0:=-\inf\,\sigma(H_\Gamma)>\Lambda_->0.
\eeq

\subsection{The \enquote{local Hamiltonian} associated to the spectral window {\it I}}\label{SS-LH}

In this subsection we realize the first step in our proof strategy (see Subsection \ref{SS-ProofStr})
by defining and analyzing the $2\times2$ Hermitian matrix $M_I(\theta)$ in \eqref{F-HIcirc}. This matrix associated to our periodic Hamiltonian $H_\Gamma$ in the Bloch-Floquet representation, for the spectral interval $I\subset\mathbb{R}$ and for a neighborhood of the crossing-point $\theta_0$ plays a crucial role in the spectral analysis we are doing as we shall see in Section \ref{S-red-ModelH}.

For $\theta\in\Sigma_I\setminus\{\theta_0\}$, we recall the notation in \eqref{DF-lambda-pm} and the definition in \eqref{DF-pi-k} and introduce
\beq\label{DF-pi-pm}
\widehat{\pi}_-(\theta):=\widehat{\pi}_{k_0}(\theta),\quad\widehat{\pi}_+(\theta):=\widehat{\pi}_{k_0+1}(\theta).
\eeq
With these notations we shall consider
the following family of bounded self-adjoint operators acting on $\mathscr{F}_\theta$:
\begin{equation} \label{eq:1.19}
\widehat{H}_I(\theta)\,:=
\left\{\begin{array}{ll}
\lambda_{-}(\theta)\,\widehat{\pi}_{-}(\theta)\,+\,\lambda_{+}(\theta)\, \widehat{\pi}_{+}(\theta) &\mbox{ for }\theta\in\Sigma_I\setminus\{\theta_0\}\,,\\
0  & \mbox{ for }\theta=\theta_0\,.
\end{array}\right.
\end{equation}
We note that the  $\widehat{\pi}_\pm(\theta)$ have rank one for $\theta\in\Sigma_I\setminus\{\theta_0\}$ and that 
the operator 
\beq\label{D-PiI}
\widehat{\Pi}_I(\theta):=
\left\{\begin{array}{l}
	\widehat{\pi}_{-}(\theta)+\widehat{\pi}_{+}(\theta),\quad\forall\theta\in\Sigma_I\setminus\{\theta_0\}\,,\\
	\widehat{\pi}_{k_0}(\theta_0)\quad\text{ for }\theta=\theta_0\,,
\end{array}\right.
\eeq 
defines a family of rank two  orthogonal projections which, unlike the rank one families $\widehat{\pi}_\pm(\theta)$, form a local smooth section $\widehat{\Pi}_I:\Sigma_I\rightarrow\int^\oplus_{\Sigma_I}d\theta\,\big[\mathbb{B}(\mathscr{F}_\theta)\big]$. If we choose the spectral interval $I$ small enough, the smoothness of the local section $\widehat{\Pi}_I(\theta)$ implies that 
\beq\label{F-theta-var}
\big\|\mathcal{V}_\theta^{-1}\widehat{\Pi}_I(\theta)\mathcal{V}_\theta-\mathcal{V}_{\theta_0}^{-1}\widehat{\Pi}_I(\theta_0)\mathcal{V}_{\theta_0}\big\|_{\mathbb{B}(L^2(\mathbb{T}))}\,\leq\,1/2,\quad\forall\theta\in\Sigma_I.
\eeq
We can thus use the Sz.-Nagy  unitary intertwining operator (see I.4.6 in \cite{Ka}) for the pair of 2-dimensional orthogonal projections $\check{\Pi}_I(\theta):=\mathcal{V}_\theta^{-1}\widehat{\Pi}_I(\theta)\mathcal{V}_\theta$ and $\check{\Pi}_I(\theta_0):=\mathcal{V}_{\theta_0}^{-1}\widehat{\Pi}_I(\theta_0)\mathcal{V}_{\theta_0}$ acting in $L^2(\mathbb{T})$:
\beq\label{F-Nagy-unit}
\widetilde{U}_0(\theta):=\big(\bb1-(\check{\Pi}_I(\theta)-
\check{\Pi}_I(\theta_0))^2\big)^{-1/2}\big(\check{\Pi}_I(\theta)
\check{\Pi}_I(\theta_0)\,+\,\check{\Pi}_I(\theta)^\bot\check{\Pi}_I(\theta_0)^\bot\big)
\eeq
which belongs to 
$C^\infty\big(\Sigma_I;\mathbb{U}(L^2(\mathbb{T}))\big)$ and satisfies the intertwining equation:
$$
\widetilde{U}_0(\theta_0)=\bb1,\quad \widetilde{U}_0(\theta)\check{\Pi}_I(\theta_0)=\check{\Pi}_I
(\theta)\widetilde{U}_0(\theta),
\quad\forall\theta\in\Sigma_I.
$$
We conclude that it exists a smooth diffeomorphism $\widehat{\Upsilon}:\int^\oplus_{\Sigma_I}d\theta\,\big[\widehat{\Pi}_I\mathscr{F}_\theta\big]\overset{\sim}{\longrightarrow}\mathbb{C}^2\times\Sigma_I$ that is unitary on each fiber, i.e. 
\beq\label{F-Upsilon}
\widehat{\Upsilon}(\theta):\ \widehat{\Pi}_I(\theta)\mathscr{F}_\theta\,\overset{\sim}{\rightarrow}\, \mathbb{C}^2\times
\{\theta\}
\eeq 
is a smooth family of unitary operators between each of the 2-dimensional Hermitian complex spaces $\widehat{\Pi}_I(\theta)\mathscr{F}_\theta$ and $\mathbb{C}^2$ for $\theta\in\Sigma_I$.
\begin{remark}\label{R-contr-triv}
In fact, using some standard results in topology one may give up the restriction in \eqref{F-theta-var} and just ask for $\Sigma_I$ to be contractible.
\end{remark}

\begin{definition}\label{D-onb-SigmaI}
	Let $\{\varepsilon_-,\varepsilon_+\}$ be the canonical orthonormal basis in $\mathbb{C}^2$ and for $\theta\in\Sigma_I$ and $j=\pm$ let us define 
	$$
	\hat \varphi_j(\theta):=\widehat{\Upsilon}(\theta)^{-1}(\varepsilon_j)\in\widehat{\Pi}_I(\theta)\mathfrak{F}_
	\theta\subset\mathfrak{F}^2_\theta.
	$$
\end{definition}
\begin{remark}\label{R-D-H-local}
The pair $\{\hat \varphi_-(\theta), \hat \varphi_+(\theta)\}$ forms a smooth orthonormal basis of $\widehat{\Pi}_I(\theta)\mathscr{F}_\theta\subset
\mathscr{F}^2_\theta$ over $\Sigma_I\subset\mathbb{T}_*$ and we obtain a smooth family of $2\times 2$ matrices
\beq\label{FD-Psi}
\Sigma_I\ni\theta\mapsto M_I
(\theta)_{jk}:=\left\langle \hat \varphi_j(\theta)\,,\,\widehat{H}_I(\theta) \hat \varphi_k(\theta)\right\rangle_{\mathscr{F}_\theta}=
\left\langle\varepsilon_j\,,\,\widehat{\Upsilon}(\theta)\widehat{H}_I(\theta)\widehat{\Upsilon}(\theta)^{-1}\varepsilon_k\right\rangle_{\mathbb{C}^2}
\in\mathcal{M}_{2\times2}(\mathbb{C})
\eeq
defining \textit{the \enquote{local Hamiltonian}} $\mathring{H}_I:=\widehat{\Upsilon}(\theta)\widehat{H}_I(\theta)\widehat{\Upsilon}(\theta)^{-1}$ associated to the spectral interval $I$, in the fixed basis.
\end{remark}

Once we have fixed the canonical orthonormal basis $\{\varepsilon_1,\varepsilon_2\}$ on $\mathbb{C}^2$ we shall work with the orthogonal basis of the real algebra of Hermitian matrices on $\mathbb{C}^2$ given by the Pauli matrices $\{\sigma_1,\sigma_2,\sigma_3\}$ and the identity $\bb1_2$. Then our local Hamiltonian $\mathring{H}_I$ is described by four functions $F_0\in C^\infty\big(\Sigma_I;\mathbb{R}\big)$ and $ \mathbf{F}\equiv(F_1,F_2,F_3)\in C^\infty\big(\Sigma_I;\mathbb{R}^3\big)$ such that
\beq\label{F-HIcirc}
M_I(\theta)\,=\,F_0(\theta)\,\bb1\,+\,\underset{\ell =1,2,3}{\sum}F_\ell (\theta)\, \sigma_\ell
\ \in\mathcal{M}_{2\times2}(\mathbb{C})\,.
\eeq
Its eigenvalues are given by: 
\begin{equation}\label{F-lambdapm}
\lambda_\pm(\theta)=F_0(\theta)\pm|\mathbf{F}(\theta)| ,\ \forall  \theta\in \Sigma_I.
\end{equation}
If we introduce the notations
\begin{align*}\begin{split}
&\lambda_\circ (\theta) := (\lambda_+(\theta) + \lambda_-(\theta))/2=(1/2)\Tr\big(M_I(\theta)\big)\,,\quad \delta(\theta) :=( \lambda_+(\theta) -\lambda_-(\theta))/2\,,\\
&\mathcal{d}(\theta):=\lambda_-(\theta)\lambda_+(\theta)=\det\big(M_I(\theta)\big),
\end{split}\end{align*}
we notice that $F_0(\theta)=\lambda_\circ(\theta)$ and $\big|\mathbf{F}(\theta)\big|=\delta(\theta)$ for any $\theta\in\Sigma_I$ and thus are invariant to a change of basis in $\mathcal{M}_{2\times2}(\mathbb{C})$.
Let us recall that $\lambda_-(\theta_0)=\lambda_+(\theta_0)=\lambda_\circ(\theta_0)=\delta(\theta_0)=\mathcal{d}(\theta_0)=0$ and 
let us also define
\beq\label{F-desc-F-v}
\left\{\begin{array}{l}
	v^{(1)}:=\big(\partial_{\theta_1}\mathbf{F}\big)(\theta_0)\,\in\mathbb{R}^3,\quad v^{(2)}:=\big(\partial_{\theta_2}\mathbf{F}\big)(\theta_0)\,\in\mathbb{R}^3, \\
	f_1:=\big(\partial_{\theta_1}F_0\big)(\theta_0)\,\in\mathbb{R}\,,\quad
	f_2:=\big(\partial_{\theta_2}F_0\big)(\theta_0)\,\in\mathbb{R}\,
\end{array}\right.
\eeq
so that we can write the Taylor expansions:
\begin{align} \label{F-v12}
\mathbf{F}(\theta)&=(\theta-\theta_0)_1v^{(1)}+(\theta-\theta_0)_2v^{(2)}\,+\,\mathscr{O}(|\theta-\theta_0|^2)\,,\\
\label{F-descF0}
F_0(\theta)&=(\theta-\theta_0)_1f_1+(\theta-\theta_0)_2f_2+
\mathscr{O}(|\theta-\theta_0|^2)\,.
\end{align}

\begin{remark}\label{R-F-dep}
	Although the exact form of the functions $F_0(\theta)$ and $\big|\mathbf{F}(\theta)\big|$ strongly depends on the family of unitaries $\big\{\widehat{\Upsilon}(\theta)\big\}_{\theta\in\Sigma_{I}}$, the elements $f_1$, $f_2$, $|v^{(1)}|$, $|v^{(2)}|$ and $v^{(1)}\cdot v^{(2)}=\mu\,|v^{(1)}|\,|v^{(2)}|$, (where $\mu\in(-1,1)$ is the cosine of the angle between them in $\mathbb{R}^3$), are completely determined by the behavior of the eigenvalues $\lambda_\pm$ at $\theta_0$.
\end{remark}

Let us analyze the consequences of Hypotheses \ref{H-1} and \ref{H-2} concerning the functions $\{F_0,\mathbf{F}\}$.
\begin{remark}\label{R-d-smooth}
 It turns out that although the eigenvalues $\lambda_\pm(\theta)$ are not separately smooth, their product is smooth due to the identity
$$ 2\lambda_-(\theta)\lambda_+(\theta)=2\det \big(M_I(\theta)\big)=\left ({\rm Tr}(M_I(\theta))\right )^2 -{\rm Tr}(M_I(\theta)^2).$$
\end{remark}

Let us recall that given a function $f\in C^2(\mathring{\Sigma}_I;\mathbb{R})$, its Hessian is defined as the $2\times2$ matrix-valued continuous function
$$
(\Hes\,f)(\theta):=\left(\begin{array}{cc}
	\big(\partial_{\theta_1}\partial_{\theta_1}f\big)(\theta)&\big(\partial_{\theta_1}\partial_{\theta_2}f\big)(\theta)\\\big(\partial_{\theta_2}\partial_{\theta_1}f\big)(\theta)&\big(\partial_{\theta_2}\partial_{\theta_2}f\big)(\theta)
\end{array}\right).
$$

\begin{proposition}\label{P-cons-Hyp}
Hypotheses \ref{H-1} and \ref{H-2} imply that:
\begin{enumerate}
\item The map  $\Sigma_I\ni\theta\mapsto\mathcal{d}(\theta)\in\mathbb{R}$ is smooth and Hypothesis \ref{H-2} is equivalent to 
$$
\mathcal{d}(\theta_0)=0,\quad\big(\nabla\mathcal{d}\big)(\theta_0)=0,\quad\exists a_0>0,\ \left\langle\zeta\,,\,(\Hes\,\mathcal{d})(\theta_0)\zeta\right\rangle_{\mathbb{R}^2}\leq(-a_0)\|\zeta\|_{\mathbb{R}^2}^2,\,\forall\zeta\in\mathbb{R}^2.
$$
\item There exist $C>c>0$ such that 
	\begin{equation}\label{aprilie2} 
	c\;\vert \theta-\theta_0\vert\leq  \lambda_+(\theta)-\lambda_-(\theta)\leq C\;\vert \theta-\theta_0\vert,\quad\forall\theta\in\Sigma_I.
	\end{equation} 
\item The map $\Sigma_I\ni\theta \mapsto \delta(\theta)^2\in\mathbb{R}$ is smooth, has a unique non-degenerate zero at $\theta_0$, and its Hessian is given by
	\begin{align}
	&\big(\Hes\,\delta^2)(\theta_0)\big)=2\left(\begin{array}{cc}
	\big|\big(\partial_{\theta_1}\mathbf{F}\big)(\theta_0)\big|^2&\big(\partial_{\theta_1}\mathbf{F}\big)(\theta_0)\cdot\big(\partial_{\theta_2}\mathbf{F}\big)(\theta_0)\\ \label{F-non-deg}
	\big(\partial_{\theta_1}\mathbf{F}\big)(\theta_0)\cdot\big(\partial_{\theta_2}\mathbf{F}\big)(\theta_0)&\big|\big(\partial_{\theta_2}\mathbf{F}\big)(\theta_0)\big|^2
	\end{array}\right).
	\end{align}
	It uniquely determines $|v^{(1)}|,\,|v^{(2)}|$ and $v^{(1)}\cdot v^{(2)}$, (thus also the cosine of the angle between them in $\mathbb{R}^3$) and being non-degenerate one has
	\begin{equation}\label{hypv1v2}
	v^{(1)}\wedge v^{(2)}\neq 0\,.
	\end{equation}
Moreover there exists $\rho >0$ such that
\begin{equation}\label{F-est-Hess-delta2}
	\langle (\Hes\,\delta^2)(\theta_0) \zeta \,,\, \zeta\rangle \geq \rho |\zeta|^2\,,\, \forall \zeta \in \mathbb R^2\,.
\end{equation}
	\item There exists $0\leq  a<1$ such that 
    \begin{equation} \label{F-lambdacirc-delta}
	\big|\big\langle \zeta ,   \nabla \lambda_\circ (\theta_0)\big\rangle_{\mathbb{R}^2}\big|^2  \leq a^2 \langle (\Hes\,\delta^2) (\theta_0) \zeta\,,\, \zeta \rangle\,,\, \forall \zeta \in \mathbb R^2\,.
	\end{equation}
	\item There exists an open neighborhood $V_0$ of $\theta_0$, included in $\Sigma_I$ and some $\tilde{a}\in(a,1)$ such that 
	$$
	\big|F_0(\theta)\big|\leq \tilde a \,\big|\mathbf{F}(\theta)\big|
	\,,\,  \forall\theta\in V_0.
	$$
\end{enumerate}
\end{proposition}
\begin{proof} We recall that $F_0(\theta_0)=\lambda_\circ(\theta_0)=0=\delta(\theta_0)=|\mathbf{F}(\theta_0)|$ and the fact that the maps $\Sigma_I\ni\theta\mapsto F_0(\theta)\in\mathbb{R}$ and $\Sigma_I\ni\theta\mapsto \mathbf{F}(\theta)\in\mathbb{R}^3$ are of class $C^\infty$ so that the maps $\Sigma_I\ni\theta\mapsto \lambda_\circ(\theta)\in\mathbb{R}$ and $\Sigma_I\ni\theta\mapsto \delta^2(\theta)=\underset{1\leq j\leq3}{\sum}\mathbf{F}_j^2(\theta)\in\mathbb{R}_+$ are also of class $C^\infty$.
\begin{enumerate}
\item The first point follows from the definitions and Remark \ref{R-d-smooth}.   
\item In the second point the left inequality follows from Hypothesis \ref{H-2} noticing that  $\lambda_+(\theta)-\lambda_-(\theta)\geq\sqrt{-2\lambda_+(\theta)\lambda_-(\theta)}=\sqrt{-2\mathcal{d}(\theta)}$, while the right inequality follows from the same hypothesis and the relation
$$
(\lambda_+(\theta)-\lambda_-(\theta))^2=2{\rm Tr}(M_I(\theta)^2)-\left ({\rm Tr}(M_I(\theta))\right )^2=\mathcal{O}(|\theta-\theta_0|^2).
$$
\item For the third conclusion we notice that 
\beq\label{F-d-delta}
\mathcal{d}=\lambda_\circ^2-\delta^2
\eeq 
and $\lambda_\circ=(1/2)\Tr\big(M_I(\theta)\big)$ so that the smoothness of $\Sigma_I\ni\theta\mapsto\delta^2(\theta)\in\mathbb{R}_+$ follows. The formulas \eqref{F-non-deg} and \eqref{hypv1v2} are evident.  Here $\big(\Hes\lambda_\circ^2\big)_{jk}(\theta_0)=2\big(\partial_{\theta_j}\lambda_\circ(\theta_0)\big)\big(\partial_{\theta_k}\lambda_\circ(\theta_0)\big)$ defines a positive matrix so that the first conclusion and the formula \eqref{F-d-delta} imply that
\beq\label{F-Hess-d-delta}\begin{split}
0<2\left|\left\langle\zeta,\nabla\lambda_\circ(\theta_0)\right\rangle_{\mathbb{R}^2}\right|^2&=\left\langle\zeta,\big(\Hes\lambda_\circ^2\big)(\theta_0)\zeta\right\rangle_{\mathbb{R}^2}=\left\langle\zeta,\big(\Hes\; \delta^2\big)(\theta_0)\zeta\right\rangle_{\mathbb{R}^2}+\left\langle\zeta,\big(\Hes\; \mathcal{d}\big)(\theta_0)\zeta\right\rangle_{\mathbb{R}^2} \\
&\leq \left\langle\zeta,\big(\Hes\delta^2\big)(\theta_0)\zeta\right\rangle_{\mathbb{R}^2}-a_0|\zeta|_{\mathbb{R}^2}^2
\end{split}\eeq 
and we obtain \eqref{F-est-Hess-delta2} with some $\rho\geq a_0>0$.
\item For \eqref{F-lambdacirc-delta} we use once again  \eqref{F-Hess-d-delta} and define $a\in[0,1)$ by
	\beq\label{FD-a}
	a^2:=\underset{|\zeta|=1}{\max}\,\frac{2\big|\big\langle\big(\nabla\lambda_\circ\big)(\theta_0),\zeta\big\rangle_{\mathbb{R}^2}\big|^2}{\big\langle(\Hes\,\delta^2)(\theta_0)\zeta,\zeta\big\rangle_{\mathbb{R}^2}}\,\leq\frac{\rho-a_0}{\rho}\,<\,1.
	\eeq
\item We notice that the derivatives of $\lambda_\circ^2$ and $\delta^2$ in $\theta_0$ are 0 and we have the inequalities \eqref{F-Hess-d-delta} and \eqref{FD-a}.
By the Leibniz Newton formula we can write that for any $\tilde{a}\in(a,1)$ and any $\theta$ in some small enough neighborhood $V_0$ of $\theta_0$ (depending on $\tilde{a}$) we have
\begin{align*} \begin{split}
\lambda_\circ(\theta)^2&=\int_0^1dt\int_0^tds\left\langle(\theta-\theta_0),\big(\Hes\; \lambda_\circ^2\big)(\theta(t,s))(\theta-\theta_0)\right\rangle_{\mathbb{R}^2} \\
&\leq\tilde{a}^2\int_0^1dt\int_0^tds\left\langle(\theta-\theta_0),\big(\Hes\; \delta^2\big)(\theta(t,s))(\theta-\theta_0)\right\rangle_{\mathbb{R}^2}\  \\ & \leq\ \tilde{a}^2\delta^2(\theta)
\end{split}\end{align*}
(where we have used the notation $\theta(t,s):=(1-s)\theta_0+s\big((1-t)\theta_0+t\theta\big)$) and Conclusion 5 of the proposition follows.
\end{enumerate}
\end{proof} 

\subsection{The global frame defining the quasi-band.}\label{SS-LB-frame}
Our goal in this subsection is to build up the smooth global sections $ \hat{\Psi}_\pm:\mathbb{T}_*\rightarrow\mathscr{F}$ satisfying \eqref{F-Psi-PiI}- \eqref{F-Psi-delta+} and thus to realize the second step of our proof strategy in Subsection \ref{SS-ProofStr}. The first objective is to create a small gap in the spectrum of $\mathring{H}_I$ contained in $I$ by a small perturbation. The second objective is to smoothly extend its eigenfunctions outside $\Sigma_I$ preserving a spectral gap containing $I$.

\subsubsection{Separation of the crossing Bloch levels.}\label{SSS-eigval-sep}

We shall perturb the local Hamiltonian near $\theta_0$ in order to avoid the eigenvalue crossing.
Let $0\leq g\leq 1$ be a smooth cut-off function such that $g(x)=1$ if $|x|\leq 1/2$ and $g(x)=0$ if $|x|>1$. Let $c$ be the constant from the lower bound of \eqref{aprilie2}. Then let us define 
$$g_\delta(\theta):= g\left (\frac{c(\theta-\theta_0)}{\delta}\right ).$$
We assume that $0<\delta<\tilde{\delta}$ for some fixed value $\tilde{\delta}>0$ for which 
\beq\label{C-delta0-1}
B_{\tilde{\delta}/c}(\theta_0)\subset \Sigma_I.
\eeq

 Let us fix a unit vector $v^{(3)}\in \R^3$ such that
\begin{equation}\label{apr3}
v^{(3)}\perp {\rm Span}(v^{(1)}, v^{(2)}),\quad \vert v^{(3)}\vert =1 
\end{equation}
and define
$$\Sigma_I\ni\theta \mapsto \mathbf{F}_\delta(\theta):= \mathbf{F}(\theta) + v^{(3)}g_\delta (\theta)\, \delta/8\,.$$
If $|\theta-\theta_0|\geq \delta/(2c) $ then from  \eqref{aprilie2} (remember that $2|\mathbf{F}|=\lambda_+-\lambda_-$) we get 
$$|\mathbf{F}_\delta(\theta)|\geq |\mathbf{F}(\theta)|-\delta/8 \geq \delta/8\,.$$
If $|\theta-\theta_0|\leq \delta/(2c)$ then using \eqref{apr3} and \eqref{F-v12} we get: 
\beq\label{apr4}
|\mathbf{F}_\delta(\theta)|=|v^{(1)}(\theta-\theta_0)_1+v^{(2)}(\theta-\theta_0)_2+v^{(3)}\delta/8|+\mathcal{O}(|\theta-\theta_0|^2)\geq \delta/8 +\mathcal{O}(\delta^2)\geq \delta/16\,,
\eeq
so that it exists some $\delta_F>0$ such that the last inequality is true for any $\delta\in(0,\delta_F]$. Thus,
taking into account both upper bounds in \eqref{C-delta0-1} and in \eqref{apr4} for $\delta>0$,
we have proved that there exist positive constants $C$ and $\delta_0$ such that
\beq\label{F-est-gap}
\frac{\delta}{C} \leq |\mathbf{F}_\delta(\theta)| ,\quad \forall \theta\in \Sigma_I,\ \forall\delta\in(0,\delta_0]  \,.
\eeq

For $\theta\in\Sigma_I$ we define the perturbed local Hamiltonian $\mathring{H}_{I,\delta}(\theta)$
\beq\label{hc0}
\begin{split}
&\left\langle \varepsilon_j(\theta)\,,\,\mathring{H}_{I,\delta}(\theta)\varepsilon_k(\theta)\right\rangle_{\mathbb{C}^2}:=M_{I,\delta}(\theta)_{jk} \\
&M_{I,\delta}(\theta)\,:=\,M_I(\theta)\,+\,v^{(3)}\cdot\sigma  g_\delta(\theta) \delta/8\,=\,F_0(\theta)\bb1_2\,+\mathbf{F}_\delta(\theta)\cdot \sigma\,\in\mathcal{M}_{2\times2}(\mathbb{C})
\end{split}
\eeq
and its images through $\widehat{\Upsilon}(\theta)^{-1}$ in the Floquet representation, that act in $\widehat{\Pi}_I(\theta)\mathscr{F}_\theta$ (see \eqref{FD-Psi}):
$$
\widehat{H}_{I,\delta}(\theta)\,:=\,\widehat{\Upsilon}(\theta)^{-1}\mathring{H}_{I,\delta}(\theta)\widehat{\Upsilon}(\theta),\quad\forall\theta\in\Sigma_I.
$$
Due to our choice for $\supp\, g_\delta$,  we obtain that $\widehat{H}_{I,\delta}(\theta)=\widehat{H}_{I}(\theta)$ for $|\theta-\theta_0|\geq \delta/c$.  
The fiber operator $\widehat{H}_{I,\delta}(\theta)$ lives in the 2-dimensional space $\widehat{\Pi}_{I}(\theta)\mathscr{F}_\theta$ that does not depend on $\delta$. It has two distinct non-degenerate eigenvalues $$\lambda_{-,\delta}(\theta)<0<\lambda_{+,\delta}(\theta)\,,\, \forall \theta\in \Sigma_I\,.$$
The minimal distance between them, as function of $\theta\in\Sigma_I$, is bounded from below by the infimum of $2\; |\mathbf{F}_\delta|$ which is proportional to $\delta$, see \eqref{F-est-gap}. Since for $\delta>0$ the eigenvalues are non-degenerate, the regular perturbation theory allows us to deduce that they are smooth functions of $\theta$.  Thus for $0 < \delta\leq \delta_0$, with $\delta_0>0$ given by \eqref{F-est-gap} in agreement with \eqref{C-delta0-1} and \eqref{apr4}, the associated Riesz spectral projections $\widehat{\pi}_{\pm,\delta}(\theta)$ define smooth functions of $\theta\in\Sigma_I$. Their sum is equal to the 2-dimensional orthogonal projection $\widehat{\Pi}_{I}$. Finally we can write
\beq\label{F-PiI}
\widehat{\pi}_{-,\delta}(\theta)\oplus\widehat{\pi}_{+,\delta}(\theta)=\widehat{\Pi}_I(\theta),\quad\forall\theta\in\Sigma_I,\qquad \pi_{\pm, \delta}:=\mathcal{U}_\Gamma^{-1}
\left(\int_{\Sigma_{I}}d\theta\,\widehat{\pi}_{\pm,\delta}(\theta)\right)\mathcal{U}
_\Gamma
\eeq
and 
\begin{align}\label{hc5}
\widehat{H}_{I,\delta}(\theta)=\lambda_{-,\delta}(\theta)\, \widehat{\pi}_{-,\delta}(\theta)+
\lambda_{+,\delta}(\theta)\, \widehat{\pi}_{+,\delta}(\theta), \quad\forall\theta\in\Sigma_I\,.
\end{align}

Recalling \eqref{hc1}, let us also define the \enquote{perturbed fiber Hamiltonian}: 
\beq\label{hc1delta}
\widehat{H}_\delta(\theta):=\widehat{H}(\theta)+(\delta/8) \, 
\widehat{\Upsilon}(\theta)^{-1}
g_\delta(\theta) v^{(3)}\cdot \sigma  \, \widehat{\Upsilon}(\theta),\quad\forall\theta\in\mathbb{T}_*
\eeq
as self-adjoint operator in $\mathscr{F}_\theta=\mathcal{V}_\theta\big(L^2(\mathbb{T})\big)$ and the perturbed Hamiltonian
\beq\label{FD-Hdelta}
H_{\Gamma,\delta}\ :=\ \mathscr{U}_\Gamma^{-1}\left(\int_{\mathbb{T}_*}^\oplus d\theta\;\widehat{H}_\delta(\theta)\right)\mathscr{U}_\Gamma\,\equiv\,H_\Gamma\,+\,\delta \, K_g
\eeq
as self-adjoint operator acting in $L^2(\X)$, with the bounded self-adjoint perturbation
\beq\label{hckg}
K_g:=\mathscr{U}_\Gamma^{-1}\left(\frac{1}{8}\int_{\Sigma_{I}}^\oplus d\theta\;\widehat{\Upsilon}(\theta)^{-1}
g_\delta(\theta) v^{(3)}\cdot \sigma  \, \widehat{\Upsilon}(\theta)\right)\mathscr{U}_\Gamma\,.
\eeq
Then we have the equality $\widehat{H}_{I,\delta}(\theta)\,=\,\widehat{H}_\delta(\theta)\widehat{\Pi}_I(\theta)$.

\begin{remark}\label{R-spgap-Hdelta}  If $\delta\in (0,\delta_0)$ with $\delta_0$ defined in \eqref{F-est-gap},  we obtain  by regular perturbation theory:
	\begin{enumerate}
		\item The operators $H_{\Gamma,\delta}$ and $H_\Gamma$  are self-adjoint on the same domain $\mathscr{H}^2(\X)$.
		\item The self-adjoint operator $H_{\Gamma,\delta}$ has a spectral gap containing zero and having a width of order $\delta$, see \eqref{F-est-gap}. Thus, the spectrum of $H_{\Gamma,\delta}$ has a well separated bounded part  contained in $(-\infty,0)$.
		\item For any $\theta\in\Sigma_I$ we know that $\widehat{\Pi}_I(\theta)\mathscr{F}_\theta\subset\mathscr{F}_\theta^\infty$ so that we can construct some eigenfunctions of $\widehat{H}_{I,\delta}(\theta)$ which are smooth functions of $\X$.
		\item The spectral projection $P_{(-\infty,0]}(\widehat{H}_\delta)(\theta))$ is globally smooth in $\theta$ and has a constant rank $ k_0+1\geq 2$. Its orthogonal  complement $  P_{(0,\infty)}(\widehat H_\delta(\theta))$ is also globally smooth and has infinite rank. 
	\end{enumerate}
\end{remark}

We have already remarked that $H_{\Gamma,\delta}$ ia a self-adjoint operator in $L^2(\X)$ having a domain that contains $\mathscr{S}(\X)$; thus it defines a continuous linear operator from $\mathscr{S}(\X)$ to $\mathscr{S}^\prime(\X)$ and by the Kernels Theorem of Schwartz it has a distribution kernel of class $\mathscr{S}^\prime(\X\times\X)$. Then it also has a Weyl symbol $h_\delta \in \mathscr{S}^\prime(\Xi)$ such that:
\beq\label{F-hdelta}
H_{\Gamma,\delta}=\Op\big(h_\delta\big)\,.
\eeq

\begin{proposition}\label{P-Symb-Hdelta}
	The tempered distribution $h_\delta$ in \eqref{F-hdelta} belongs to the H\"{o}rmander type class $S^2_1(\X\times\X^*)$.
\end{proposition}

\begin{proof}
	One can generalize the decomposition in \eqref{F-ent-dec} to higher dimensions and to general lattices. We introduce the self-explanatory notation: 
	\beq\label{F-int-dec}
	\X\ni x\mapsto(\gamma(x),\hat{x})\in\Gamma\times\mathcal{E}
	\eeq 
	and we shall simply write $x=\gamma(x)+\hat{x}$.
	
	Using \eqref{FD-Hdelta} we shall write the symbol $h_\delta$ as the sum of our initial periodic symbol $h\in S^2_1(\X\times\X^*)$ and the symbol $\mathfrak{S}_{K_g}$ of the bounded self-adjoint operator $K_g$ from \eqref{hckg}. Let us start by studying the integral kernel $\mathfrak{K}_{K_g}$ of the rank 2 self-adjoint operator $K_g$ introduced in \eqref{hckg}. Let us notice that with $\{\widehat{\Upsilon}(\theta)\}_{\theta\in\Sigma_I}$ the smooth family of unitary operators defined by \eqref{F-Upsilon} we can write
	$$
	\widehat{\Upsilon}(\theta)^{-1}\sigma_\ell\widehat{\Upsilon}(\theta)=\underset{j,k}{\sum}\,(\sigma_\ell)_{jk} \hat{\varphi}_j(\theta,\hat{x})\overline{ \hat{\varphi}_k(\theta,\hat{y})}
	$$
	where $\{ \hat{\varphi}_j\}_{j=\pm}$ are defined in Definition \ref{D-onb-SigmaI}. Going further we obtain
	$$
	\mathfrak{K}_{K_g}(x,y)=\frac{1}{8(2\pi)^2}\underset{1\leq \ell\leq3}{\sum\,}v^{(3)}_\ell\int_{\mathbb{T}_*}\,d\theta\,g_\delta(\theta)\,e^{i<\theta,\gamma(x)-\gamma(y)>}\underset{j,k}{\sum}\,(\sigma_\ell)_{jk} \hat{\varphi}_j(\theta,\hat{x})\overline{ \hat{\varphi}_k(\theta,\hat{y})}
	$$
	
	 The functions $\{\hat{\varphi}_j(\theta,x)\}_{j=1,2}$ are smooth local sections in $\mathscr{F}^\infty$ (see \eqref{aprilie6}), thus the kernel $\mathfrak{K}_{K_g}$ is smooth in both variables $(x,y)\in\X\times\X$. Also, due to the smoothness in $\theta$ and the fact that the support of $g_\delta$ belongs to $\Sigma_I$, the kernel 
	  has rapid decay in the variable $x-y\in\X$. Moreover, we have periodicity in the variable $(x+y)/2$. We can then conclude that its associated distributional symbol (see also Subsection 2.1 in \cite{CIP})
	$$
	\mathfrak{S}_{K_g}(x,\xi)=(2\pi)^{-2}\int_{\X}dv\,e^{-i<\xi,v>}\mathfrak{K}_{K_g}(x+v/2,x-v/2)
	$$
	is of class $S^{-\infty}(\X\times\X^*)$. We conclude that $h_\delta\in S^2_1(\X\times\X^*)$.
\end{proof}

By a straightforward perturbation argument and Proposition 6.5 in \cite{IMP-2}, we obtain the following statement.
\begin{proposition}\label{C-Symb-rez-Hdelta}
	The resolvent of $H_{\Gamma,\delta}$ (where it exists) has a symbol of class $S^{-2}_1(\X\times\X^*)$.
\end{proposition}
\begin{proof}
    For completeness, let us sketch a short proof. Let $z$ be such that $(H_{\Gamma,\delta}-z\bb1)^{-1}$ exists. Applying the Beals commutator criterion \cite{Bea-77, CHP-3, IMP-2}, we obtain that  this resolvent is a pseudodifferential operator with a symbol $r_{\Gamma,\delta}\in S_0^0(\X\times\X^*)$. 
    
    Since $h_{\Gamma,\delta}-z$ is an elliptic symbol of type $S^{2}_1(\X\times \X^*)$, we may find a left parametrix $p(x,\xi)$ in $S^{-2}_1(\X\times \X^*)$ such that 
    $$p\sharp (h_{\Gamma,\delta}-z)= 1+ R,\quad R\in S^{-\infty}(\X\times \X^*).$$
    Composing with $r_{\Gamma,\delta}$ to the right we have
    $r_{\Gamma,\delta} =-p+ R\sharp r_{\Gamma,\delta},$ 
    which apriori holds in $S_0^0$, but due to the composition with a smoothing symbol, $R\sharp r_{\Gamma,\delta}$ is also smoothing, thus the previous identity shows that $r_{\Gamma,\delta}\in S^{-2}_1$.
\end{proof}

\subsubsection{The global smooth sections.}\label{SS-Gl-Sections.}

Using once again the Sz-Nagy intertwining unitary (and reducing if necessary the spectral window $I$ and its associated quasi-momentum domain $\Sigma_I$) as in \eqref{F-theta-var} and \eqref{F-Nagy-unit}, or using Remark \ref{R-contr-triv}, we can find two smooth local sections:
\beq\label{FD-local-basis}
\hat \psi_{\pm,\delta}:\Sigma_I\rightarrow\int^\oplus_{\Sigma_I}d\theta\,\mathscr{F}_\theta,\quad\hat \psi_{\pm,\delta}(\theta)\in\widehat{\pi}_{\pm,\delta}(\theta)\mathscr{F}_\theta,{\quad\|\hat \psi_{\pm,\delta}(\theta)\|_{\mathscr{F}_\theta}=1,}\quad  \forall\theta\in\Sigma_I.
\eeq
\begin{remark}\label{R-smoothness-local-frame}
Due to point (3) in Remark \ref{R-spgap-Hdelta} it follows that $\widehat{\pi}_{\pm,\delta}(\theta)\mathfrak{F}_\theta\subset\mathfrak{F}_\theta^\infty$ and we conclude that the above defined sections are smooth functions that satisfy:
$$
\hat \psi_{-,\delta}(\theta) \in P_{(-\infty,0]}(\widehat{H}_\delta(\theta))\mathscr{F}_\theta\,, \quad \hat \psi_{-,\delta}(\theta)\in P_{(0,\infty)}(\widehat{H}_\delta(\theta))\mathscr{F}_\theta,\quad
\forall\theta\in\Sigma_I.
$$
\end{remark}

Using Lemmas \ref{L-A1} and \ref{L-A2} in Appendix \ref{A-BBdl} we can extend these two local smooth sections to two global smooth sections:
\begin{align}\label{hc6}
\hat \Psi_{-,\delta}(\theta)\in P_{(-\infty,0]}(\widehat{H}_\delta(\theta))\mathscr{F}_{\theta},\quad \hat \Psi_{+,\delta}(\theta)\in P_{(0,\infty)}(\widehat{H}_\delta(\theta))\mathscr{F}_{\theta},\quad\|\hat \Psi_{\pm,\delta}(\theta)\|_{\mathscr{F}_\theta}=1,\quad\forall\theta\in\mathbb{T}_*
\end{align}
that are orthogonal to each other. At this step, the condition $k_0\geq2$ is crucial (in order to obtain $\hat{\Psi}_{-,\delta}$).
Thus we have proven the following statement.
\begin{proposition}\label{P-smooth-glob-frame}
	There exist two smooth global sections
	$$
	\mathbb{T}_*\ni\theta\mapsto\hat{\Psi}_{\pm,\delta}(\theta)\in\{\psi\in\mathscr{F}_\theta\,,\,\|\psi\|_{\mathscr{F}_\theta}=1\}
	$$
	that satisfy the properties:
\begin{align}\label{hc8}
&\hat{\Psi}_{-,\delta}(\theta)\in P_{(-\infty,0]}(\widehat{H}(\theta))\mathscr{F}_\theta,\quad \hat{\Psi}_{+,\delta}(\theta)\in P_{ (0,\infty)}(\widehat{H}(\theta))\mathscr{F}_\theta,\quad\forall\theta\not\in \Sigma_I; \\ \label{hc8-bis}
&\mathcal{L}_{\mathbb{C}}\{\hat{\Psi}_{-,\delta} (\theta),\hat{\Psi}_{+,\delta} (\theta)\}\,=\,\widehat{\Pi}_I(\theta)\mathscr{F}_\theta,\quad\forall\theta\in\Sigma_I.
\end{align}
\end{proposition}
This is the global smooth frame satisfying \eqref{F-Psi-PiI} - \eqref{F-Psi-delta+} that we were looking for.
We shall use the above defined global smooth frame in order to define the \enquote{quasi-band} subspace and its associated orthogonal projection.
As in  the construction of a Wannier basis for an isolated Bloch band (see \cite{HS}, \cite{Ne-RMP}), we define the orthonormal system $\{\Psi_{-,\delta},\Psi_{+,\delta}\}$:
$$
\Psi_{\pm,\delta}\,:=\,\mathcal{U}_\Gamma^{-1}\left(\int^\oplus_
{\mathbb{T}_*}d\theta\,\hat{\Psi}_{\pm,\delta}(\theta)\right),
$$
and the family of their translations  by elements from $\Gamma$\,:
\beq\label{DF-Psi-gamma}
\Psi_{\pm,\gamma,\delta}\,:=\,\tau_\gamma\,\Psi_{\pm,\delta}\,,\quad\forall\gamma\in\Gamma.
\eeq
\begin{remark}\label{R-smoothness-global-frame}
Combining Remark \ref{R-smoothness-local-frame} with the global smoothness of the sections we deduce that $\hat{\Psi}_{\pm,\delta}\in\mathscr{S}(\X)$.
\end{remark}

Let us notice here for further use that
$$
\mathcal{U}_\Gamma\Psi_{\pm,\gamma,\delta}=\int^\oplus_{\mathbb{T}_*}d\theta\,  \hat{\Psi}_{\pm,\gamma,\delta}(\theta),\quad
\hat{\Psi}_{\pm,\gamma,\delta}(x,\theta)=
\hat{\Psi}_{\pm,\delta}(x-\gamma,\theta)=e^{-i<\theta,\gamma>}\hat{\Psi}_{\pm,\delta}(x,\theta)\,,
$$
and that for any $(\alpha,\beta)\in\Gamma\times\Gamma$
$$
2\pi\langle\Psi_{\pm,\alpha,\delta},\Psi_{\pm,\beta,\delta}\rangle_{L^2(\X)}=\int\limits_{\mathbb{T}_*}d\theta\,\langle\hat{\Psi}_{\pm,\alpha,\delta}(\theta),\hat{\Psi}_{\pm,\beta,\delta}(\theta)\rangle_{\mathscr{F}_\theta}=\int\limits_{\mathbb{T}_*
}d\theta\,e^{i<\theta,\alpha-\beta>}\|\hat{\Psi}_{\pm,\delta}(\theta)\|^2_{\mathscr{F}_\theta}=2\pi\delta_{\alpha\beta}\,,
$$
and
$$
2\pi\langle\Psi_{-,\alpha,\delta},\Psi_{+,\beta,\delta}\rangle_{L^2(\X)}=\int\limits_{\mathbb{T}_*}d\theta\,\langle
\hat{\Psi}_{-,\alpha,\delta}(\theta),\hat{\Psi}_{+,\beta,\delta}(\theta)\rangle_{\mathscr{F}_\theta}=\int\limits_
{\mathbb{T}_*
}d\theta\,e^{i<\theta,\alpha-\beta>}\langle\hat{\Psi}_{-,\delta}(\theta),\hat{\Psi}_{+,\delta}(\theta)\rangle_{\mathscr{F}_\theta}=0\,.
$$
Thus, for any fixed $\delta\in(0,\delta_0)$, the family $\big\{\Psi_{j,\gamma,\delta}\,,\,j=\pm,\,\gamma\in\Gamma\}$ is an orthonormal system in $L^2(\X)$. 

\subsection{The \enquote{quasi-band} orthogonal projection.}

From here on we shall fix some $\delta\in(0,\delta_0)$. We define two rank $1$ projections in $L^2(\X)$
$$
\pi_{\pm,0,\delta}:=|\Psi_{\pm,\delta}\rangle\langle\Psi_{\pm,\delta}|=\left|\mathcal{U}_\Gamma^{-1}\hat{\Psi}_{\pm,\delta}
\right\rangle\left\langle\mathcal{U}_\Gamma^{-1}\hat{\Psi}_{\pm,\delta}\right|
$$
that are integral operators on $L^2(\X)$. In order to write down their integral kernels let us recall some results and fix some notations. First, we recall the decomposition \eqref{F-int-dec} for the representation $\X\simeq\Gamma\times\mathcal{E}$. 
Going back to Definition \eqref{UGamma} of $\mathcal{U}_\Gamma$, let us recall the explicit form of its inverse:
$$
\big(\mathcal{U}_\Gamma^{-1}\hat{\Phi}\big)(x)=\frac{1}{2\pi}\int_{\mathbb{T}_*}d\theta\,e^{i<\theta,\gamma(x)>}\hat{\Phi}(\theta,\hat{x})
$$
A simple computation shows that the integral kernels of the rank one orthogonal projections $\pi_{\pm,0,\delta}$ are given by:
\beq\label{F-pi-pm-int-kernel}
\mathfrak{K}_{\pm,0,\delta}(x,y)=\frac{1}{(2\pi)^4}\int_{\mathbb{T}_*}d\theta\,e^{i<\theta,\gamma(x)>}\hat{\Psi}_{\pm,\delta}(\theta,\hat{x})\int_{\mathbb{T}_*}d\omega\,e^{-i<\omega,\gamma(y)>}\overline{\hat{\Psi}_{\pm,\delta}(\omega,\hat{y})}.
\eeq

Let us consider the translations of the two rank one orthogonal projections $\pi_{\pm,0,\delta}$ by elements from $\Gamma$:
$$
\pi_{\pm,\alpha,\delta}:=|\tau_\alpha\Psi_{\pm,\delta}\rangle\langle
\tau_\alpha\Psi_{\pm,\delta}|=\tau_\alpha\, \pi_{\pm,0,\delta}\,\tau_\alpha^*,\quad\forall\alpha\in\Gamma.
$$
They have integral kernels
$$
\mathfrak{K}_{\pm,\alpha,\delta}(x,y)=\frac{1}{(2\pi)^4}\int_{\mathbb{T}_*}d\theta\,e^{i<\theta,\gamma(x)>}\hat{\Psi}_{\pm,\delta}(\theta,\hat{x})\int_{\mathbb{T}_*}d\omega\,e^{-i<\omega,\gamma(y)>}e^{i<\theta-\omega,\alpha>}\overline{\hat{\Psi}_{\pm,\delta}(\omega,\hat{y})}
$$
and we have seen before that they are mutually orthogonal. Let us consider their orthogonal sum:
$$
\pi_{\pm, \delta}\,:=\,\underset{\gamma\in\Gamma}{\bigoplus}\pi_{\pm, \gamma,\delta}.
$$
The smoothness of the sections $\hat{\Psi}_{\pm,\delta}:\mathbb{T}_*\rightarrow\mathscr{F}$ implies the rapid decay of the integral kernels in \eqref{F-pi-pm-int-kernel} so that we may define the integral kernels of $\pi_{\pm,\delta}$ as the limit of the following series:
\begin{align}\label{F-int-kernel-pi-pm}
\mathfrak{K}_{\pm,\delta}(x,y)&:=\underset{\alpha\in\Gamma}{\sum}\mathfrak{K}_{\pm,\alpha,\delta}(x,y)\nonumber \\&=\frac{1}{(2\pi)^4}\int_{\mathbb{T}_*}d\theta\,e^{i<\theta,\gamma(x)>}\hat{\Psi}_{\pm,\delta}(\theta,\hat{x})\int_{\mathbb{T}_*}d\omega\,e^{-i<\omega,\gamma(y)>}\left(\underset{\alpha\in\Gamma}{\sum}e^{i<\theta-\omega,\alpha>}\right)\overline{\hat{\Psi}_{\pm,\delta}(\omega,\hat{y})} \\
&=\frac{1}{(2\pi)^2}\int_{\mathbb{T}_*}d\theta\,e^{i<\theta,\gamma(x)-\gamma(y)>}\hat{\Psi}_{\pm,\delta}(\theta,\hat{x})
\overline{\hat{\Psi}_{\pm,\delta}(\theta,\hat{y})}\nonumber
\end{align}
by using the inverse Fourier theorem. It is easy to verify that for any $\gamma\in\Gamma$ we have that $\mathfrak{K}_{\pm,\delta}(x+\gamma,y+\gamma)=\mathfrak{K}_{\pm,\delta}(x,y)$ so that they define $\Gamma$-invariant operators on $L^2(\X)$. By construction they are infinite rank orthogonal projections.

Similarly we define the rank 2 projections in $L^2(\X)$
given by the orthogonal sums:
\beq\label{FD-Pi-delta}
\Pi_{0,\delta}:=\pi_{-,0,\delta}\oplus\pi_{+,0,\delta}\,,\quad\Pi_{\gamma,\delta}:=\pi_{-,\gamma,\delta}\oplus
\pi_{+,\gamma, \delta}\, ,\quad\forall\gamma\in\Gamma.
\eeq

\begin{definition}\label{D-Pdelta}
	We define the \enquote{quasi-band} orthogonal projection
	$$
	P_{I,\delta} \ :=\ \underset{\gamma\in\Gamma}{\sum} \,\Pi_{\gamma,\delta}\ \in\ \mathbb{LP}\big(L^2(\X)\big).
	$$
\end{definition}
For the fixed $\delta\in(0,\delta_0)$ we shall introduce the following orthogonal decomposition of unity in $L^2(\X)$:
\beq \label{defE-delta}
\bb1=E_{-,\delta}\oplus E_{+,\delta},\qquad E_{-,\delta}=P_{(-\infty,0]}\big(H_{\Gamma,\delta}\big),\qquad E_{+,\delta}=P_{(0,\infty)}\big(H_{\Gamma,\delta}\big).
\eeq
with $E_{\pm,\delta}$ having the fibers  $\widehat E_{\pm,\delta}(\theta)$.
\begin{remark}\label{R-E-pm}
	By construction we have that $$\pi_{\pm,\gamma,\delta} E_{\pm,\delta}=\pi_{\pm,\gamma,\delta}$$ for any $\gamma\in\Gamma$ (the projections $E_{\pm,\delta}$ being $\Gamma$-periodic). Thus $$\pi_{\pm,\delta} E_{\pm,\delta}=\pi_{\pm,\delta}\,.$$
\end{remark}

With the second point in Remark \ref{R-spgap-Hdelta} in mind, we can find in the negative half-plane
a simple loop $\mathscr{C}_0$ which  surrounds the negative part of the spectrum of $H_{\Gamma,\delta}$  and  remains at a distance  of order $\delta$ from the spectrum of $H_{\Gamma,\delta}$. Thus we can write that
\beq\label{F-Edelta-}
E_{-,\delta}\ =\ \frac{1}{2\pi i}\oint_{\mathscr{C}_0}d\z\,\big(H_{\Gamma,\delta}-\z\bb1\big)^{-1}.
\eeq
The results in Section 6 of \cite{IMP-2} and some standard arguments imply the following corollary.
\begin{corollary}\label{C-Edelta-pm}
	Denoting by $\mathfrak{S}_T\in\mathscr{S}^\prime(\Xi)$ the distribution symbol of $T\in\mathbb{B}\big(L^2(\X)\big)$ (as in \eqref{N-magn-symb} with the usual  non-magnetic Weyl quantization), and considering $\delta_0>0$ given in \eqref{F-est-gap}, we have that for any $\delta \in (0,\delta_0)$, the symbols  $\mathfrak{S}_{E_{-,\delta}}\in S^{-\infty}(\X\times\X^*)$ and $\mathfrak{S}_{E_{+,\delta}}\in S^0_1(\X\times\X^*)$.
\end{corollary}
Note that we do not claim uniform control with respect to $\delta$ and recall that $\delta_0>0$ in \eqref{F-est-gap} is supposed to satisfy the conditions \eqref{C-delta0-1} and \eqref{apr4}.
\begin{proof}
	We have:
	$$
	E_{-,\delta}\ = \frac{1}{2\pi i}\oint_{\mathscr{C}_0}d\z\,\big(H_{\Gamma,\delta}-\z\bb1\big)^{-1}\,.
	$$
	Using that $(E_{-,\delta})^2=E_{-,\delta}$ and recalling that the order of the Moyal product of two H\"{o}rmander type symbols is the sum of their orders, we conclude that $\mathfrak{S}_{E_{-,\delta}}\in S^{-4}_1(\X\times\X^*)$. Repeating this argument as many times as necessary we obtain that $\mathfrak{S}_{E_{-,\delta}}\in S^{-N}_1(\X\times\X^*)$ for any $N\in\mathbb{N}$, i.e. $\mathfrak{S}_{E_{-,\delta}}\in S^{-\infty}(\X\times\X^*)$. For $\mathfrak{S}_{E_{+,\delta}}$ we just notice that by \eqref{defE-delta} $\mathfrak{S}_{E_{+,\delta}}=1-\mathfrak{S}_{E_{-,\delta}}$.
\end{proof}

Noticing that $P_{I,\delta}=\pi_{-,\delta}\oplus\pi_{+,\delta}$ and using \eqref{F-int-kernel-pi-pm} one obtains the integral kernel and the Weyl symbol for the orthogonal projection $P_{I,\delta}$ and the following statement.
\begin{proposition}\label{P-Symb-PI}
	 There exists $\delta_0 >0$ (given in \eqref{F-est-gap}) such that, for $\delta \in (0,\delta_0)$, there exist symbols $p_{I,\delta} \in S^{-\infty}(\X\times\X^*)$, and $p_{\pm,\delta} \in S^{-\infty}(\X\times\X^*)$ that are $\Gamma$-periodic in $\X$, such that 
	$
	P_{I,\delta} :=\Op(p_{I,\delta})$ and $\pi_{\pm,\delta}=\Op(p_{\pm,\delta} ).
	$
\end{proposition}

\subsection{The \enquote{quasi-band} Hamiltonian.}

The next step after defining the \enquote{quasi-band} subspace associated to a spectral window is to consider the projected Hamiltonian:
\beq\label{FD-qBandHamil}
P_{I,\delta}H_\Gamma P_{I,\delta}
\eeq
and in order to achieve the third step of our proof strategy, to compare its spectrum in the fixed spectral region with the \enquote{true} spectrum of  $H_\Gamma$. The fact that $P_{I,\delta}$ is \enquote{close} in the Floquet representation to a spectral projection of $H_\Gamma$ only \enquote{locally} on $I$ makes this comparison rather difficult. We present in Appendix \ref{A-FS-arg} an abstract result that allows us to deal with this problem. In this subsection we verify that the  triple $\big(H_\Gamma,P_{I,\delta},I_\delta\big)$ is admissible  for some $I_\delta\subset\mathring{I}$ (see Definition \ref{H-band-red}), as a first step in verifying the same property for the problem with a magnetic field. 

\paragraph{Notation.} 
\begin{itemize}
\item We shall use the notation $P_1\prec P_2$ for two orthogonal projections satisfying the identity $P_1=P_1P_2$, and in this case we shall denote by $P_2\ominus P_1$ the orthogonal projection on the orthogonal subspace of $P-1\mathcal{H}$ in $P_2\mathcal{H}$, i.e. $P_2\mathcal{H}\cap P_1\mathcal{H}^\bot$. Moreover, when we want to emphasize the orthogonality of the terms of a sum of two orthogonal projections $P$ and $Q$ we shall use the notation $P\oplus Q$. 
\item Given an orthogonal projection $P$ in $\mathcal{H}$ we shall denote by $\text{\tt rk}(P)$ the dimension of its image.
\item We define
\begin{align}\label{F-desc-QI}
Q_{I,\delta}:&=\bb1-P_{I,\delta}=  P_{I,\delta}^\bot=\bb1-\big(\pi_{-,\delta}\oplus\pi_{+,\delta}
\big)
\\ \nonumber
&=\mathcal{U}_\Gamma^{-1}\left(
\int^\oplus_{\mathbb{T}_*}d\theta\,\widehat{Q}_{I,\delta}(\theta)\right)\mathcal{U}_\Gamma=\bb1-
\mathcal{U}_\Gamma^{-1}\left(
\int^\oplus_{\mathbb{T}_*}d\theta\,\big(\widehat{\pi}_{-,\delta}(\theta)\oplus
\widehat{\pi}_{+,\delta}(\theta)\big)
\right)\mathcal{U}_\Gamma.
\end{align}
\item We denote by 
$\widehat{E}_{\leq k,\delta}(\theta)$  the orthogonal projection on the first $k+1$ Bloch eigenvalues of $\widehat{H}_\delta(\theta)$  (the first Bloch level having index $k=0$) and $\widehat{E}_{\geq k,\delta}(\theta)$ for the orthogonal projection on the Bloch eigenvalues of $\widehat{H}_\delta(\theta)$ greater or equal to $\lambda_{k,\delta}(\theta)$. 
\end{itemize}

Starting from \eqref{F-desc-QI} we shall  emphasize an orthogonal decomposition of the projection $Q_{I,\delta}$. Remark \ref{R-E-pm} implies that $\widehat{\pi}_{-,\delta}(\theta)\prec\widehat{E}
_{-,\delta}(\theta)$ for any $\theta\in\mathbb{T}_*$ and we can write
$$
\big(\bb1-\widehat{\pi}_{-,\delta}(\theta)\big)=\big(
\widehat{E}_{-,\delta}(\theta)\ominus\widehat{\pi}_{-,\delta}(\theta)\big)\oplus\widehat{E}_{+,\delta}(\theta),\quad\forall\theta\in\mathbb{T}_*.
$$
Moreover, for $\theta\in \Sigma_I$ we have the decomposition
$$
\bb1-\widehat{\pi}_{-,\delta}(\theta)=\widehat{E}_{\leq k_0-1,\delta}(\theta)\oplus\widehat{\pi}_{+,\delta}(\theta)\oplus\widehat{E}_{\geq k_0+2,\delta}(\theta).
$$
For $\theta\in\mathbb{T}_*\setminus \Sigma_I$ we define the following orthogonal projection: 
\beq\label{F-hatF-}
\widehat{F}_{-}
(\theta):=\big(\bb1-\widehat{\pi}_{-,\delta}(\theta)\big)\ominus\widehat{E}_{\geq k_0+1,\delta}(\theta).
\eeq
It has  $\text{\tt rk\,}(\widehat{F}_{-}(\theta))=k_0\geq1$ and verifies $\widehat{F}_{-}
	(\theta)<\underset{k\leq k_0}{\sum}\widehat{\pi}_k(\theta)$.

Remark \ref{R-E-pm} and Remark \ref{R-D-H-local} also imply that $\widehat{\pi}_{+,\delta}(\theta)\prec\widehat{E}_{+,\delta}(\theta)\cap\mathscr{F}^2_\theta$ for any $\theta\in\mathbb{T}_*$ and that we can write $\big(\bb1-\widehat{\pi}_{+,\delta}(\theta)\big)=
\widehat{E}_{-,\delta}(\theta)\oplus\big(
\widehat{E}_{+,\delta}(\theta)\ominus\widehat{\pi}_{+,\delta}(\theta)\big)$. Moreover for $\theta\in \Sigma_I$ we have the decomposition
$$
\bb1-\widehat{\pi}_{+,\delta}(\theta)=\widehat{\pi}_{-,\delta}(\theta)\oplus\widehat{E}_{\leq k_0-1,\delta}(\theta)\oplus\widehat{E}_{\geq k_0+2,\delta}(\theta).
$$
For $\theta\in\mathbb{T}_*\setminus \Sigma_I$ we define the following orthogonal projection: 
\beq\label{F-hatF+}
\widehat{F}_{+}
(\theta):=\big(\bb1-\widehat{\pi}_{+,\delta}(\theta)\big)\ominus\widehat{E}_{\leq k_0,\delta}(\theta)
\eeq
 and  note that  $\text{\tt rk\,}(\widehat{F}_{+}(\theta))=+\infty$ and $\widehat{F}_{+}
 (\theta)<\underset{k\geq k_0+1}{\sum}\widehat{\pi}_k(\theta)$.

\begin{remark}\label{R-desc-QI-Epm}
The previous discussion implies that we have the orthogonal decomposition
\beq\label{F-desc-QI-Epm}
\widehat{Q}_{I,\delta}(\theta)=\big(
\widehat{E}_{-,\delta}(\theta)\ominus\widehat{\pi}_{-,\delta}(\theta)\big)\oplus\big(
\widehat{E}_{+,\delta}(\theta)\ominus\widehat{\pi}_{+,\delta}(\theta)\big)=\big(\widehat{Q}_{I,\delta}(\theta)\widehat{E}_{-,\delta}(\theta)\big)\oplus\big(\widehat{Q}_{I,\delta}(\theta)
\widehat{E}_{+,\delta}(\theta)\big).
\eeq
Moreover, for $\theta\in \Sigma_I$ we can write that  $\widehat{Q}_{I,\delta}(\theta)=\widehat{E}_{\leq k_0-1,\delta}(\theta)\oplus\widehat{E}_{\geq k_0+2,\delta}(\theta)$ while for $\theta\in\mathbb{T}_*\setminus \Sigma_I$ we have that $\widehat{Q}_{I,\delta}(\theta)=\widehat{F}_{-}
(\theta)\oplus\widehat{F}_{+}(\theta)$.
\end{remark}

\begin{definition}\label{D-Qmp}
	Associated to the orthogonal decomposition \eqref{F-desc-QI-Epm} we introduce the notations:
	$$
	\widehat{Q}_{I,-,\delta}(\theta):=\big(\widehat{Q}_I(\theta)\widehat{E}_{-,\delta}(\theta)\big),\qquad
	\widehat{Q}_{I,+,\delta}(\theta):=\big(\widehat{Q}_I(\theta)
	\widehat{E}_{+,\delta}(\theta)\big).
	$$
\end{definition}

\begin{proposition}
	The infinite dimensional orthogonal projections $Q_{I,\pm,\delta}$ in $L^2(\X)$ associated with the  orthogonal projections $\widehat{Q}_{I,\pm,\delta}(\theta)$ defined above 
	$$
	Q_{I,\pm,\delta}\,:= \,\mathcal{U}_\Gamma^{-1}\left(\int_{\mathbb{T}_*}^
	\oplus d\theta\,\widehat{Q}_{I,\pm,\delta}(\theta)\right)\mathcal{U}_\Gamma
	$$
	are $\Gamma$-periodic orthogonal projections giving an orthogonal decomposition of  $Q_{I,\delta}L^2(\X)$ and we have the equalities:
	\begin{align}\label{F-incl-QI-}
	&\widehat{Q}_{I,-,\delta}(\theta)=\left\{\begin{array}{l}
	\widehat{E}_{\leq k_0-1,\delta}(\theta),\quad\forall\theta\in \Sigma_I\\
	\widehat{F}_-(\theta),\quad\forall\theta\in \mathbb{T}_*\setminus \Sigma_I
	\end{array}\right. \\\label{F-incl-QI+}
	&\widehat{Q}_{I,+,\delta}(\theta)=\left\{\begin{array}{l}
	\widehat{E}_{\geq k_0+2,\delta}(\theta),\quad\forall\theta\in \Sigma_I\\
	\widehat{F}_+(\theta),\quad\forall\theta\in \mathbb{T}_*\setminus \Sigma_I
	\end{array}\right.
	\end{align}
	defining smooth global projection-valued sections.
\end{proposition}
\begin{proof}
	From the previous considerations we notice that
	$$
	\widehat{Q}_I(\theta)\widehat{E}_{-,\delta}(\theta)=
	\widehat{E}_{-,\delta}(\theta)\widehat{Q}_I(\theta)=
	\widehat{E}_{-,\delta}(\theta)\ominus
	\widehat{\pi}_{-,\delta}(\theta),
	$$
and	
	$$
	\widehat{Q}_I(\theta)\widehat{E}_{+,\delta}(\theta)=
	\widehat{E}_{+,\delta}(\theta)\widehat{Q}_I(\theta)=
	\widehat{E}_{+,\delta}(\theta)\ominus
	\widehat{\pi}_{+,\delta}(\theta),
	$$
	while the right hand side objects have been proved to be smooth global projection-valued sections.
\end{proof}

\begin{remark}\label{R-ordSymbOpm}
	We notice that $$Q_{I,\pm,\delta}=E_{\pm,\delta}\ominus\pi_{\pm, \delta}=E_{\pm,\delta}\big(\bb1-\pi_{\pm, \delta}\big)=
	\Op\big(\mathfrak{S}_{E_{\pm,\delta}}\sharp(1-p_{\pm,\delta})\big)\,.$$
	Moreover, using Corollary \ref{C-Edelta-pm} and Proposition \ref{P-Symb-PI} we infer that
	$q_{I,-,\delta}:=\mathfrak{S}_{E_{-,\delta}}\sharp(1-p_{-,\delta} )$ belongs to  $S^{-\infty}(\X\times\X^*)$ and 
	$q_{I,+,\delta}:=\mathfrak{S}_{E_{+,\delta}}\sharp(1-p_{+,\delta})$ belongs to $ S^{0}_1(\X\times\X^*)$.
\end{remark}

\begin{remark}\label{R-HPI}
	Our choice $\Psi_\pm(\theta)\in\mathscr{F}^2_\theta$ implies that $P_{I{,\delta}}L^2(\X)\subset\mathscr{H}^2(\X)$ and thus the product $H_\Gamma P_{I{,\delta}}$ has a bounded extension to $L^2(\X)$.
	This also implies that $Q_{I{,\delta}}H_\Gamma Q_{I{,\delta}}$ is a self-adjoint operator in $L^2(\X)$.
\end{remark}
We consider the open interval containing $0\in\mathbb{R}$:
\beq\label{DF-I-delta}
\mathring{I}=(-\Lambda_-\,,\,\Lambda_{+})\subset  I.
\eeq
\begin{proposition}\label{P-H-band-0-field}
	The triple $\big(H_\Gamma, P_{I,\delta},\mathring{I}\big)$ is admissible in the sense of Definition \ref{H-band-red}. 
\end{proposition}
\begin{proof}~

\noindent\textbf{Step 1:}
	We prove that formula \eqref{F-desc-QI-Epm} induces the following decomposition of the projected Hamiltonian $Q_{I,\delta}H_\Gamma Q_{I,\delta}$:
	$$
	Q_{I,\delta}H_\Gamma Q_{I,\delta}=Q_{I,-,\delta}H_\Gamma Q_{I,-,\delta}\,\oplus\,Q_{I,+,\delta}H_\Gamma Q_{I,+,\delta}.
	$$
	
Changing to the Floquet representation we can write that
\begin{align} \nonumber
Q_{I,\delta}H_\Gamma Q_{I,\delta}&=\mathcal{U}_\Gamma^{-1}\left(\int_{\mathbb{T}_*}^
\oplus d\theta\,\big(\widehat{Q}_{I,-,\delta}(\theta)\oplus
\widehat{Q}_{I,+,\delta}(\theta)\big)
\left(\underset{k\in
	\mathbb{N}}{\sum}\lambda_k(\theta)\widehat{\pi}_k(\theta) \right) \big(\widehat{Q}_{I,-,\delta}(\theta)\oplus
\widehat{Q}_{I,+,\delta}(\theta)\big)\right)
\mathcal{U}_\Gamma.
\end{align}
We decompose the integral over $\mathbb{T}_*$ as the sum of the integrals over $\Sigma_I$ and its complementary in $\mathbb{T}_*$. Using \eqref{F-incl-QI-} and \eqref{F-incl-QI+} we notice that for $\theta \in \Sigma_I$:
$$
 \widehat{E}_{\leq k_0-1,\delta}(\theta)=\underset{k\leq k_0-1}{\sum}\widehat{\pi}_k(\theta),\quad\widehat{E}_{\geq k_0+2,\delta}(\theta)=\underset{k\geq k_0+2}{\sum}\widehat{\pi}_k(\theta)
$$
so that
\begin{align}
&\int\limits_{\Sigma_I}^
\oplus \;d\theta\,\big(\widehat{E}_{\leq k_0-1,\delta}(\theta)\oplus
\widehat{E}_{\geq k_0+2,\delta}(\theta)\big)
\left(\underset{k\in
	\mathbb{N}}{\sum}\lambda_k(\theta)\widehat{\pi}_k(\theta) \right) \big(\widehat{E}_{\leq k_0-1,\delta}(\theta)\oplus
\widehat{E}_{\geq k_0+2,\delta}(\theta)\big) \label{F-Q-HQ-} \\ \nonumber 
& \hspace*{0.2cm}=\int\limits_{\Sigma_I}^
\oplus d\theta\,\widehat{E}_{\leq k_0-1,\delta}(\theta)\Big(\underset{k\leq k_0-1}{\sum}\lambda_k(\theta)\widehat{\pi}_k(\theta) \Big)
\widehat{E}_{\leq k_0-1,\delta}(\theta) \,\oplus\,
\widehat{E}_{\geq k_0+2,\delta}(\theta)\Big(\underset{k\geq k_0+2}{\sum}\lambda_k(\theta)\widehat{\pi}_k(\theta) \Big)\widehat{E}_{\geq k_0+2,\delta}(\theta).
\end{align}
Using once again \eqref{F-incl-QI-} and \eqref{F-incl-QI+} and \eqref{F-hatF-} and \eqref{F-hatF+} we obtain
\begin{align}\label{F-Q+HQ+}
&\int\limits_{\mathbb{T}_*
	\setminus \Sigma_I}^
\oplus d\theta\,\big(\widehat{F}_{-}(\theta)\oplus
\widehat{F}_{+}(\theta)\big)
\left(\underset{k\in
	\mathbb{N}}{\sum}\lambda_k(\theta)\widehat{\pi}_k(\theta) \right) \big(\widehat{F}_{-}(\theta)\oplus
\widehat{F}_{+}(\theta)\big) \\ \nonumber 
&  \hspace*{0.5cm}=\int\limits_{\mathbb{T}_*\setminus \Sigma_I}^
\oplus d\theta\,\widehat{F}_{-,\delta}(\theta)\Big(\underset{k\leq k_0}{\sum}\lambda_k(\theta)\widehat{\pi}_k(\theta) \Big) \widehat{F}_{-,\delta}(\theta)\,\oplus\,
\widehat{F}_{+,\delta}(\theta)\Big(\underset{k\geq k_0+1}{\sum}\lambda_k(\theta)\widehat{\pi}_k(\theta)\Big).
\end{align}
Similar arguments show that the mixed terms $Q_{I,-,\delta}H_\Gamma Q_{I,+,\delta}$ and $Q_{I,+,\delta}H_\Gamma Q_{I,-,\delta}$ are zero.

\noindent\textbf{Step 2:} We localize the spectra of the self-adjoint operators $Q_{I,-,\delta}H_\Gamma Q_{I,-,\delta}$ and $Q_{I,+,\delta}H_\Gamma Q_{I,+,\delta}$ and prove that they are well separated inducing a spectral gap for $Q_{I,\delta}H_\Gamma Q_{I,\delta}$. 

We shall separately consider the two contributions to \eqref{F-Q-HQ-} and the two contributions to \eqref{F-Q+HQ+}. We notice that
$$
(-E_0)\widehat{E}_{\leq k_0-1,\delta}(\theta)\,\leq\,\widehat{E}_{\leq k_0-1,\delta}(\theta)\Big(\underset{k\leq k_0-1}{\sum}\lambda_k(\theta)\widehat{\pi}_k(\theta) \Big)
\widehat{E}_{\leq k_0-1,\delta}(\theta)\,\leq\,\big(-\Lambda_-\big)\widehat{E}_{\leq k_0-1,\delta}(\theta),\quad\forall\theta\in \Sigma_I.
$$
{  On the complementary of $\Sigma_I$ in $\mathbb{T}_*$ we have that 
$$
(-E_0)\widehat{E}_{\leq k_0,\delta}(\theta)\,\leq\,\widehat{F}_{-}(\theta)\Big(\underset{k\leq k_0}{\sum}\lambda_k(\theta)\widehat{\pi}_k(\theta) \Big) \widehat{F}_{-}(\theta)\,\leq\,(- \Lambda_-)\widehat{E}_{\leq k_0,\delta}(\theta),\quad\forall\theta\in\mathbb{T}_*\setminus \Sigma_I.
$$
	We conclude that
\begin{align}\label{F-est-}
\big(-E_0\big)Q_{I,-,\delta}\,\leq\,Q_{I,-,\delta}\,H_\Gamma\, Q_{I,-,\delta}\,&\leq\,(- \Lambda_-)Q_{I,-,\delta}.
\end{align}
In a similar way we notice that
\begin{align*}
\Lambda_+\widehat{E}_{\geq k_0+2,\delta}(\theta)\,\leqq\,\widehat{E}_{\geq k_0+2}(\theta)\Big(\underset{k\geq k_0+2}{\sum}\lambda_k(\theta)\widehat{\pi}_k(\theta) \Big)
\widehat{E}_{\geq k_0+2,\delta}(\theta),\quad\forall\theta\in \Sigma_I,\\
\Lambda_+\widehat{E}_{\geq k_0+1,\delta}(\theta)\,\leq\,\widehat{F}_{+}(\theta)\Big(\underset{k\geq k_0+1}{\sum}\lambda_k(\theta)\widehat{\pi}_k(\theta) \Big) \widehat{F}_{+}(\theta),\quad\forall\theta\in\mathbb{T}_*\setminus \Sigma_I
\end{align*}
and conclude that:
\begin{align}\label{F-est+}
Q_{I,+,\delta}\,H_\Gamma \, Q_{I,+,\delta}\,&\geq\,\Lambda_+Q_{I,+,\delta}.
\end{align}}

	Recalling that $Q_{I,\delta}=Q_{I,-,\delta}\oplus Q_{I,+,\delta}$, 
	this means that for $\delta\in(0,\delta_0)$, the interval $\mathring{I}=(-\Lambda_-,\Lambda_+)$ belongs to the resolvent set of $P_{I,\delta}^\bot H_\Gamma P_{I,\delta}^\bot$ considered as self-adjoint operator in $P_{I,\delta}^\bot\mathcal{H}$ and this finally obeys the conditions of Definition \ref{H-band-red}.
\end{proof}

\begin{remark}\label{R-QHQ-sp-desc}The analysis in the proof of  Proposition \ref{P-H-band-0-field} shows also that for $\delta\in(0,\delta_0)$ we have:
	\begin{enumerate}
		\item The operator $Q_{I,\delta}H_\Gamma Q_{I,\delta}$ has a spectrum composed of three isolated parts:
		$$
		\sigma\big(Q_{I,\delta}H_\Gamma Q_{I,\delta}\big)=S_-\sqcup\{0\}\sqcup S_+:\quad S_-\subset[-E_0,-\Lambda_-],\ S_+\subset[\Lambda_+,+\infty)\,,
		$$
		with $Q_{I,\delta}H_\Gamma Q_{I,\delta}=0$ on $P_{I,\delta}\mathcal{H}$.
		\item The Riesz orthogonal projections associated with the above three spectral components define the orthogonal decomposition $\bb1=Q_{I,-,\delta}\oplus P_{I,\delta}\oplus Q_{I,+,\delta}$\,.
		\item Given any $E\in \mathring{I}$,  the operator $Q_{I,-,\delta}(H_\Gamma-E\bb1) Q_{I,-,\delta}$, as bounded self-adjoint operator in $Q_{I,-,\delta}L^2(\X)\,$,  has an inverse $R_{I,-,\delta}(E) \equiv\Op\big(r_{I,-,\delta}(E)\big)$, and  $Q_{I,+,\delta}(H_\Gamma-E\bb1) Q_{I,+,\delta}$, as self-adjoint operator in $Q_{I,+,\delta}L^2(\X)$ has an inverse $R_{I,+,\delta}(E)\equiv\Op\big(r_{I,+,\delta}(E)\big)\,.$ 
	\end{enumerate}
\end{remark}

\section{The magnetic \enquote{quasi-bands}}\label{S-MLB}

In this section we define the \textit{magnetic version} of our quasi-band projection in Definition \ref{D-Pdelta} and quasi-band Hamiltonian \eqref{FD-qBandHamil}. We use the same ideas as in \cite{CHP-1}, \cite{CHP-2} and \cite{CIP}. As explained in Subsection 3.4 of \cite{CIP} the magnetic quantization of the symbol of a quasi-band projection defined by a family of quasi-Wannier functions may be written as a quasi-band projection defined by a family of \textit{magnetic \enquote{quasi-Wannier functions}}. When the magnetic field is constant these functions are the \textit{Zak magnetic translations} of a given principal Wannier function (see Subsection~8.1 in \cite{CHP-2}). 

Let us emphasize once more that all the arguments and computations are done for a fixed value of the parameter $\delta\in(0,\delta_0)$ with $\delta_0>0$ given by the requirements for \eqref{F-est-gap}. All the statements are true for any value of $\delta>0$ in the given interval but no uniform dependence is assumed. Although the notations keep trace of this dependence on $\delta>0$ no longer reference to this fact will be made in the coming statements.

Given a \textit{magnetic field} $B$ with an associated \textit{vector potential} $A$ the main mathematical objects appearing in the magnetic pseudodifferential calculus (\cite{MP-1}, \cite{IMP-1}, \cite{IMP-2}, \cite{AMP}) are the circulation of the vector potential along an oriented compact interval:
\beq  \label{FD-vect-pot-circ}
\int_{[x,y]}A:=\underset{j=1,2}{\sum}(y_j-x_j)\int_0^1dtA_j\big(x+t(y-x)\big)\,,
\eeq
and the flux of the magnetic field through an oriented triangle:
\beq \label{FD-magn-flux}
\int_{<x,y,z>}B:=\underset{j,k=1,2}{\sum}(y_j-x_j)(z_k-x_k)\int_0^1dt\int_0^tds\,B_{jk}\big(x+t(y-x)+s(z-y)\big)\,.
\eeq 
Important ingredients are  their imaginary exponentials:
\beq\label{eq:1.32}
\begin{array}{ll}
	\Lambda^A(x,y) &:=\exp\Big(-i\int_{[x,y]}A\Big)\,, \\
	\Lambda^{\epsilon,\kappa}(x,y)&:=\exp\Big(-i\int_{[x,y]}A^{\epsilon,\kappa}\Big) =\exp\Big(-i\epsilon\int_{[x,y]}\big(A^\circ+\kappa A\big)\Big),\\
	\Lambda^{\epsilon}(x,y)&:=\Lambda^{\epsilon,0}(x,y) =\exp\Big(-i\,\epsilon\int_{[x,y]}A^{\circ}\Big)\,,
\end{array}
\eeq
and 
\begin{align}\label{FD-Omega}
	\Omega^B(x,y,z) &:=\exp\Big(-i\int_{<x,y,z>}B\Big)\,,\nonumber \\
	\Omega^{\epsilon,\kappa}(x,y,z)&:=\exp\Big(-i\int_{<x,y,z>}B^{\epsilon,\kappa}\Big)=\exp\Big(-i\epsilon\int_{<x,y,z>}\big(B^\circ+\kappa B\big)\Big)\,,\nonumber  \\
	\Omega^\epsilon(x,y,z)&:=\exp\Big(-i\epsilon\int_{<x,y,z>}B^\circ\Big)\,.
\end{align}

By Stokes' Theorem we have that
\beq\label{F-Stokes}
\Omega^B(x,y,z)\ =\ \Lambda^A(x,y)\,\Lambda^A(y,z)\,\Lambda^A(z,x).
\eeq
We shall use the shorthand notation:
\beq\label{shn}
\Op^{\epsilon,\kappa}:=\Op^{A^{\epsilon,\kappa}},\quad\Op^\epsilon=\Op^{A^{\epsilon}}.
\eeq 

\subsection{The magnetic quasi-Wannier functions}
\label{SS-mWfunc}

Considering the magnetic field introduced in  \eqref{Bek} and the "quasi-band" projection defined in Definition \ref{D-Pdelta} we define \textit{the magnetic \enquote{quasi-Wannier functions}} by the procedure elaborated in \cite{CHN} that we used in \cite{CHP-1} and \cite{CHP-2}. In fact this subsection is intended mainly to recall some definitions, notations and results from the Subsections 3.1 and 3.2 in \cite{CHP-1}. First let us remind that for the constant magnetic field $\epsilon B^\circ$, the magnetic "quasi-Wannier functions" are  the Zak magnetic translations considered in \cite{HS} and \cite{Ne-RMP}. 

An important ingredient in the following computations is the form of the exponential function $\Lambda^\epsilon(x,y)$. 
We notice that due to the choice of \textit{transverse gauge} that we made in \eqref{defA0}, we have that
\begin{align}\label{rem1a}
	\int_{[x,y]}A^\circ \; 
	%=(1/2)B^\circ\Big((y_2-x_2)\int_0^1dt\,\big(x_1+t(y_1-x_1)\big)-(y_1-x_1)\int_0^1dt\,\big(x_2+t(y_2-x_2)\big)\Big)\\
	%&=(1/2)B^\circ\Big((y_2-x_2)\big(x_1+(1/2)(y_1-x_1)\big)-(y_1-x_1)\big(x_2+(1/2)(y_2-x_2)\big)\Big) \\
	=(B^\circ/2)x\wedge y.
\end{align}
Thus for any $z\in\X$ we obtain that
\begin{align}\label{rem1b}
	\int_{<x+w,y+w,z+w>}B^\circ
	%=\int_{[x+w,y+w]}A^\circ +\int_{[y+w,z+w]}A^\circ +\int_{[z+w,x+w]}A^\circ  \\
	%&=(B^\circ/2)\big(x\wedge y+y\wedge z+z\wedge x\big)\\ & =\int_{[x,y]}A^\circ +\int_{[y,z]}A^\circ +\int_{[z,x]}A^\circ
	=\int_{<x,y,z>}B^\circ,\quad\forall w\in\X.
\end{align}

\begin{definition}\label{D-P-epsilon}
	With $\Lambda^{\epsilon}$ defined in \eqref{eq:1.32} and using \eqref{DF-Psi-gamma}, we define the magnetic \enquote{quasi} Wannier functions:
$$
	\mathring{\phi}^{\epsilon}_{\pm,\gamma,\delta}=\Lambda^{\epsilon}(x,
	\gamma)\Psi_{\pm,\gamma,\delta}(x)\,,
	$$
	and the magnetic quasi-band as being the closed linear span of
	$\{\mathring{\phi}^{\epsilon}_{-,\gamma,\delta},\mathring{\phi}^{\epsilon}_{+,\gamma,\delta}\}_{\gamma\in\Gamma}\,$ with associated orthogonal projection $P_{I,\delta}^{\epsilon}$; let $p^\epsilon_{I,\delta}\in\mathscr{S}^\prime(\Xi)$ be its magnetic symbol, i.e. $\Op^\epsilon\big(p^\epsilon_{I,\delta}\big)=P^\epsilon_{I,\delta}$. We also denote by $Q^\epsilon_{I,\delta}:=\bb1-P^\epsilon_{I,\delta}$.
\end{definition}
Let us recall some properties of the \textit{Zak magnetic translations} in a constant magnetic field (see \cite{HS}, \cite{Ne-RMP}, \cite{CHN} and  Proposition 8.1 in \cite{CHP-2}).
\begin{proposition}\label{P-Zak-magn-transl} 
	The family of unitary operators
	$\big\{
	\mathcal{T}^\epsilon_\gamma:=\Lambda^\epsilon(\cdot,\gamma)\,\tau_{\gamma}\big\}_{\gamma\in\Gamma}$
	satisfies the following properties:
	\begin{enumerate}
		\item $\mathcal{T}^\epsilon_\alpha\mathcal{T}^\epsilon_\beta=
		\Lambda^\epsilon(\beta,\alpha)\mathcal{T}^\epsilon_{\alpha+\beta}$\,.
		\item The tempered distribution $F\in\mathscr{S}^\prime(\Xi)$ is $\Gamma$-periodic with respect to the variable in $\X$ if and only if the following commutation relations hold true for any $\gamma\in\Gamma$: $\Op^\epsilon(F)\mathcal{T}^\epsilon_\gamma=\mathcal{T}^\epsilon_\gamma\Op^\epsilon(F)$.
	\end{enumerate}
\end{proposition}

\begin{remark}
Let us consider the family of scalar products 
$\mathbb{G}^{\epsilon,\delta}_{(j,\alpha),(k,\beta)}\,:=\,
\langle\mathring{\phi}^{\epsilon}_{j,\alpha,\delta},
\mathring{\phi}^{\epsilon}_{k,\beta,\delta}\rangle_{L^2(\X)}$
indexed by the set of indices $\big(\{-,+\}\times\Gamma\big)\times\big(\{-,+\}\times\Gamma
\big)\,$. We notice that
 \begin{align*}
\langle\mathring{\phi}^{\epsilon}_{j,\alpha,\delta},
\mathring{\phi}^{\epsilon}_{k,\beta,\delta}\rangle_{L^2(\X)}&=\langle\tau_\alpha
\Psi_{j,\delta},\Lambda^{\epsilon}(\alpha,\cdot) \Lambda^{\epsilon}(\cdot,\beta)\tau_\beta\Psi_{k,\delta}\rangle_{L^2(\X)}\nonumber \\ &=
\Lambda^{\epsilon}(\alpha,\beta)\langle\tau_\alpha
\Psi_{j,\delta},\Omega^{\epsilon}(\beta,\alpha,\cdot)\tau_\beta\Psi_{k,\delta}\rangle_{L^2(\X)}.
\end{align*} 
\end{remark}

Due to the rapid decay of $\Psi_{\pm,\delta}$, we obtain that (see the construction of the Wannier functions through the magnetic translations in \cite{CHN} (Lemmas~3.1 and 3.2) and Lemma 3.15 in \cite{CIP}):
\begin{proposition}\label{P-Gepsilon}
	The matrix $\mathbb{G}^{\epsilon,\delta}$ defines a positive bounded
	operator on $\ell^2(\Gamma)\otimes\mathbb{C}^2$ and  
	$$
	\mathbb{G}^{\epsilon,\delta}_{(j,\alpha),(k,\alpha)}=\delta_{jk}\,,\quad\forall\alpha\in\Gamma\,,
	$$
	Moreover, for any $m\in\mathbb{N}$, there exists $C_m >0$ such that for any $(j,k)\in\{-,+\}\times\{-,+\}$:
	$$
	\underset{(\alpha,\beta)\in\Gamma\times \Gamma}{\sup}\
	<\alpha-\beta>^m\big| \mathbb{G}^{\epsilon,\delta}_{(j,\alpha),(k,\beta)}\,-\,\delta_{\alpha,\beta}\delta_{jk}\big| \,\leq\,
	C_m\, \epsilon\quad\forall\epsilon>0.
	$$ 
\end{proposition}\label{P-prop-G-epsilon}
\begin{definition}\label{D-F-epsilon}
For some small enough $\epsilon_0>0$ (in order to have invertibility) and for any $\epsilon \in [0,\epsilon_0]$, we can define $\mathbb{F}^{\epsilon,\delta}\,:=\big(\mathbb{G}^
{\epsilon,\delta}\big)^{-1/2}$ and the magnetic 'quasi-Wannier' functions:
\beq\label{def-q-Wannier}
\phi^{\epsilon}_{k,\gamma,\delta}\,:=\,\underset{\alpha\in\Gamma}{\sum}\ 
\underset{j\in\{-,+\}}{\sum}\, \mathbb{F
}^{\epsilon,\delta}_{(j,\alpha),(k,\gamma)}\, \mathring{\phi}^{\epsilon}_{j,\alpha,\delta}\,\in\,\mathscr{S}(\X)\,,\quad   \forall (k,\gamma) \in \{-,+\}\times\Gamma \,.
\eeq
\end{definition}
\begin{remark}\label{R-dep-eps0-delta}
We emphasize that the value of $\epsilon_0>0$ depends on our choice for $\delta\in(0,\delta_0)$. As we shall not vary this fixed value of $\delta\in(0,\delta_0)$ we shall not need any control on this dependence.
\end{remark}
The magnetic 'quasi-Wannier' functions $\big\{\phi^{\epsilon}_{k,\gamma,\delta}\,,\,k=\pm,\,\gamma\in\Gamma\big\}$ form an orthonormal basis of $P^{\epsilon}_{I,\delta} \, L^2(\X)$.

In order to compare with the results in \cite{CHP-2} we shall consider $\mathbb{G}^{\epsilon}$ and $\mathbb{F}^{\epsilon}$ as infinite matrices indexed by $\Gamma\times\Gamma$ and having entries $2\times2$ complex matrices.

\begin{proposition}\label{P-prop-F-epsilon}
	For $\epsilon\in[0,\epsilon_0]$ with $\epsilon_0>0$ fixed in Definition \ref{D-F-epsilon}, $\mathbb{F}^{\epsilon,\delta}$ has the following properties (see the construction of the Wannier functions through the magnetic translations in \cite{CHN} (cf Lemmas~3.1 and 3.2)):
	\begin{enumerate}
		\item  $\mathbb{F}^{\epsilon,\delta}\in\,\mathbb B \big(\ell^2(\Gamma)\otimes\mathbb{C}^2\big) \cap \mathbb B \big(\ell^\infty(\Gamma)\otimes\mathbb{C}^2\big) \,$.
		\item For any $m\in\mathbb{N}$, there exists $C_m >0$ such that
		\beq\label{decay-F}
		\underset{(\alpha,\beta)\in\Gamma\times \Gamma}{\sup}\
		<\alpha-\beta>^m\big| \mathbb{F}^{\epsilon,\delta}_{(j,\alpha),(k,\beta)}\,-\,\delta_{\alpha,\beta}\delta_{jk}\big| \,\leq\,
		C_m\, \epsilon\,.
		\eeq
	\end{enumerate}
\end{proposition}

Moreover, due to Proposition \ref{P-Zak-magn-transl}, we have the following result (see Subsection 3.2  in \cite{CHN}):
\begin{proposition}\label{P-magn-W-function}~
	For $\epsilon\in[0,\epsilon_0]$ with $\epsilon_0>0$ fixed in Definition \ref{D-F-epsilon}, we have that
	\begin{enumerate}
		\item There exists a rapidly decaying function
		$\mathbf{F}^\epsilon_\delta :\Gamma\rightarrow\mathbb{B}\big(\mathbb{C}^2\big)$ such that for any pair
		$(\alpha,\beta)\in\Gamma\times\Gamma$ we have:
		$$
		\mathbb{F}^{\epsilon,\delta}_{\alpha,\beta}=  \Lambda^\epsilon(\alpha,\beta)\,
		\mathbf{F}^\epsilon_\delta(\alpha-\beta)\,.
		$$
		\item With $\psi^\epsilon_{\pm,0,\delta}\in\mathscr S(\mathbb R^2)$defined by
		\beq\label{F-psi-e0}
		\psi^\epsilon_{\pm,0,\delta}(x)=\underset{\alpha\in\Gamma}{\sum}\underset{j=\pm}{\sum}\mathbf{F}^\epsilon_\delta (\alpha)_{j\pm}\,
		\Omega^\epsilon(\alpha,0,x)\, \Psi_{j,\alpha,\delta}(x)\,,
		\eeq
		we have 
		\beq\label{F-psi-epsilon-0}
		\phi^\epsilon_{j,\gamma,\delta}\,=\,\Lambda^\epsilon(\cdot ,\gamma)(\tau_{\gamma}
		\psi^\epsilon_{j,0,\delta})\,,\qquad\forall\gamma\in\Gamma\,.
		\eeq
		\item 
		For  any $m\in\mathbb{N}\,$,
		$\alpha\in\mathbb{N}^2$, there exists  $C_{m,\alpha} >0$ such that 
		$$
		<x>^m\left|[\partial_x^\alpha(\psi^\epsilon_{\pm,0,\delta}-\Psi_{\pm,0,\delta})](x)\right|\leq\,C_{m,\alpha}\, \epsilon\,,\quad\forall x\in \mathbb R^2.
		$$
	\end{enumerate}
\end{proposition}
From \eqref{F-psi-epsilon-0} we conclude that
\begin{equation*}
\big\langle\psi^\epsilon_{j,0,\delta}\,,\,\psi^\epsilon_{k,0,\delta}\big\rangle_{L^2(\X)}=\big\langle\phi^\epsilon_{j,0,\delta}\,,\,\phi^\epsilon_{k,0,\delta}\big\rangle_{L^2(\X)}=\delta_{jk}\,.
\end{equation*}
Note also that the above defined magnetic "quasi-Wannier" functions belong to $\mathscr{S}(\X)$ (details may be found in Subsection 3.1 of \cite{CHP-1}).

We can write, with the series converging in the strong operator topology:
	\beq\label{FD-pepsilonI}
	P^\epsilon_{I,\delta}\,=\,\underset{\gamma\in\Gamma}{\sum}\,\underset{j=\pm}{\sum}\,\big|\phi^\epsilon_{j,\gamma,\delta}\big\rangle\big\langle\phi^\epsilon_{j,\gamma,\delta}\big|
	\eeq
In fact, due to the estimations proved in Proposition \ref{P-magn-W-function} and the fact that the quasi-Wannier functions belong to $\mathscr{S}(\X)$ uniformly with respect to $\gamma\in\Gamma$, we can write \eqref{FD-pepsilonI} as
\beq\label{F-desc-PepsilonI}
P^\epsilon_{I,\delta}\,=\,\underset{\gamma\in\Gamma}{\sum}\Pi^\epsilon_{\gamma,\delta},\quad\Pi^\epsilon_{\gamma,\delta}:=\underset{j=\pm}{\sum}\,\big|\phi^\epsilon_{j,\gamma,\delta}\big\rangle\big\langle\phi^\epsilon_{j,\gamma,\delta}\big|
\eeq
with the series converging in operator norm, as one can prove using the Cotlar-Stein procedure (see Lemma 18.6.5 in \cite{H-3}). In fact we prove that the integral kernels converge uniformly.

\begin{proposition}\label{P-PepsilonI-per} For $\epsilon\in[0,\epsilon_0]$ with $\epsilon_0>0$ fixed in Definition \ref{D-F-epsilon}, the symbol $p^\epsilon_{I,\delta}$ belongs to  $S^{-\infty}(\Xi)$ and  is $\Gamma$-periodic.
\end{proposition}
\begin{proof}
	By definition \ref{D-P-epsilon} and \eqref{FD-pepsilonI},  we can write
	$$
	P^\epsilon_{I,\delta}=\Op^\epsilon\big(p^\epsilon_{I,\delta}\big)=
	\Int\big(\mathfrak{K}_{P^\epsilon_{I,\delta}}\big)
	$$
	with the integral kernel
	$$
	\mathfrak{K}_{P^\epsilon_{I,\delta}}(x,y)=\underset{\gamma\in\Gamma}{\sum}\Lambda^\epsilon(x,\gamma)\psi^\epsilon_{j,0,\delta}(x-\gamma)\overline{\Lambda^\epsilon(y,\gamma)}\,\overline{\psi^\epsilon_{k,0,\delta}(y-\gamma)}.
	$$
	The fact that the magnetic quasi-Wannier functions belong to $\mathscr{S}(\X)$ clearly implies that the magnetic symbol associated to this kernel is of class $S^{-\infty}(\Xi)$. 
	
	Following the second conclusion in Proposition \ref{P-Zak-magn-transl} let us compute the commutator
	\begin{align*}
\big(\big[\mathcal{T}^\epsilon_\alpha\big]^{-1}\,\Int(\mathfrak{K}_{P^\epsilon_{I,\delta}})\,\mathcal{T}^\epsilon_\alpha\Phi\big)(x)&=\Lambda^\epsilon(\alpha,x+\alpha)\big(\Int(\mathfrak{K}_{P^\epsilon_{I,\delta}})\,\mathcal{T}^\epsilon_\alpha\Phi\big)(x+\alpha) \\ \nonumber
&=\Lambda^\epsilon(\alpha,x+\alpha)\int_\X\,dz\,\mathfrak{K}_{P^\epsilon_{I,\delta}}(x+\alpha,z)\big(\mathcal{T}^\epsilon_\alpha\Phi\big)(z) \\
&=\Lambda^\epsilon(\alpha,x+\alpha)\int_\X\,dz\,\mathfrak{K}_{P^\epsilon_{I,\delta}}(x+\alpha,z)\Lambda^\epsilon(z,\alpha)\Phi(z-\alpha) \\
&=\Lambda^\epsilon(\alpha,x+\alpha)\int_\X\,dy\,\mathfrak{K}_{P^\epsilon_{I,\delta}}(x+\alpha,y+\alpha)\Lambda^\epsilon(y+\alpha,\alpha)\Phi(y)
	\end{align*}

	Using \eqref{rem1a}-\eqref{rem1b},  the definition of $\Omega^\epsilon$ in \eqref{FD-Omega} as well as its property \eqref{F-Stokes}, we obtain:
	\begin{align*}
	\forall\alpha\in\Gamma:\quad&\mathfrak{K}_{P^\epsilon_{I,\delta}}(x+\alpha,y+\alpha)\\
		&\hspace*{14pt}=\underset{\gamma\in\Gamma}{\sum}\Lambda^\epsilon(x+\alpha,\gamma)\psi^\epsilon_{j,0,\delta}(x+\alpha-\gamma)\overline{\Lambda^\epsilon(y+\alpha,\gamma)}\,\overline{\psi^\epsilon_{k,0,\delta}(y+\alpha-\gamma)} \\
		&\hspace*{14pt}=\underset{\gamma'\in\Gamma}{\sum}\Lambda^\epsilon(x+\alpha,\gamma'+\alpha)\psi^\epsilon_{j,0,\delta}(x-\gamma')\overline{\Lambda^\epsilon(y+\alpha,\gamma'+\alpha)}\,\overline{\psi^\epsilon_{k,0,\delta}(y-\gamma')}\\
		&\hspace*{14pt}=\Lambda^\epsilon(x+\alpha,y+\alpha)\underset{\gamma\in\Gamma}{\sum}\Omega^\epsilon(y+\alpha,x+\alpha,\gamma+\alpha)\psi^\epsilon_{j,0,\delta}(x-\gamma)\overline{\psi^\epsilon_{k,0,\delta}(y-\gamma)}\\
		& \hspace*{14pt} =\Lambda^\epsilon(x+\alpha,y+\alpha)
		\underset{\gamma\in\Gamma}{\sum}\Omega^\epsilon(y,x,\gamma)\psi^\epsilon_{j,0,\delta}(x-\gamma)\overline{\psi^\epsilon_{k,0,\delta}(y-\gamma)}\\
		&\hspace*{14pt}=\Lambda^\epsilon(x+\alpha,y+\alpha)\overline{\Lambda^\epsilon(x,y)}\,\underset{\gamma\in\Gamma}{\sum}\Lambda^\epsilon(x,\gamma)\psi^\epsilon_{j,0,\delta}(x-\gamma)\overline{\Lambda^\epsilon(y,\gamma)}\,\overline{\psi^\epsilon_{k,0,\delta}(y-\gamma)}\\  
		& \;\quad=\Lambda^\epsilon(x+\alpha,y+\alpha)\overline{\Lambda^\epsilon(x,y)}\,\mathfrak{K}_{P^\epsilon_{I,\delta}}(x,y).
	\end{align*}
	Finally, using \eqref{rem1a} we get
	\begin{align*}
\Lambda^\epsilon(\alpha,x+\alpha)&\Lambda^\epsilon(x+\alpha,y+\alpha)\overline{\Lambda^\epsilon(x,y)}\Lambda^\epsilon(y+\alpha,\alpha) \\
&=\exp\big\{i(B^\circ\epsilon/2)\}\big(\alpha\wedge x+(x+\alpha)\wedge(y+\alpha)+y\wedge x+y\wedge\alpha\big)\big\}=1.
	\end{align*}
\end{proof}

Here we recall the following important estimate proved in Subsection 3.2 of \cite{CHP-1} (see Formula (3.13) therein).

\begin{proposition}\label{F-est-dif-bdproj} 
	There exists $\epsilon_0\in(0,1]$, fixed to satisfy the condition in Definition \ref{D-F-epsilon} such that for any semi-norm $\widetilde{\nu}:S^{-\infty}(\X\times\X^*)_{2\times2}\rightarrow\mathbb{R}_+$ defining the topology of $S^{-\infty}(\X\times\X^*)_{2\times2}$ there exists some constant $C(\widetilde{\nu})>0$ such that 
	$$\widetilde{\nu}\big(p^{\epsilon}_{I,\delta}-p_{I,\delta}\big)\,\leq\,C(\widetilde{\nu})\, \epsilon\,,\, \forall \epsilon \in [0,\epsilon_0],$$
	with $p_{I,\delta}$ defined in Proposition \ref{P-Symb-PI} and $p^\epsilon_{I,\delta}$ defined in \eqref{FD-pepsilonI}.
\end{proposition}

\subsection{The magnetic quasi-band Hamiltonian.}
\label{SS-magn-qband-H}

\begin{definition}
	We call  \textit{magnetic quasi-band Hamiltonian} associated to the spectral interval $I$, the operator $P_{I,\delta}^{\epsilon}H^{\epsilon}_\Gamma P_{I,\delta}^{\epsilon}$ (with  $H^{\epsilon}_\Gamma$ introduced in \eqref{FD-Hepsilon}). 
\end{definition}

The estimation in Proposition \ref{F-est-dif-bdproj} and the properties of the smooth global sections $\Psi_{\pm,\delta}:\mathbb{T}_*\rightarrow\mathscr{F}$ imply the following statement (see also the arguments in Subsection 3.2 of \cite{CHP-1}).

\begin{proposition}
	There exists $\epsilon_0>0$, fixed to satisfy the condition in Definition \ref{D-F-epsilon}, such that for any $\epsilon\in[0,\epsilon_0]$, the range of $P_{I,\delta}^{\epsilon}$ belongs to the domain of $H^{\epsilon}_\Gamma$.
\end{proposition}

We intend to compare the spectrum of $H^{\epsilon}_\Gamma$ in the interval $I\subset\mathbb{R}$ with the spectrum of the \textit{magnetic quasi-band Hamiltonian} $P^{\epsilon}_IH^{\epsilon}_\Gamma P^{\epsilon}_I$ in the same interval. Working with quasi-Wannier functions instead of true Wannier functions has as a consequence that although the product $H^{\epsilon}_\Gamma P^{\epsilon}_I$ is bounded, the norm of $(\bb1 -P^{\epsilon}_I)H^{\epsilon}_\Gamma P^{\epsilon}_I$  is not small of order $\epsilon$. Nevertheless, the modified Feshbach-Schur procedure elaborated in Appendix \ref{A-FS-arg} {will allow us to compare the spectrum of  $H^{\epsilon}_\Gamma$ in a neighborhood of $0$  with the spectrum of a \textit{\enquote{dressed} modified magnetic quasi-band Hamiltonian} of the type \eqref{F-Htilde}:
\begin{align}\label{F-Htilde-magn}
&\widetilde{H}^{\epsilon}_{I,\delta}\ \\
\nonumber &:=\ (Y_\delta^{\epsilon})^{-1/2}\left(P^{\epsilon}_{I,\delta}H^{\epsilon}_\Gamma P^{\epsilon}_{I,\delta}\,-\,P^{\epsilon}_{I,\delta} H^{\epsilon}_\Gamma\big\{P^{\epsilon}_{I,\delta}\big\}^\bot R^{\epsilon}_\bot(0)\big\{P^{\epsilon}_{I,\delta}\big\}^\bot H^{\epsilon}_\Gamma P^{\epsilon}_{I,\delta}\right)(Y^{\epsilon}_\delta)^{-1/2}\ \in\ \mathbb{B}\big(P^{\epsilon}_{I,\delta}L^2(\X)\big),
\end{align}
with 
\beq\label{F-Y}
Y^{\epsilon}_\delta\ :=\ P^{\epsilon}_{I,\delta}\,+\,P^{\epsilon}_{I,\delta} H^{\epsilon}_\Gamma\big\{P^{\epsilon}_{I,\delta}\big\}^\bot R^{\epsilon}_\bot(0)^2\big\{P^{\epsilon}_{I,\delta}\big\}^\bot H^{\epsilon}_\Gamma P^{\epsilon}_{I,\delta}\,.
\eeq
In Paragraph \ref{SSS-out-cr-est} we shall prove that some local estimate valid on a neighborhood of $\theta_0$  (see Paragraph \ref{SSS-cr-est}) is in fact sufficient for our analysis. In this  subsection we verify the admissibility of the triple $\big(H^\epsilon_\Gamma,P^\epsilon_{I,\delta},\mathring{I}\big)$ (see Definition \ref{H-band-red}) and apply Proposition \ref{P-FS-arg}  in order to estimate the \enquote{distance} between the parts of the spectra of $H^\epsilon_\Gamma$ and $ \widetilde{H}^\epsilon_{I,\delta}$ contained in the interval $I$.

We begin by proving a \enquote{magnetic version} of Proposition \ref{P-H-band-0-field}. Our proof makes use of the properties of the magnetic pseudodifferential calculus as developed in \cite{IMP-1} and \cite{IMP-2} and briefly summarized in the Appendix B of \cite{CHP-1}.
We recall the notation  $\Op^\epsilon$ introduced in \eqref{shn} (see also \eqref{D-MPsiDO}) and the magnetic Moyal product  $\sharp^\epsilon$ defined by \eqref{DF-magn-MoyalPr}.
For the convenience of the reader, let us recall Proposition B.14 from \cite{CHP-1} that we shall use several times in our arguments.
	\begin{proposition}\label{propB.14}
		For any pair $(p,s)\in\mathbb{R}\times\mathbb{R},$ any $\rho\in[0,1]$ and any $\epsilon_0\in[0,1]$ there exists a bilinear continuous map $\mathfrak{z}_\epsilon:S^p_\rho(\X\times\X^*)\times S^s_\rho(\X\times\X^*)\ni(F,G)\mapsto F\sharp^\epsilon G\in S^{s+p-2}_\rho(\X\times\X^*)$ uniformly in $\epsilon\in[0,\epsilon_0]$ such that
		$$ 
		F\sharp^\epsilon G\,=\,F\sharp G\,+\,\epsilon\, \mathfrak{z}_\epsilon(F,G).
		$$
	\end{proposition}

\begin{proposition}\label{P-Hyp-H-magn}
	Recalling Definition \ref{D-P-epsilon} and \eqref{DF-I-delta}, 
	for any interval $I_\circ\subset \mathring{I}$ containing $0$ in its interior there exists some $\epsilon_0>0$, such that the triple $\big(H^{\epsilon}_\Gamma,P^{\epsilon}_{I,\delta},I_\circ\big)$ is admissible for any $\epsilon\in[0,\epsilon_0]$.
\end{proposition}

\begin{proof}
	With Definition \ref{H-band-red} in mind, we have to find a non-trivial interval $I_\circ$ as in the statement above, that is contained in the resolvent set of $Q^\epsilon_{I,\delta}H^\epsilon_\Gamma Q^\epsilon_{I,\delta}$.  We intend to use the conclusion of Proposition \ref{P-H-band-0-field} giving a spectral gap for the 0 field Hamiltonian and 
	proceed as in \cite{CHP-2} using the results in \cite{AMP} or \cite{CP-1} concerning the continuity of the spectrum with respect to the magnetic field.
	Denoting by $q^{\epsilon}_{I,\delta}$ the magnetic symbol of $Q^\epsilon_{I,\delta}$ and by $q_{I,\delta}$ the Weyl symbol of $Q_{I,\delta}$ we use Propositions \ref{propB.14}
		and \ref{F-est-dif-bdproj} in order to obtain the following estimate:
	\begin{equation}\label{F-pert}
	\begin{array}{ll}
	&Q_{I,\delta}^{\epsilon}H^{\epsilon}_\Gamma Q_{I,\delta}^{\epsilon}=\Op^{\epsilon}\big(q^{\epsilon}_{I,\delta}
	\,\sharp^\epsilon\,h\,\sharp^\epsilon\, q^{\epsilon}_{I,\delta}\big)=\Op^{\epsilon}\big((1-p^{\epsilon}_{I,\delta})\sharp^\epsilon h\sharp^\epsilon(1-p^{\epsilon}_{I,\delta})\big)\\
	&=\Op^{\epsilon}(h)-\Op^{\epsilon}\big(p^{\epsilon}_{I,\delta}\sharp^\epsilon h\sharp^\epsilon(1-p^{\epsilon}_{I,\delta})\big)-\Op^{\epsilon}\big(h\sharp^\epsilon p^{\epsilon}_{I,\delta}\big)\\
	&=\Op^{\epsilon}(h)-\Op^{\epsilon}\big(p_{I,\delta}\sharp h\sharp(1-p_{I,\delta})\big)-\Op^{\epsilon}\big(h\sharp p_{I,\delta}\big)\,+\,\mathcal{O}(\epsilon)\\
	&=\Op^{\epsilon}\big(q_{I,\delta}
	\,\sharp\,h\,\sharp\,q_{I,\delta}\big)\,+\,\mathcal{O}(\epsilon),
	\end{array} \end{equation}
with $\mathcal{O}(\epsilon)\in\mathbb{B}\big(L^2(\X)\big)$.

\noindent\textbf{Step 1:}
	We begin by noticing that Propositions \ref{propB.14}
	and \ref{F-est-dif-bdproj} allow us to estimate the difference $q^\epsilon_{I,\delta}-q_{I,\delta}=(1-p^\epsilon_{I,\delta})-(1-p_{I,\delta})=p_{I,\delta}-p^\epsilon_{I,\delta}$ as an $\mathcal{O}(\epsilon)\in\mathbb{B}\big(L^2(\X)\big)$. Moreover 
	from Remark \ref{R-ordSymbOpm} we know that $q_{I,-,\delta}\in S^{-\infty}(\X\times\X^*)$ and  $q_{I,+,\delta}\in S^0_1(\X\times\X^*)$. 
	
	We  recall Definition \ref{D-Qmp} and Remark \ref{R-desc-QI-Epm}, denote by $q_{I,\pm,\delta}$ the Weyl symbols of the orthogonal projections $Q_{I,\pm,\delta}$ and notice that
	\beq\label{F-desc-symb}
	\Op\big(q_{I,\delta}
	\,\sharp\,h\,\sharp\,q_{I,\delta}\big)=
	\Op\big(q_{I,-,\delta}\,\sharp\,h\,\sharp\,q_{I,-,\delta}\big)+\Op\big(q_{I,+,\delta}\,\sharp\,h\,\sharp\,q_{I,+,\delta}\big)\,.
	\eeq
	Hence, using Proposition \ref{propB.14}, the decomposition \eqref{F-desc-symb}, the fact that $Q_{I,\pm,\delta}=\Op\big(q_{I,\pm,\delta}\big)$ are projections and finally Proposition \ref{propB.14} once again, we can write
\begin{align}\nonumber
	q_{I,-,\delta}\,\sharp\,h\,\sharp\,q_{I,-,\delta}&=q_{I,-,\delta}\,\sharp\,q_{I,-,\delta}\,\sharp\,h\,\sharp\,q_{I,-,\delta}\,\sharp\,q_{I,-,\delta}=q_{I,-,\delta}\,\sharp\,(q_{I,-,\delta}\,\sharp\,h\,\sharp\,q_{I,-,\delta})\,\sharp\,q_{I,-,\delta}\\ \nonumber
	&=q_{I,-,\delta}\,\sharp^\epsilon\,(q_{I,-,\delta}\,\sharp\,h\,\sharp\,q_{I,-,\delta})\,\sharp^\epsilon\,q_{I,-,\delta}\,+\,\mathcal{O}(\epsilon),\\ \label{F-magn--}
\Op^\epsilon\big(q_{I,-,\delta}\,\sharp\,h\,\sharp\,q_{I,-,\delta}\big)&=\Op^\epsilon\big(q_{I,-,\delta}\big)\Op^\epsilon\big(q_{I,-,\delta}\,\sharp\,h\,\sharp\,q_{I,-,\delta}\big)\Op^\epsilon\big(q_{I,-,\delta}\big)\,+\,\mathcal{O}(\epsilon),
	\end{align}
\begin{align}\nonumber
q_{I,+,\delta}\,\sharp\,h\,\sharp\,q_{I,+,\delta}&=q_{I,+,\delta}\,\sharp\,q_{I,+,\delta}\,\sharp\,h\,\sharp\,q_{I,+,\delta}\,\sharp\,q_{I,+,\delta}\\ \nonumber
&=q_{I,+,\delta}\,(1-p_{I,\delta}-q_{I,-,\delta})\,\sharp\,h\,\sharp\,(1-p_{I,\delta}-q_{I,-,\delta})\,\sharp\,q_{I,+,\delta}\\ \nonumber
&=q_{I,+,\delta}\sharp h\sharp q_{I,+,\delta}-q_{I,+,\delta}\sharp\big[(p_{I,\delta}+q_{I,-,\delta})\sharp\,h\,\sharp\,(1-p_{I,\delta}-q_{I,-,\delta})\big]\sharp q_{I,+,\delta}\\ \nonumber
&\hspace*{2,65cm}-q_{I,+,\delta}\sharp\big[h\sharp(p_{I,\delta}+q_{I,-,\delta})\big]\sharp q_{I,+,\delta}
\end{align} 
\begin{align}
\nonumber
&=h-(p_{I,\delta}+q_{I,-,\delta})\sharp h\sharp(1-p_{I,\delta}-q_{I,-,\delta})-h\sharp(p_{I,\delta}+q_{I,-,\delta})\\ \nonumber
&\hspace*{2,65cm}-q_{I,+,\delta}\sharp\big[(p_{I,\delta}+q_{I,-,\delta})\sharp\,h\,\sharp\,(1-p_{I,\delta}-q_{I,-,\delta})\big]\sharp q_{I,+,\delta}\\ \nonumber
&\hspace*{2,65cm}-q_{I,+,\delta}\sharp\big[h\sharp(p_{I,\delta}+q_{I,-,\delta})\big]\sharp q_{I,+,\delta}\\ \nonumber
&=q_{I,+,\delta}\sharp^\epsilon h\sharp^\epsilon q_{I,+,\delta}-q_{I,+,\delta}\sharp^\epsilon\big[(p_{I,\delta}+q_{I,-,\delta})\sharp\,h\,\sharp\,(1-p_{I,\delta}-q_{I,-,\delta})\big]\sharp^\epsilon q_{I,+,\delta}\\ \nonumber
&\hspace*{2,65cm}-q_{I,+,\delta}\sharp^\epsilon\big[h\sharp(p_{I,\delta}+q_{I,-,\delta})\big]\sharp^\epsilon q_{I,+,\delta}\ +\ \mathcal{O}(\epsilon),
\end{align}
\begin{align}\label{F-mag-+}
\Op^\epsilon\big(q_{I,+,\delta}\,\sharp\,h\,\sharp\,q_{I,+,\delta}\big)&=\Op^\epsilon\big(q_{I,+,\delta}\big)\Op^\epsilon\big(q_{I,+,\delta}\,\sharp\,h\,\sharp\,q_{I,+,\delta}\big)\Op^\epsilon\big(q_{I,+,\delta}\big)\,+\,\mathcal{O}(\epsilon),
\end{align}
with $\mathcal{O}(\epsilon)\in\mathbb{B}\big(L^2(\X)\big)$.
	
\noindent\textbf{Step 2:} 
	For any $E\in \mathring{I}$, taking into account that we keep $\delta\in(0,\delta_0)$ fixed, we shall eliminate the reference to $\delta$ in the notation for the resolvents defined in point (3) of Remark \ref{R-QHQ-sp-desc} and their symbols and we shall also consider the magnetic operators defined by these symbols:
	\beq\label{D-RIepsilon}
	R^\epsilon_{I,\pm}(E)\,:=\,\Op^\epsilon\big(r_{I,\pm}(E)\big),\qquad R^\epsilon_I(E)\,:=\,R^\epsilon_{I,-}(E)+R^\epsilon_{I,+}(E)\,.
	\eeq
	For 0 magnetic field, we conclude from Remark \ref{R-QHQ-sp-desc} (3) that 
	$$
	R_{I,\pm}(E)=Q_{I,\pm,\delta}\,R_{I,\pm}(E)=R_{I,\pm}(E)\, Q_{I,\pm,\delta}\,,
	$$
	i.e. at the level of the symbols we have the equalities:
	\begin{equation}\label{F-band-res}
	r_{I,\pm}(E)= q_{I,\pm,\delta}\,\sharp\, r_{I,\pm}(E)= r_{I,\pm}(E)\,\sharp\, q_{I,\pm,\delta}\,.
	\end{equation}
	In order to use these equalities in our next estimation \eqref{F-epsilon-ortog} we have to study the regularity of the distributions $r_{I,\pm}(E)$. 
	We shall first consider the symbol $r_{I,-}(E)$ for the resolvent in $E$ of the bounded self-adjoint operator $$Q_{I,-,\delta}H_\Gamma Q_{I,-\delta}=\Op\big(q_{I,-,\delta}\sharp h\sharp q_{I,-,\delta}\big)$$ 
	acting in $Q_{I,-,\delta}\mathcal{H}$. 
	We know that $h\in S^2_1(\X\times\X^*)$ and $q_{I,-,\delta}\in S^{-\infty}(\X\times\X^*)$ so that $q_{I,-,\delta}\sharp h\sharp q_{I,-,\delta}\in S^{-\infty}(\X\times\X^*)\subset S^0_1(\X\times\X^*)$ and by Proposition 6.1 in \cite{IMP-2} we deduce that $r_{I,-}(E)\in S^0_1(\X\times\X^*)$. But this together with \eqref{F-band-res} imply that $r_{I,-}(E)\in  S^{-\infty}(\X\times\X^*)$.
	\medskip

\noindent\textbf{Step 3:}
	We shall prove the following fact that replaces Proposition 7.1 in \cite{CHP-2}:
\beq\label{F-symb-RI+} 
		r_{I,+}(E)\,\in\,S^{-2}_1(\X\times\X^*).
\eeq
		We intend to apply Proposition 6.3 from \cite{IMP-2} giving the class of the symbol of the inverse of an invertible  magnetic pseudodifferential operator.
		Let us recall that $ q_{I,+,\delta}\in S^0_1(\X\times\X^*)$  (see Remark \ref{R-ordSymbOpm}) and $h\in S^2_1(\X\times\X^*)$.
		Thus, if we recall that the order of the Moyal product of two H\"{o}rmander type symbols is the sum of their orders, we deduce that the operator $H_{\Gamma,+}:=Q_{I,+,\delta}H_\Gamma Q_{I,+,\delta}$ has a symbol 
		$$
		h_+\,=\,q_{I,+,\delta}\,\sharp\, h\,\sharp\, q_{I,+,\delta}\in S^2_1(\X\times\X^*).
		$$
		On the other hand, we may consider the operators $H_{\Gamma,+}$ and $R_{I,+}(E):=\big(H_{\Gamma,+}-E\bb1\big)^{-1}$ as elements of $ \mathbb{B}\big(Q_{I,+,\delta}L^2(\X)\big)$ and then we have the identity $ \big(H_{\Gamma,+}-E\bb1_{Q_{I,+,\delta}L^2(\X)}\big)R_{I,+}(E)=\bb1_{Q_{I,+,\delta}L^2(\X)}$. 
		If we denote by
		$$
		Q_{I,\leq}\,:= P_I\oplus Q_{I,-,\delta}\,:= \,\Op\big(q_{I,\leq}\big) 
		$$
		and using Proposition \ref{P-Symb-PI} and Remark \ref{R-ordSymbOpm} we deduce that $q_{I,\leq}\in S^{-\infty}(\X\times\X^*)$ and we obtain the following relation in $\mathbb{B}\big(L^2(\X)\big)$:
		$$
		\forall E\in \mathring{I},\qquad \Op\big(r_{I,+}(E)\big)\,=\,Q_{I,+, \delta}\,R_{I,+}(E)\, Q_{I,+,\delta}\,=\,Q_{I,+, \delta}\Big(Q_{I,\leq}+\big(H_{\Gamma,+}-EQ_{I,+, \delta}\big)\Big)^{-1}Q_{I,+,\delta}.
		$$
Let us consider the operator 
			$$
			\Big(Q_{I,\leq}+\big(H_{\Gamma,+}-EQ_{I,+, \delta}\big)\Big)^{-1}=
			Q_{I,\leq}+Q_{I,+,\delta}\,R_{I,+}(E)\, Q_{I,+, \delta}\in\mathbb{B}\big(L^2(\X)\big).
			$$
			Coming back to the symbols of these pseudodifferential operators, we would like to prove that
			\beq\label{mai1}
			\widetilde{r}_{I,+}(E) =q_{I,\leq}+r_{I,+}(E)\,\in S_1^{-2}(\X\times\X^*)
			\eeq
			where $\widetilde{r}_{I,+}(E) $ denotes the symbol of  the bounded operator $\Big(Q_{I,\leq}+\big(H_{\Gamma,+}-EQ_{I,+}\big)\Big)^{-1}$. 
			In order to control the regularity of this symbol  we shall use a standard procedure based on the Beals criterion. The original idea of the proof can be found in Section 3 of \cite{Bea-77}, but we shall make use here of the magnetic version of the result appearing in Subsection 6.1 of \cite{IMP-2}, Proposition 6.3. Let us include it here for the convenience of the reader:
			\begin{proposition}
				Suppose that $F\in S^{2}_{1}(\X\times\X^*)$ is such that $T:=\Op^A(F)$ is invertible in $\mathbb{B}\big(L^2(\X)\big)$ with bounded inverse $T^{-1}$. We denote by $T^{-1}:=\Op^A(F^-)$ for some $F^-\in\mathscr{S}^\prime(\Xi)$, and by  $\mathcal{s}_{2}(\xi):=<\xi>^{2}$. If  we have $\Op^A\big(\mathcal{s}_{2}\, \sharp^B\,  F^-\big)\in\mathbb{B}\big(L^2(\X)\big)$, then $F^-\in S^{-2}_{1} (\X\times\X^*)$.
			\end{proposition}
			We intend to use this statement with $A=0$ taking 
			\beq\label{F-h+}
			F=q_{I,\leq}+h_+-Eq_{I,+,\delta}=q_{I,\leq}+q_{I,+,\delta}\,\sharp\, h\,\sharp\, q_{I,+,\delta}-{ E}q_{I,+,\delta} \in
			S^2_1(\X\times\X^*).
			\eeq
			From the above arguments we know that $\Op(F)$ is invertible in $\mathbb{B}\big(L^2(\X)\big)$ with bounded inverse $\Op(F^-)$ and 
			$$
			F^-=\widetilde{r}_{I,+}(E)=q_{I,\leq}+r_{I,+}(E)=q_{I,\leq}+q_{I,+,\delta}\sharp r_{I,+}(E)\,\sharp\, q_{I,+,\delta}.
			$$
			Thus we only have to prove that 
			$$
			\Op\big(\mathcal{s}_2\, \sharp\,  \widetilde{r}_{I,+}(E)\big)\,\in\,\mathbb{B}\big(L^2(\X)\big).
			$$
			The fact that $h\in S^2_1(\X\times\X^*)$ is elliptic (see Remark \ref{R-HGammaOp}) implies the existence of a parametrix $\widetilde{h}\in S^{-2}_1(\X\times\X^*)$, satisfying, for some $h_\infty \in S^{-\infty}(\X\times\X^*)$ the relation:
			$$
			\widetilde{h}\,\sharp\,h\,=\,1\,+\,h_\infty\,.
			$$
			Composing to the left by $\mathcal{s}_2$, we get 
			$$\mathcal{s}_2=\mathcal{s}_2\,\sharp\,\widetilde{h}\,\sharp\,h\,-\,\mathcal{s}_2\,\sharp\,h_\infty,$$
			with $\mathcal{s}_2\,\sharp\,\widetilde{h}\in S^0_1(\X\times\X^*)$ and $\mathcal{s}_2\,\sharp\,h_\infty\in S^{-\infty}(\X\times\X^*)$.  These two symbols generate $L^2$-bounded operators. Since $h-h_+$ is a smoothing symbol, we are left with $h_+\sharp \widetilde{r}_{I,+}(E)$. But from \eqref{F-h+} we get:
			$$ h_+\sharp\widetilde{r}_{I,+}(E)= 1+E q_{I,+,\delta}\sharp\widetilde{r}_{I,+}(E)-q_{I,\leq}\sharp\widetilde{r}_{I,+}(E),$$
			defining a bounded operator in $L^2(X)$. This ends the proof of \eqref{mai1} and \eqref{F-band-res}.

\noindent\textbf{Step 4:}
	Using Proposition~\ref{propB.14} and the results of Step 2 and 3, we conclude that
		$$
		q_{I,\pm,\delta},\sharp^\epsilon\, r_{I,\pm}(E)=q_{I,\pm,\delta}\,\sharp\, r_{I,\pm}(E)\,+\,\mathfrak{z}_\pm^\epsilon
		$$
		with $\mathfrak{z}_-^\epsilon\in S^{-\infty}(\X\times\X^*)$ and $\mathfrak{z}_+^\epsilon\in S^{-4}_1(\X\times\X^*)$ and also 
		$ \|\Op^\epsilon\big(\mathfrak{z}_\pm^\epsilon\big)\|_{\mathbb{B}(L^2(\X))}=\mathcal{o}(\epsilon)\,$.
	Finally, using similar arguments as above, we deduce  that
	\begin{align} \label{F-epsilon-ortog}
		R^\epsilon_{I,\pm} (E)  =\Op^\epsilon\big(r_{I,\pm}(E)\big)&=  \Op^\epsilon \big( q_{I,\pm,\delta}\,\sharp\, r_{I,\pm} (E) \big)= \Op^\epsilon \big( q_{I,\pm,\delta}\,\sharp^\epsilon\, r_{I,\pm} (E) \big)+\mathcal{O}(\epsilon) \nonumber \\
		&=\Op^\epsilon \big(q_{I,\pm,\delta} \big)\, R^\epsilon_{I,\pm}(E) +\mathcal{O}(\epsilon) \\ 
		&\hspace*{-2,45cm}=\Op^\epsilon \big(r_{I,\pm}(E)\,\sharp\,q_{I,\pm,\delta} \big) 
		=\Op^\epsilon \big( r_{I,\pm}(E)\,\sharp^\epsilon\, q_{I,\pm,\delta} \big)+\mathcal{O}(\epsilon)
		=R^\epsilon_{I,\pm}(E) \Op^\epsilon \big(q_{I,\pm,\delta} \big)+\mathcal{O}(\epsilon)\,.\nonumber 
	\end{align}

\noindent\textbf{Step 5:}
	The remaining difficulty is that the operators $\Op^\epsilon(q_{I,\pm,\delta})$ are no longer two orthogonal projections in $L^2(\X)$. Nevertheless, they are not \enquote{far} from two mutually orthogonal, orthogonal projections and that is what we need to prove now.
	
{  We use Remark \ref{R-QHQ-sp-desc}, the formulas for Riesz projections associated with some isolated spectral intervals and the results concerning the continuity (even Lipschitz regularity) of the spectral edges with respect to the intensity of the constant magnetic field (see \cite{AMP}, \cite{CP-1}) in order to get the following result.
\begin{proposition} There exists some $\epsilon_0>0$ such that for any $\epsilon\in[0,\epsilon_0]$ we can find
	\begin{itemize}
		\item a smooth contour $\mathscr{C}_-\subset\mathbb{C}$ having the segment $[-E_0,-\Lambda_-]$ in its interior and the semi-axis $[0,+\infty)$ in its exterior and remaining at a strictly positive distance from the spectrum of $\Op^{\epsilon}\big(q_{I,\delta}\,\sharp\, h\,\sharp\, q_{I,\delta}\big)$, 
		\item a smooth contour $\mathscr{C}_0\subset\mathbb{C}$ leaving the segment $[-C\epsilon,C\epsilon]$ in its interior and the rest of the spectrum of $\Op^{\epsilon}\big(q_{I,\delta}\,\sharp\, h\,\sharp\, q_{I,\delta}\big)$ in its exterior at some strictly positive distance,
	\end{itemize} 
		 such that we can write the identities
		 $$
		 q_{I,-,\delta}=\frac{1}{2\pi i}\int_{\mathscr{C}_-}\widetilde{q}^{0}(\z)\,d\z,\qquad p_{I,\delta}=\frac{1}{2\pi i}\int_{\mathscr{C}_0}\widetilde{q}^{0}(\z)\,d\z,
		 $$
where	
	$$
	\widetilde{q}^{0}(\z)\,:= \,\big(q_{I,\delta}\,\sharp\, h\,\sharp\, q_{I,\delta}-\z\bb1\big)^{-}
	$$
is the Weyl symbol of the resolvent of $q_{I,\delta}\,\sharp\, h\,\sharp\, q_{I,\delta}$ in $\z$, the inverse being taken for the usual Moyal product,
	\end{proposition}}

Let us define $\widetilde{q}^{\epsilon}(\z):=\big(q_{I,\delta}\,\sharp\, h\,\sharp\, q_{I,\delta}-\z\bb1\big)^{-}_\epsilon$, with the inverse  taken with respect to the magnetic Moyal product associated to the magnetic field $\epsilon B^\circ$, and also
		$$
		\widetilde{Q}^{\epsilon}_{I,-}=\Op^{\epsilon}\Big(\frac{1}{2\pi i}\int_{\mathscr{C}_-}\widetilde{q}^{\epsilon}(\z)\,d\z\Big)
		$$
which is an orthogonal projection. Thus, using Proposition \ref{propB.14} we conclude that:
	\begin{align*}
	& \Op^{\epsilon}\big(q_{I,-,\delta}\big)=\Op^\epsilon\Big(\frac{1}{2\pi i}\int_{\mathscr{C}_-}\widetilde{q}^{0}(\z)\,d\z\Big)=\widetilde{Q}^{\epsilon}_{I,-}+\mathcal{O}(\epsilon).
	\end{align*}
	
	We can also notice that
	\begin{align*}
	\Op^{\epsilon}\big(q_{I,+,\delta}\big)&=\Op^{\epsilon}\big(1-p_{I,\delta}-q_{I,-,\delta}\big)=\bb1-\Op^{\epsilon}\big(p_{I,\delta}\big)-\Op^{\epsilon}\big(q_{I,-,\delta}\big) \\
	&=\bb1-\Op^{\epsilon}\left(\frac{1}{2\pi i}\oint_{\mathscr{C}_0}d\z\,\widetilde{q}^{0}(\z)\right)-\Op^{\epsilon}\left(\frac{1}{2\pi i}\oint_{\mathscr{C}_-}d\z\,\widetilde{q}^{0}(\z)\right).
	\end{align*}
	We introduce the orthogonal projections:
	\begin{align*}
\widetilde{P}^{\epsilon}_I:=\Op^{\epsilon}\left(\frac{1}{2\pi i}\oint_{\mathscr{C}_0}d\z\,\widetilde{q}^{\epsilon}(\z)\right),\quad \widetilde{Q}^{\epsilon}_{I,+}:=\bb1\,\ominus\,\widetilde{P}^{\epsilon}_I\,\ominus\,\widetilde{Q}^{\epsilon}_{I,-}
	\end{align*}
and notice that due to Proposition \ref{propB.14} we have that
	\begin{equation*}
	\Op^{\epsilon}\big(q_{I,+,\delta}\big)=\widetilde{Q}^{\epsilon}_{I,+}\,+\,\mathcal{O}(\epsilon).
	\end{equation*}
	Finally, putting all these formulas together we see that
	\begin{align*}
		\Op^\epsilon\big(q_{I,\pm,\delta}\big)\Op^\epsilon\big(q_{I,\pm,\delta}\big)&=
		\Big[\widetilde{Q}^\epsilon_{I,\pm}+\mathcal{O}(\epsilon)\Big]\Big[\widetilde{Q}^\epsilon_{I,\pm}+\mathcal{O}(\epsilon)\Big]=\widetilde{Q}^\epsilon_{I,\pm}+\mathcal{O}(\epsilon)=\Op^\epsilon\big(q_{I,\pm,\delta}\big)+\mathcal{O}(\epsilon),\\
		\Op^\epsilon\big(q_{I,\pm,\delta}\big)\Op^\epsilon\big(q_{I,\mp}\big)&=
		\Big[\widetilde{Q}^\epsilon_{I,\pm}+\mathcal{O}(\epsilon)\Big]\Big[\widetilde{Q}^\epsilon_{I,\mp}+\mathcal{O}(\epsilon)\Big]=\mathcal{O}(\epsilon),\\
		\Op^\epsilon\big(q_{I,\pm,\delta}\big)\Op^\epsilon\big(p_{I,\delta}\big)&=
		\Big[\widetilde{Q}^\epsilon_{I,\pm}+\mathcal{O}(\epsilon)\Big]\Big[\widetilde{P}^\epsilon_{I}+\mathcal{O}(\epsilon)\Big]=\mathcal{O}(\epsilon).
	\end{align*}
	Recalling \eqref{F-magn--} and \eqref{F-mag-+} we conclude that for any $E\in I_\circ$ we have the estimations
	\begin{align} \nonumber
	&\Big(Q^\epsilon_{I,\delta}H^\epsilon_\Gamma Q^\epsilon_{I,\delta}-E\bb1\Big)R^\epsilon_I(E)=\Big[Q^\epsilon_{I,\delta}\Big(H^\epsilon_\Gamma-E\Big)Q^\epsilon_{I,\delta}-EP^\epsilon_{I,\delta}\Big]R^\epsilon_I(E)=\\ \nonumber
	&\hspace*{2cm}\hspace*{14pt}=\Big[\Op^\epsilon\big(q_{I,-,\delta}\big)\Op^\epsilon\big(q_{I,-,\delta}\,\sharp\,(h-E)\,\sharp\, q_{I,-,\delta}\big)\Op^\epsilon\big(q_{I,-,\delta}\big) \\
	&\hspace*{2cm}\hspace*{28pt}+\Op^\epsilon\big(q_{I,+,\delta}\big)\Op^\epsilon\big(q_{I,+,\delta}\,\sharp\,(h-E)\,\sharp\ q_{I,+,\delta}\big)\Op^\epsilon\big(q_{I,+,\delta}\big)\,+\,\mathcal{O}(\epsilon)-EP^\epsilon_{I,\delta}\Big]\times \nonumber \\ \nonumber
	&\hspace*{2cm}\hspace*{28pt}\times\Big(\Op^\epsilon\big(q_{I,-,\delta}\big)\Op^\epsilon\big(r^\epsilon_-(E)\big)+\Op^\epsilon\big(q_{I,+,\delta}\big)\Op^\epsilon\big(r^\epsilon_+(E)\big)+\mathcal{O}(\epsilon)\Big)\\ \label{F-RIepsilon}
	&\hspace*{2cm}\hspace*{14pt}=\widetilde{Q}^\epsilon_{I,-}+\widetilde{Q}^\epsilon_{I,+}+\mathcal{O}(\epsilon).
	\end{align}
	Thus, 
	for any closed interval $I_\circ\subset\mathring{I}$ there exists some $\epsilon_0>0$, such that for any $E\in I_\circ$ the operator \break  $\Big(Q^\epsilon_{I,\delta}H^\epsilon_\Gamma Q^\epsilon_{I,\delta}-E\bb1\Big)$ is invertible as operator in $Q^\epsilon_{I,\delta}L^2(\X)$, uniformly for $(E,\epsilon)\in I_\circ\times[0,\epsilon_0]$ and this finishes our proof.
\end{proof}

\begin{remark}\label{R-control-2}
We notice that the upper bound on $\epsilon>0$ and thus the value of $\epsilon_0$ is controlled by the invertibility conditions involved by \eqref{F-RIepsilon} and thus depend on the norm of $R^\epsilon_I(E)$ that is controlled uniformly in $\epsilon>0$ by the distance $\dist(E,\mathbb{R}\setminus I)\leq\dist(I_\circ,\mathbb{R}\setminus I)$.
\end{remark}

Applying Proposition \ref{P-FS-arg} and taking in \eqref{F-Htilde-magn}:
\beq\label{FD-R-epskappa-bot}
R^{\epsilon}_\bot(0)\,:=\,R^{\epsilon}_{I}(0) \,=\,R^{\epsilon}_{I,+}(0)+R^{\epsilon}_{I,-}(0)
\eeq  
with the notation introduced  in the proof of  Proposition \ref{P-Hyp-H-magn}, we obtain the following statement.

\begin{proposition}\label{P-magn-FS-est}
	Given any compact interval $I_\circ\subset\mathring{I}$ containing $0$ in its interior there exist $\epsilon_0>0$ and $C>0$ (depending on $\epsilon_0>0$, such that for any $\epsilon\in[0,\epsilon_0]$ we have the estimation:
	$$
	\max\left \{ \sup_{e\in \sigma(H^{\epsilon}_\Gamma)\cap I_\circ}\dist\Big(e,\sigma(\widetilde{H}^{\epsilon}_{I,\delta})\Big),\; \sup_{e\in \sigma(\widetilde{H}^{\epsilon}_{I,\delta})\cap I_\circ}\dist\Big(e,\sigma(H^{\epsilon}_\Gamma)\Big)\right \}\leq\,\|H^{\epsilon}_\Gamma P^{\epsilon}_I\|^2\left(\frac{\ell_{I_\circ}}{\Lambda_I}\right)^2\frac{1}{d_{I_\circ}}\; \leq C \; \ell_{I_\circ}^2/d_{I_\circ}
	$$
	where $d_{I_\circ}:=\dist(I_\circ,\R\setminus \mathring{I})$, $\ell_{I_\circ}:=\sup\{|t|\,,\,t\in I_\circ\}$ and $\Lambda_I:=\min\big\{\Lambda_-,\Lambda_+\big\}$.
\end{proposition} 

\section{The Peierls-Onsager effective Hamiltonian.}\label{S-PO-eH}

\subsection{The magnetic matrix.}\label{SS-symbol-PO-eff-ham}

An important step in our construction, when working with a constant magnetic field, is to put into evidence a Toeplitz matrix associated to the magnetic quasi-band Hamiltonian in the orthonormal basis of the quasi-Wannier functions. This allows then to define the Peierls-Onsager effective Hamiltonian of the quasi-band (the 5-th step in our strategy as described in Subsection \ref{SS-ProofStr}).

In our situation when the fiber of the quasi-band has dimension two, we have to pay attention to the 2-dimensional basis that we choose in order to compare the Peierls-Onsager effective Hamiltonian with the magnetic quasi-band Hamiltonian. More precisely, we shall compare the matrices of the two operators in the basis given by $\hat{\varphi}_\pm(\theta):=\widehat{\Upsilon}(\theta)^{-1}\varepsilon_\pm$ associated by the trivialization $\widehat{\Upsilon}(\theta):\int^\oplus_{\Sigma_I}d\theta\,\big[\widehat{\Pi}(\theta)\mathscr{F}_\theta\big]\overset{\sim}{\rightarrow}\Sigma_I\times\mathbb{C}^2$ to the canonical basis $\{\varepsilon_-,\varepsilon_+\}$ of $\mathbb{C}^2$. 
Let us define the \enquote{change-of-base matrix} unitary operator defined on each $\widehat{\Pi}_I(\theta)\mathscr{F}_\theta$ for $\theta\in \Sigma_I$:
\beq\label{DF-mItheta}
\widehat{\mathfrak{u}}_{I,\delta}(\theta)\hat \Psi_{\pm,\delta}(\theta)=\hat{\varphi}_\pm(\theta).
\eeq
This family of unitaries is smooth as a function of $\theta$ on $\Sigma_I$, thus we may find a ball $B_{2R}(\theta_0)$ where one can define a smooth logarithm.   Let us now choose a smooth cut-off function $\chi\in C^\infty_0\big(B_{2R}(\theta_0);[0,1]\big)$ with $\chi(\theta)=1$ for $\theta\in B_{R}(\theta_0)$ and define the smooth section
$$
\mathbb{T}_*\ni\theta\mapsto A(\theta):=\chi(\theta)\ln\widehat{\mathfrak{u}}_{I,\delta}(\theta)\in\mathbb{B}\big(\widehat{\Pi}_I(\theta)\mathscr{F}_\theta\big).
$$ 
We define the smooth global orthonormal frame
\beq\label{DF-Phi-pm}
\hat{\Phi}_{\pm,\delta}(\theta):=e^{ A(\theta)}\hat{\Psi}_{\pm,\delta}
\eeq
verifying:
\beq\begin{split}\label{mai2}
	\hat{\Phi}_{\pm,\delta}(\theta)&=\hat{\varphi}_\pm(\theta),\ \forall\theta\in B_R(\theta_0),\\
	\hat{\Phi}_{\pm,\delta}(\theta)&=\hat{\Psi}_{\pm,\delta}(\theta),\ \forall\theta\in\mathbb{T}_*\setminus\Sigma_I.
\end{split}\eeq
We can now define the new pair of quasi-Wannier functions 
$$
\Phi_{\pm,\delta}:=\mathcal{U}_\Gamma^{-1}\int_{\mathbb{T_*}}^\oplus d\theta\;  \hat{\Phi}_{\pm,\delta}
$$
and their $\Gamma$-translations $\Phi_{\pm,\gamma,\delta}$ for any $\gamma\in\Gamma$. They still form an orthonormal basis for the quasi-band subspace $P_{I,\delta}L^2(\X)$ given in Definition \ref{D-Pdelta}, so that we have the equality (similar to the one in Definition \ref{D-Pdelta}):
$$
P_{I,\delta}=\underset{\gamma\in\Gamma}{\sum}\left(\underset{j=\pm}{\sum}\big|\Phi_{j,\gamma,\delta}\big\rangle\big\langle\Phi_{j,\gamma,\delta}\big|\right)
$$

Passing now to the magnetic quasi-band let us notice that all the arguments in Subsection \ref{SS-mWfunc} only depend on two facts: (i) the existence of an orthonormal basis for the quasi-band projection and (ii) the Zak magnetic translations we apply to this orthonormal basis. Thus, starting with the basis $\big\{\Phi_{\pm,\gamma,\delta}\big\}_{\gamma\in\Gamma}$ we obtain the same orthogonal projection $P^\epsilon_{I,\delta}$ and we have an analogue of Propositions \ref{P-Gepsilon} - \ref{P-magn-W-function}.

We shall denote by
	\beq\label{F-psi-e2}
\widetilde{\psi}^\epsilon_{\pm,0,\delta}(x)=\underset{\alpha\in\Gamma}{\sum}\underset{j=\pm}{\sum}\mathbf{F}^\epsilon_\delta (\alpha)_{j\pm}\,
\Omega^\epsilon(\alpha,0,x)\, \Phi_{j,\alpha,\delta}(x)\,,
\eeq
and notice that the family
\beq\label{F-psi-epsilon-2}
\widetilde{\phi}^\epsilon_{j,\gamma,\delta}\,=\,\Lambda^\epsilon(\cdot ,\gamma)(\tau_{\gamma}
\widetilde{\psi}^\epsilon_{j,0,\delta})\,,\qquad\forall\gamma\in\Gamma\,
\eeq
also forms an orthonormal basis of $P^\epsilon_{I,\delta}L^2(\X)$.
Let us consider the unitary operator 
$$
\mathbf{\Phi}^\epsilon_\delta:P^\epsilon_{I,\delta}L^2(\X)\overset{\sim}{\rightarrow}\ell^2(\Gamma)\otimes\mathbb{C}^2,\quad\mathbf{\Phi}^\epsilon_\delta\big(\widetilde{\phi}^\epsilon_{\pm,\gamma,\delta}\big):=\mathfrak{e}_\gamma\otimes\varepsilon_\pm
$$
where $\{\varepsilon_-,\varepsilon_+\}$ is the canonical orthonormal basis in $\mathbb{C}^2$ (see Definition \ref{D-onb-SigmaI}) and 
 $\{\mathfrak{e}_\gamma\}_{\gamma\in\Gamma}\subset\ell^2(\Gamma)$ the canonical orthonormal basis of $\ell^2(\Gamma)$.

Starting from \eqref{F-desc-PepsilonI} and using the Cotlar-Stein procedure to control the convergence of the series, we can write that:
$$
P^{\epsilon}_{I,\delta}\widetilde{H}^{\epsilon}_{I,\delta}P^{\epsilon}_{I,\delta}=\underset{\alpha,\beta\in\Gamma}{\sum}\Pi^\epsilon_{\alpha,\delta}\widetilde{H}^{\epsilon}_{I,\delta}\Pi^\epsilon_{\beta,\delta}=\underset{\alpha,\beta\in\Gamma}{\sum}\left(\,\underset{j,k=\pm}{\sum}\big|\widetilde{\phi}^\epsilon_{j,\alpha,\delta}\big\rangle\left\langle\widetilde{\phi}^{\epsilon}_{j,\alpha,\delta}\,,\widetilde{H}^{\epsilon}_{I,\delta}\,\widetilde{\phi}^{\epsilon}_{k,\beta,\delta}\right\rangle_{L^2(\X)}\big\langle\widetilde{\phi}^\epsilon_{k,\beta,\delta}\big|\right)
$$
with the series converging in the operator-norm topology and resp. in the uniform topology for the integral kernels and in the topology of $BC^\infty(\Xi)$ for the magnetic symbols. This allows us to view our \textit{dressed modified magnetic quasi-band Hamiltonian} as an infinite matrix, in fact as an infinite matrix with entries from $\mathcal{M}_{2\times2}(\mathbb{C})$:
\beq\label{FD-magn-matrix}
\big[\mathscr{M}^{\epsilon}_{I,\delta}\big]_{(j,\alpha),(k,\beta)}\,:=\,\left\langle\widetilde{\phi}^{\epsilon}_{j,\alpha,\delta}\,,\widetilde{H}^{\epsilon}_{I,\delta}\,\widetilde{\phi}^{\epsilon}_{k,\beta,\delta}\right\rangle_{L^2(\X)},\quad(j,k)\in\{-,+\}\times\{-,+\}.
\eeq
From Propositions \ref{P-Zak-magn-transl} and \ref{P-magn-W-function}, taking into account Proposition \ref{P-PepsilonI-per} and the definition of $\widetilde{H}^\epsilon_{I,\delta}$ given in  \eqref{F-Htilde-magn} we can write
\begin{align*}
	\big[\mathscr{M}^{\epsilon}_{I,\delta}\big]_{(j,\alpha),(k,\beta)}&=\left\langle\widetilde{\phi}^{\epsilon}_{j,\alpha,\delta}\,,\widetilde{H}^{\epsilon}_{I,\delta}\,\widetilde{\phi}^{\epsilon}_{k,\beta,\delta}\right\rangle_{L^2(\X)}=\left\langle\mathcal{T}^\epsilon_\alpha\widetilde{\psi}^{\epsilon}_{j,0,\delta}\,,\widetilde{H}^{\epsilon}_{I,\delta}\,\mathcal{T}^\epsilon_\beta\widetilde{\psi}^{\epsilon}_{k,0,\delta}\right
	\rangle_{L^2(\X)}\\
	&=\Lambda^\epsilon(\alpha,\beta)\left\langle\mathcal{T}^\epsilon_{\alpha-\beta}
	\widetilde{\psi}^{\epsilon}_{j,0,\delta}\,,\widetilde{H}^{\epsilon}_{I,\delta}\,\widetilde{\psi}^{\epsilon}_{k,0,\delta}\right\rangle_{L^2(\X)}=\Lambda^\epsilon(\alpha,\beta)\big[\mathscr{M}^{\epsilon}_{I,\delta}\big]_{(j,\alpha-\beta),(k,0)}\\
	&\equiv\Lambda^\epsilon(\alpha,\beta)\mathcal{m}^\epsilon_{I,\delta}(\alpha-\beta)_{jk}
\end{align*}
and conclude that
$$
P^{\epsilon}_{I,\delta}\widetilde{H}^{\epsilon}_{I,\delta}P^{\epsilon}_{I,\delta}=\underset{\alpha,\beta\in\Gamma}{\sum}\left(\,\underset{j,k=\pm}{\sum}\,\Lambda^\epsilon(\alpha,\beta)\mathcal{m}^\epsilon_{I,\delta}(\alpha-\beta)_{jk}\,\big|\widetilde{\phi}^\epsilon_{j,\alpha,\delta}\big\rangle\big\langle\widetilde{\phi}^\epsilon_{k,\beta,\delta}\big|\right).
$$
Now we can define the discrete Fourier transform
\beq\label{FD-kIepsilon}
\mathcal{k}^\epsilon_{I,\delta}(\theta):=\underset{\gamma\in\Gamma}{\sum}e^{-i<\theta,\gamma>}\mathcal{m}^\epsilon_{I,\delta}(\gamma)\in\mathcal{M}_{2\times2}(\mathbb{C}),\quad\forall\theta\in\mathbb{T}_*
\eeq
and  repeat the arguments in Subsection 8.3 in \cite{CHP-2}. In fact, we can view the above function as a matrix-valued symbol which is $\Gamma_*$-periodic with respect to $\theta$ and does not depend on the variables from $\X$. It belongs to the class $S^0_0(\X\times\X^*)_{2\times2}$ defined in Notation \ref{N-matrix-symbols}. 

\textit{The magnetic quantization} of our matrix-valued symbol $ \mathcal{k}^\epsilon_{I,\delta}\in S^0_0(\X\times\X^*)_{2\times2}$ with the constant magnetic field $\epsilon B^\circ$ is
$$\Op^\epsilon(\mathcal{k}^\epsilon_{I,\delta})\,\in\,\mathbb{B}\big(L^2(\X)\otimes\mathbb{C}^2\big).$$
 Also, for all $u\in\mathscr{S}(\X;\mathbb{C}^2)$ we have 
\begin{align*}
\big(\Op^\epsilon(\mathcal{k}^\epsilon_{I,\delta})u\big)_\pm(x)&=(2\pi)^{-2}\int_\X dy\,\Lambda^\epsilon(x,y)\int_{\X^*}d\xi\,e^{i<\xi,x-y>}\Big(\underset{j=\pm}{\sum}\mathcal{k}^\epsilon_{I,\delta}\big(\xi\big)_{\pm,j}u_j(y)\Big)\\ \nonumber
&=\underset{\gamma\in\Gamma}{\sum}\Lambda^\epsilon(x,x-\gamma)\Big(\underset{j=\pm}{\sum}\mathcal{m}^\epsilon_{I,\delta}(\gamma)_{\pm,j}u_j(x-\gamma)\Big).
\end{align*}
If we combine  the unitary transformation
$$
\mathcal{W}_\Gamma:L^2(\X;\mathbb{C}^2)\overset{\sim}{\longrightarrow}\ell^2(\Gamma;\mathbb{C}^2)\otimes L^2(E),\qquad\big(\mathcal{W}_\Gamma\phi\big)(\gamma,\hat{x})_\pm:=\phi(\gamma+\hat{x})_\pm,\ \forall\phi\in\mathscr{S}(\X;\mathbb{C}^2)
$$
with the \textit{Luttinger gauge} transformation (\cite{Lu}):
$$
\mathcal{U}^\epsilon_L:\ell^2(\Gamma;\mathbb{C}^2)\otimes L^2(E)\rightarrow\ell^2(\Gamma;\mathbb{C}^2)\otimes L^2(E),\qquad\big(\mathcal{U}^\epsilon_L\phi\big)(\alpha,\hat{x})_\pm:=\Lambda^\epsilon(\hat{x},\alpha)\phi(\alpha,\hat{x})_\pm,
$$
we obtain 
\begin{align*}
	\big(\mathcal{U}^\epsilon_L\mathcal{W}_\Gamma\Op^
	\epsilon\big(\mathcal{k}^\epsilon_{I,\delta}\big)\phi\big)(\alpha,\hat{x})_\pm&=\Lambda^\epsilon(\hat{x},\alpha)\underset{\gamma\in\Gamma}{\sum}\Lambda^\epsilon(\hat{x}+\alpha,\hat{x}+\alpha-\gamma)\Big(\underset{j=\pm}{\sum}\mathcal{m}^\epsilon_{I,\delta}(\gamma)_{\pm,j}\phi_j(\hat{x}+\alpha-\gamma)\Big) \\
	&=\underset{\gamma\in\Gamma}{\sum}\Lambda^\epsilon(\alpha,\alpha-\gamma)\Big(\underset{j=\pm}{\sum}\mathcal{m}^\epsilon_{I,\delta}(\gamma)_{\pm,j}\Lambda^\epsilon(\hat{x},\alpha-\gamma)\phi_j(\hat{x}+\alpha-\gamma)\Big)\\
	&=\underset{\beta\in\Gamma}{\sum}\Lambda^\epsilon(\alpha,\beta)\Big(\underset{j=\pm}{\sum}\mathcal{m}^\epsilon_{I,\delta}(\alpha-\beta)_{\pm,j}\Lambda^\epsilon(\hat{x},\beta)\phi_j(\hat{x}+\beta)\Big) \\
	&=\underset{\beta\in\Gamma}{\sum}\underset{j=\pm}{\sum}\big(\Lambda^\epsilon(\alpha,\beta)\mathcal{m}^\epsilon_{I,\delta}(\alpha-\beta)_{\pm,j}\big)\big(\mathcal{U}^\epsilon_L\mathcal{W}_\Gamma\phi\big)(\hat{x},\beta)_j \\
	& =\underset{\beta\in\Gamma}{\sum}\underset{j=\pm}{\sum}\big[\mathscr{M}^{\epsilon}_{I,\delta}\big]_{(\pm,\alpha),(j,\beta)}\big(\mathcal{U}^\epsilon_L\mathcal{W}_\Gamma\phi\big)(\hat{x},\beta)_j.
\end{align*}
Thus, having the canonical orthonormal basis $\{\varepsilon_-,\varepsilon_+\}$ in $\mathbb{C}^2$ (see Definition \ref{D-onb-SigmaI}) and 
 the canonical orthonormal basis $\{\mathfrak{e}_\gamma\}_{\gamma\in\Gamma}$ in $\ell^2(\Gamma)$, we conclude that 
$$
\underset{(\alpha,\beta)\in\Gamma^2}{\sum}\underset{j=\pm,k=\pm}{\sum}\big[\mathscr{M}^{\epsilon}_{I,\delta}\big]_{(j,\alpha),(k,\beta)}\big|\mathfrak{e}_\alpha\otimes\varepsilon_j\big\rangle\big
\langle\mathfrak{e}_\beta\otimes\varepsilon_k\big|\ =\ \big(\mathcal{U}^\epsilon_L\mathcal{W}_\Gamma\big)
\Op^\epsilon\big(\mathcal{k}^\epsilon_{I,\delta}\big)\big(\mathcal{U}^
\epsilon_L\mathcal{W}_\Gamma\big)^{-1}.
$$
Moreover, by construction, we have that
\begin{align*}
\begin{split}
	\widetilde{H}^\epsilon_{I,\delta}&=\underset{(\alpha,\beta)\in\Gamma^2}{\sum}\underset{j=\pm,k=\pm}{\sum}\big[\mathscr{M}^{\epsilon}_{I,\delta}\big]_{(j,\alpha),(k,\beta)}\big|\widetilde{\phi}^\epsilon_{j,\alpha,\delta}\big\rangle\big\langle\widetilde{\phi}^\epsilon_{k,\beta,\delta}\big| \\
	&=\underset{(\alpha,\beta)\in\Gamma^2}{\sum}\underset{j=\pm,k=\pm}{\sum}\big[\mathscr{M}^{\epsilon}_{I,\delta}\big]_{(j,\alpha),(k,\beta)}\big[\mathbf{\Phi}^\epsilon_\delta\big]^{-1}\big|\mathfrak{e}_\alpha\otimes\varepsilon_j\big\rangle\big
	\langle\mathfrak{e}_\beta\otimes\varepsilon_k\big|\big[\mathbf{\Phi}^\epsilon_\delta\big] \\
	&=\big(\big[\mathbf{\Phi}^\epsilon_\delta\big]^{-1}
	\mathcal{U}^
	\epsilon_L\mathcal{W}_\Gamma\big)\Op^\epsilon\big(
	\mathcal{k}^\epsilon_{I,\delta}\big)\big(\big[\mathbf{\Phi}^\epsilon_\delta\big]^
	{-1}\mathcal{U}^\epsilon_L\mathcal{W}_\Gamma\big)^
	{-1}
\end{split}
\end{align*}
and we conclude that the following statement holds.

\begin{proposition}\label{P-magn-matrix-ke}
	The operator $\mathfrak{Op}^{\epsilon}(\mathcal{k}^\epsilon_{I,\delta})$ in $\mathbb{B}\big(L^2(\X)\otimes\mathbb{C}^2\big)$ and the \enquote{dressed} modified quasi-band Hamiltonian $\widetilde{H}^\epsilon_{I,\delta}$ in $P^{\epsilon}_{I,\delta}L^2(\X)$ are unitarily equivalent and thus have the same spectrum.
\end{proposition}

Our strategy in the following is to compare the symbol $\mathcal{k}^\epsilon_{I,\delta}(\theta)$ in the neighborhood of $\theta_0$ with the $2\times2$-matrix valued function $M_I(\theta)$ defined in \eqref{FD-Psi}.

\subsection{Estimating the magnetic perturbation.}\label{SS-est-PO-eff-ham}

In agreement with the notations for the Weyl and the magnetic quantizations \eqref{D-PsiDO} and \eqref{D-MPsiDO} and with the previous notations, we shall denote the operators and symbols for the $0$-magnetic field by just suppressing the parameter $\epsilon=0$. Let us consider the smooth matrix valued function obtained by putting $\epsilon=0$ in \eqref{FD-kIepsilon}:
$$
\mathbb{T}_*\ni\theta\mapsto\mathcal{k}_{I,\delta}(\theta)_{jk}:=\underset{\gamma\in\Gamma}{\sum}e^{-i<\theta,\gamma>}\left\langle\Phi_{j,\gamma,\delta},\widetilde{H}_{I,\delta}\Phi_{k,0,\delta}\right\rangle_{L^2(\X)}\in\mathcal{M}_{2\times2}(\mathbb{C}).
$$
We shall study the difference
\beq\label{FD-kI0}
\begin{split}
	&\big\|\big(\partial^a\mathcal{k}^\epsilon_{I,\delta}\big)(\theta)-\big(\partial^a\mathcal{k}_{I,\delta}\big)(\theta)\big\|_{\mathcal{M}_{2\times2}(\mathbb{C})}=\underset{j,k}{\max}\left|\underset{\gamma\in\Gamma}{\sum}(-i\gamma)^ae^{-i<\theta,\gamma>}\big(\mathcal{m}^\epsilon_{I,\delta}(\gamma)_{jk}-
	\mathcal{m}_{I,\delta}(\gamma)_{jk}\big)\right| \\
	&\leq \underset{j,k}{\max}\left(\underset{\gamma\in\Gamma}{\sum}<\gamma>^{|a|}\Big|\left\langle\mathcal{T}^\epsilon_{\gamma}
	\widetilde{\psi}^{\epsilon}_{j,0,\delta}\,,\widetilde{H}^{\epsilon}_{I,\delta}\,\widetilde{\psi}^{\epsilon}_{k,0,\delta}\right\rangle_{L^2(\X)}-\left\langle
	\tau_{\gamma}
	\widetilde{\psi}^{}_{j,0,\delta}\,,\widetilde{H}_{I,\delta}\,\widetilde{\psi}^{}_{k,0,\delta}\right\rangle_{L^2(\X)}\Big|\right)
\end{split}
\eeq
and compute
\begin{align}\nonumber
	\left\langle\mathcal{T}^\epsilon_{\gamma}
	\widetilde{\psi}^{\epsilon}_{j,0,\delta}\,,\widetilde{H}^{\epsilon}_{I,\delta}\,\widetilde{\psi}^{\epsilon}_{k,0,\delta}\right\rangle_{L^2(\X)}&-\left\langle
	\tau_{\gamma}
	\widetilde{\psi}^{}_{j,0,\delta}\,,\widetilde{H}_{I,\delta}\,\widetilde{\psi}^{}_{k,0,\delta}\right\rangle_{L^2(\X)} \\
	\nonumber
	&\hspace*{-2cm}=\left\langle
	\widetilde{\phi}^{\epsilon}_{j,\gamma}\,,(Y^\epsilon_\delta)^{-1/2}H^{\epsilon}_{I,\delta}
	(Y^\epsilon_\delta)^{-1/2}\,\widetilde{\phi}^{\epsilon}_{k,0,\delta}\right\rangle_{L^2(\X)}-\left\langle
	\Phi^{}_{j,\gamma}\,,Y_\delta^{-1/2}H_{I,\delta}Y_\delta^{-1/2}\,\Phi^{}_{k,0,\delta}\right\rangle_{L^2(\X)} \\ \label{FD-kIO-1}
	&\hspace*{-1,75cm}+\left\langle
	\widetilde{\phi}^{\epsilon}_{j,\gamma}\,,X^\epsilon_\delta\,\widetilde{\phi}^{\epsilon}_{k,0,\delta}\right\rangle_{L^2(\X)}-\left\langle
	\Phi^{}_{j,\gamma}\,,X_\delta\,\Phi^{}_{k,0,\delta}\right\rangle_{L^2(\X)}
\end{align}
where we used \eqref{FD-R-epskappa-bot} in order to define:
$$
X^\epsilon_\delta:=(Y^\epsilon_\delta)^{-1/2}H^\epsilon_
\Gamma Q^\epsilon_{I,\delta}R^\epsilon_\bot(0)Q^\epsilon_{I,\delta}
H^\epsilon_\Gamma (Y^\epsilon_\delta)^{-1/2}.
$$
Considering formula \eqref{F-Y}, denoting by $\mathfrak{S}^\epsilon_T$ the magnetic symbol of a given operator $T$ (as in \eqref{N-magn-symb}) and dropping for the moment the subscript $\delta$ in order to ease the reading of the formulas, we can show that
$$
\mathfrak{S}^\epsilon_{Y^\epsilon}=p^\epsilon_I+p^\epsilon_I\,\sharp^\epsilon\, h\,\,\sharp^\epsilon\, q^\epsilon_I\,\,\sharp^\epsilon\,\mathfrak{S}^\epsilon_
{R^\epsilon_\bot(0)}\,\sharp^\epsilon\, q^\epsilon_I\,\sharp^\epsilon\, h\,\sharp^\epsilon\, p^\epsilon_I\ \in\ S^{-\infty}(\X\times\X^*).
$$
Indeed, the arguments following \eqref{F-RIepsilon} and Step 3 in the proof of Proposition \ref{P-Hyp-H-magn} imply that $\mathfrak{S}^\epsilon_{R^\epsilon_\bot(0)}\in S^{-2}_1(\X\times\X^*)$ and as $p^\epsilon_I$ belongs to $S^{-\infty}(\X\times\X^*)$, the above regularity of $\mathfrak{S}^\epsilon_{Y^\epsilon}$ follows. We conclude that
$$
\mathfrak{S}^\epsilon_{{Y_\epsilon}^{-1/2}H^{\epsilon}_I
	{Y_\epsilon}^{-1/2}}\ \in\ S^{-\infty}(\X\times\X^*),\quad\mathfrak{S}^\epsilon_{X^\epsilon}\in S^{-\infty}(\X\times\X^*).
$$

We shall use an estimation similar to Lemma 8.8 in \cite{CHP-2} which in our case of a constant magnetic field $\epsilon B^\circ$ has a more interesting form.
\begin{lemma}\label{P-red-OpF}
	Let us consider $\epsilon_0>0$ as large as allowed by the arguments in the proof of Proposition \ref{P-Hyp-H-magn} and for some $\epsilon\in[0,\epsilon_0]$ and for some $p<0$, an operator of the form
	$\mathfrak{Op}^{\epsilon}(F)$
	with $F\in S^p_1(\X\times\X^*)$. Then for
	any $m\in\mathbb{N}$, there exists a semi-norm $\nu_m$ such that,  $\forall \,
	(\alpha,\beta) \in \Gamma\times \Gamma$, $j=\pm$, $k=\pm$ and $\forall\epsilon\in[0,\epsilon_0]$, we have the estimate:
	$$
	\left|\left\langle\widetilde{\phi}^\epsilon_{j,\alpha,\delta}\,,
	\,\mathfrak{Op}^{\epsilon}(F)
	\widetilde{\phi}^\epsilon_{k,\beta,\delta}\right\rangle_{L^2(\X)}\,-\,
	\Lambda^\epsilon(\alpha,\beta)\left\langle\Phi_{j,\alpha,\delta},\,\mathfrak{Op}(F)
	\Phi_{k,\beta,\delta}\right\rangle_{L^2(\X)}\right| 
	\leq\ \epsilon\,\nu_m(F)<\alpha-\beta>^{-m}\,.
	$$
\end{lemma}

\begin{proof}
	Using Proposition B.8 in \cite{CHP-1} we may conclude that for 
	$F\in S^p_1(\X\times\X^*)$ and $\epsilon\in[0,1]$ the operator $\Op^\epsilon(F)$ has the integral kernel $\Lambda^\epsilon(x,y)\mathfrak{K}_F(x,y)$,
	where 
\begin{enumerate}
	    \item[(i)] $\mathfrak{K}_F(x,y)$ is the integral kernel of $\Op(F)$,
	    \item[(ii)] there exists a semi-norm $\tilde{\nu}$ such that $\|\Op^\epsilon(F)\|_{\mathbb{B}(L^2(\X))}\leq\tilde{\nu}(F)$,
	    \item[(iii)] for any $k\in\mathbb{N}$ there exists a semi-norm
	$\nu_k:S^p_1(\Xi)\rightarrow\mathbb{R}_+$ such that
	$$
	\underset{x\in\X}{\sup}\int_\X \,\langle x-y\rangle ^k|\mathfrak{K}_F(x,y)|\,dy \leq \nu_k(F) \, .
	$$
	\end{enumerate}
	Using Points 2 and 3 of Proposition \ref{P-magn-W-function} and the estimations (see \eqref{FD-magn-flux}):
		$$
	\left|\widetilde{\Omega} ^{\epsilon}(\alpha,x,\beta)-1\right|\leq
	C\, \epsilon\, |x-\alpha|\,|x-\beta|\,,\quad
	\left|\widetilde{\Omega} ^{\epsilon}(x,x+z,\beta)-1\right|\leq
	C\, \epsilon\, |z|\,|x-\beta|\,,
	$$
	we can prove that for any $N\in\mathbb{N}$ there exists a constant $C_N>0$ and a semi-norm $\widetilde{\nu}_N$ on $S^p_1(\X\times\X^*)$ such that
	\begin{align*}
		\left\langle\widetilde{\phi}^\epsilon_{j,\alpha,\delta},
		\mathfrak{Op}^{\epsilon}(F)
		\widetilde{\phi}^\epsilon_{k,\beta,\delta}\right\rangle_{L^2(\X)}&=
		\left\langle\Lambda^\epsilon(\cdot,\alpha)\tau_\alpha\widetilde{\psi}^\epsilon_{j,0,\delta}\,,
		\,\mathfrak{Op}^{\epsilon}(F)
		\Lambda^\epsilon(\cdot,\beta)\tau_\beta\widetilde{\psi}^\epsilon_{k,0,\delta}\right\rangle_
		{L^2(\X)} \\
		&\hspace*{-3cm}=\left\langle\tau_\alpha\widetilde{\psi}^\epsilon_{j,0,\delta}\,,
		\,\Lambda^\epsilon(\alpha,\cdot)\Int\big(\Lambda^{\epsilon}(F)\mathfrak{K}_F\big)
		\Lambda^\epsilon(\cdot,\beta)\tau_\beta\widetilde{\psi}^\epsilon_{k,0,\delta}\right\rangle_
		{L^2(\X)} \\
		&\hspace*{-3cm}=\Lambda^\epsilon(\alpha,\beta)\int_\X dx\int_\X dy\,
		\Omega^\epsilon(\alpha,x,\beta)\Omega^\epsilon(x,y,\beta)\mathfrak{K}_F(x,y)\Big(\overline{
			\widetilde{\psi}^\epsilon_{j,0,\delta}(x-\alpha)}\Big)\Big(\widetilde{\psi}^\epsilon_{j,0,\delta}(y-\beta)\Big) \\
		&\hspace*{-3cm}=\Lambda^\epsilon(\alpha,\beta)\int_\X dx\int_\X dy\,
		\mathfrak{K}_F(x,y)\Big(\overline{
			\Phi_{j,0,\delta}(x-\alpha)}\Big)\Big(\Phi_{j,0,\delta}(y-\beta)\Big)  +C_N\tilde{\nu}_N(F)\epsilon<\alpha-\beta>^{-N}.
	\end{align*}
	Thus we obtain the conclusion of the lemma.
\end{proof}

\begin{proposition}\label{P-kepsilon-k}
 With our choice of gauge for the magnetic field $\epsilon B^\circ$ in \eqref{defA0}, for any $a\in\mathbb{N}^2$ there exists a constant $C_a>0$ such that
\beq\label{F-est-kepsilon-k}
\big\|\big(\partial^a\mathcal{k}^\epsilon_{I,\delta}\big)(\theta)-\big(\partial^a\mathcal{k}_{I,\delta}\big)(\theta)\big\|_{\mathcal{M}_{2\times2}(\mathbb{C})}\ \leq\, C_{a} \, \epsilon,\qquad\forall\theta\in\mathbb{T}_*.
\eeq
\end{proposition}
\begin{proof}
The Proposition follows by applying two times the above Lemma to the right hand side of the equality \eqref{FD-kI0} for the two contributions given in \eqref{FD-kIO-1} and noticing that $\Lambda^\epsilon(\gamma,0)=1$.
\end{proof}

\subsection{Estimations close to the crossing-point.}\label{SSS-cr-est}
In this subsection we shall prove that on a small neighborhood of the point $\theta_0\in\mathbb{T}_*$, the symbol $\mathcal{k}_{I,\delta}(\theta)\in\mathcal{M}_{2\times2}(\mathbb{C}^2)$ is very well approximated by the matrix $M_I(\theta)$ in \eqref{F-HIcirc} associated to the local quasi-band Hamiltonian $H_{I}$.
Recalling that
$
\mathcal{k}_{I,\delta}(\theta)_{jk}:=\underset{\gamma\in\Gamma}{\sum}e^{-i<\theta,\gamma>}\left\langle\Phi_{j,\gamma,\delta},\widetilde{H}_{I,\delta}\Phi_{k,0,\delta}\right\rangle_{L^2(\X)}
$ we prove the following.

\begin{proposition}\label{P-p1T}
	For $\theta\in B_R(\theta_0)$ (as defined in \eqref{mai2}) and with $M_I$ introduced in \eqref{F-HIcirc} we have that
	$
	\mathcal{k}_{I,\delta}(\theta)=
	M_I(\theta)
	$.
\end{proposition}

\begin{proof}
	We shall proceed as in Subsection 8.4 of \cite{CHP-2} and we shall  only focus on the technical details that are specific to our new situation. We shall concentrate on the case $m=0$, the derivatives being easily controlled due to the smoothness conditions verified by the involved functions. 
	As in \cite{CHP-2} (but working directly for the case $\epsilon=0$) we introduce a \enquote{comparison term} defining the $2\times2$-matrix valued smooth map
	$$
	\mathring{\mathcal{k}}_{I,\delta}(\xi)_{jk}=\underset{\gamma\in\Gamma}{\sum}e^{-i<\xi,\gamma>}\big\langle\Phi_{j,\gamma,\delta},
	P_{I, \delta} H_\Gamma P_{I, \delta}\Phi_{k,0,\delta}\big\rangle_{L^2(\X)}
	$$
	and study its difference with both operators appearing in  Proposition \ref{P-p1T}.
	
		\noindent\textbf{Step 1:} Let us first notice that for any $\theta\in B_R(\theta_0)$ and for $j=\pm$ and $k=\pm$:
		\begin{equation*}
		\begin{array}{ll}
		\mathring{\mathcal{k}}_{I,\delta}(\theta)_{jk}
		&=\underset{\gamma\in\Gamma}{\sum}e^{-i<\theta,\gamma>}\big\langle\Phi_{j,\gamma,\delta},
		H_\Gamma \Phi_{k,0,\delta}\big\rangle_{L^2(\X)}\\
		&=\underset{\gamma\in\Gamma}{\sum}e^{-i<\theta,\gamma>}\int_{\mathbb{T}_*}d\omega\,\underset{\gamma\in\Gamma}{\sum}e^{-i<\omega,\gamma>}\big\langle \widehat{\Phi}_{j,\delta}(\omega)\,,\,\widehat{H}(\theta)
		\widehat{\Phi}_{k,\delta}(\omega)\big\rangle_{\mathfrak{F}_\theta}
		 \\
		&=\left\langle\widehat{\Phi}_{j,\delta}(\theta)\,,\,\widehat{\Pi}_I(\theta)\widehat{H}(\theta)\widehat{\Pi}_I(\theta)\widehat{\Phi}_{k,\delta}(\theta)
		\right\rangle_{\mathscr{F}_\theta} =\left\langle\hat{\varphi}_j(\theta)\,,\,\widehat{H}_I(\theta)\hat{\varphi}_k(\theta)
		\right\rangle_{\mathscr{F}_\theta}=M_I(\theta)_{jk}.
		\end{array}
		\end{equation*}
	
	\noindent\textbf{Step 2:} We consider now the difference
	$
	\mathcal{k}_{I,\delta}(\theta)\,-\,\mathring{\mathcal{k}}_{I,\delta}(\theta)$, for $\theta\in\Sigma_I$.
	Thus let us write
	\begin{align*}
		\mathcal{k}_{I,\delta}(\theta)_{jk}\,-\,\mathring{\mathcal{k}}_{I,\delta}(\theta)_{jk}\,&=\,\underset{\gamma\in\Gamma}{\sum}e^{-i<\theta,\gamma>}\left\langle\Phi_{j,\gamma}\,,\,\Big(\widetilde{H}_{I,\delta}-P_{I,\delta} H_\Gamma P_{I,\delta}\Big)\Phi_{k,0,\delta}\right\rangle_{L^2(\X)}.
	\end{align*}
As in \cite{CHP-2} we notice that by construction the bounded self-adjoint operator $Z_\delta:=\widetilde{H}_{I,\delta}-P_{I,\delta} H_\Gamma P_{I,\delta}$ commutes with the translations by vectors $\gamma\in\Gamma$ and thus decomposes with respect to the direct integral in the Floquet representation. The analysis presented in Step 2 of the proof of Proposition 8.13 in \cite{CHP-2} remains valid and we can write:
	\begin{align}
	    \label{final-II-0}
		Z_\delta=P_{I,\delta}Z_\delta P_{I,\delta} &=P_{I,\delta}H_\Gamma\big(Y_\delta^{-1/2}-P_{I,\delta}\big)+\big(Y_\delta^{-1/2}-P_{I,\delta}\big)H_\Gamma P_{I,\delta}+
		\big(Y_\delta^{-1/2}-P_{I,\delta}\big)H_\Gamma\big(Y_\delta^{-1/2}-P_{I,\delta}\big)\nonumber \\ & \quad \qquad 
		+Y_\delta^{-1/2}P_{I,\delta} H_\Gamma Q_{I,\delta} R_\bot Q_{I,\delta} H_\Gamma P_{I,\delta} Y_\delta^{-1/2}\,,
	\end{align}
	and recall from \eqref{F-Y} that 
	$Y_\delta-P_{I,\delta}:=P_{I,\delta} H_\Gamma Q_{I,\delta} R_\bot^2 Q_{I,\delta} H_\Gamma P_{I,\delta}$.
We have the following direct integral decomposition:
	$$
	Y_\delta-P_{I,\delta}=P_{I,\delta} H_\Gamma Q_{I,\delta} R_\bot^2 Q_{I,\delta} H_\Gamma P_{I,\delta}\equiv H_\bullet R_\bot^2H_\bullet=\mathscr{U}_\Gamma^{-1}\left(\int_{\mathbb{T}_*}^
	\oplus d\theta\,
	\widehat{Y}_\bullet(\theta)\right)\mathscr{U}_\Gamma,
	$$
	where
	$$
	\widehat{Y}_\bullet(\theta):=\widehat{H}_\bullet(\theta)\widehat{R}_\bot
	(\theta)^2\big[\widehat{H}
	_\bullet(\theta)\big]^*.
	$$
	 For $\theta\in\Sigma_I$ we notice that $\widehat{H}_\bullet(\theta)$ maps $\widehat{\Pi}_I(\theta)^\bot\mathscr{F}
	_\theta$ into $\widehat{\Pi}_I(\theta)\mathscr{F}_\theta$ (see \eqref{D-PiI}, \eqref{FD-Pi-delta} and \eqref{F-PiI} for the definition of $\widehat \Pi_I(\theta)$), being equal to
	$$ 
	\widehat{H}_\bullet(\theta)
	=\widehat \Pi_I(\theta) 
	\big(\underset{n\in\mathbb{N}}{\sum}\lambda_n(\theta)\widehat{\pi}
	_n(\theta)\big)\, \widehat \Pi_I(\theta)^\perp\,=\big(\widehat{\pi}_{k_0}(\theta)+\widehat{\pi}_{k_0+1}(\theta)\big)\big(\underset{n\in\mathbb{N}}{\sum}\lambda_n(\theta)\widehat{\pi}
	_n(\theta)\big)\big(\hspace*{-18pt}\underset{n\in\mathbb{N}\setminus\{k_0,k_0+1\}\}}{\sum}\hspace*{-18pt}\widehat{\pi}
	_n(\theta)\big).
	$$	
	We conclude that for $\theta\in \Sigma_I$ we have that $ \widehat{Y}_\bullet(\theta)=0$ and thus also $
	\widehat{Z}_\delta(\theta)=0$.
	Finally, these imply:
	$$
	\mathcal{k}_{I,\delta}(\theta)\,=\,\mathring{\mathcal{k}}_{I,\delta}(\theta),\qquad\forall\theta\in\Sigma_I.
	$$
\end{proof}

\subsection{Estimations far from the crossing-point.}\label{SSS-out-cr-est}

In this paragraph we prove that for $\theta\in\mathbb{T}_*\setminus B_R(\theta_0)$ the Hermitian $2\times2$ matrix 
$\mathcal{k}_{I,\delta}(\theta)$ has a spectral gap, i.e. two eigenvalues at a strictly positive distance one from the other, and we estimate this gap.

\begin{proposition}\label{P-restr-sp-window}
Let $R>0$ such that \eqref{mai2} is satisfied, there exists $\Lambda_R>0$ such that the interval $\left (-\Lambda_R,\Lambda_R\right ) $ is in the resolvent set of the Hermitian $2\times2$ matrix $\mathcal{k}_{I,\delta}(\theta)$ for any $\theta\in\mathbb{T}_*\setminus B_R(\theta_0)$.
\end{proposition}
\begin{proof}
Due to our Hypotheses \ref{H-1} and \ref{H-2}, for any value of $\theta \in\mathbb{T}_*\setminus B_R(\theta_0)$, the fiber operator $\widehat{H}(\theta)$ has a spectral gap $I_R=[-\widetilde{\Lambda}_R,\widetilde{\Lambda}_R]\subset I$. Let $I'\subset (-\widetilde{\Lambda}_R/2,\widetilde{\Lambda}_R/2)$ be any compact interval containing $0$ and notice that the distance $d_{I'}$ between $I'$ and $\R\setminus I$ is bounded from below by  $\widetilde{\Lambda}_R/2$. Also, let us notice that:
\begin{align}\label{mai6}
\sup_{\theta\in \mathbb{T}_*\setminus B_R(\theta_0)} \Vert \widehat{H}(\theta)\widehat{P}_{I,\delta}(\theta)\Vert =:C_R< \infty. 
\end{align}

We shall apply Proposition \ref{P-FS-arg} with the  matrix $\mathcal{k}_{I,\delta}(\theta)$ playing the role of the operator $\widetilde{H}$, the fiber Hamiltonian $\widehat{H}(\theta)$ playing the role of $H$, the projection $\widehat{P}_{I,\delta}(\theta)$ playing the role of $\Pi$ and $I_R$ playing the role of $I$. For every $e\in I'$ we have that $\dist(e,\sigma(\widehat{H}(\theta))\geq \widetilde{\Lambda}_R/2$. Now if the width $\ell_{I'}$ of the interval $I'$ is such that 
$$\widetilde{\Lambda}_R/2>  C_R^2 (2\ell_{I'}^2/\widetilde{\Lambda}_R)\geq C_R^2 (\ell_{I'}^2/\widetilde{\Lambda}_R^2) d_{I'}^{-1},$$
then \eqref{mai5} implies that $e$ is in the resolvent set of $\mathcal{k}_{I,\delta}(\theta)$. This implies that $$\left (-\frac{\widetilde{\Lambda}_R}{2C_R},\frac{\widetilde{\Lambda}_R}{2C_R}\right ) $$ belongs to the resolvent set of $\mathcal{k}_{I,\delta}(\theta)$ for all $\theta\in \mathbb{T}_*\setminus B_R(\theta_0)$. 
\end{proof}

\subsection{The symbol of the Peierls-Onsager effective Hamiltonian.}

Let us summarize what we know about $\mathcal{k}_{I,\delta}$ from the previous two Subsections \ref{SSS-cr-est} and \ref{SSS-out-cr-est}:

\begin{proposition}\label{P-kI}
	The smooth function $\mathbb{T}_*\ni\theta\mapsto\mathcal{k}_{I,\delta}
	(\theta)\in\mathcal{M}_{2\times2}(\mathbb{C}^2)$ defined in \eqref{FD-kI0} has the following properties:
	\begin{enumerate}
		\item $\forall\theta\in B_R(\theta_0)$ we have the identity $\mathcal{k}_{I,\delta}
		(\theta)=M_I(\theta)=F_0(\theta)\bb1\,+\,\underset{\ell =1,2,3}{\sum}F_\ell (\theta)\, \sigma_\ell$;
		\item There exists some $L>0$ such that for any $\theta\in\mathbb{T}_*\setminus B_R(\theta_0)$ the $2\times2$ complex matrix $\mathcal{k}_{I,\delta}(\theta)$ has two real eigenvalues
		$\kappa_{-,\delta}(\theta)<\kappa_{+,\delta}(\theta)$ satisfying the estimate
		$$
		\kappa_{-,\delta}(\theta)\,<\,-L\,<\,L\,<\,\kappa_{+,\delta}(\theta).
		$$
	\end{enumerate}
	\end{proposition}
As a direct consequence of this Proposition, we can define some smooth functions $(F_0,\mathbf{F}):\mathbb{T}_*\rightarrow\mathbb{R}\times\mathbb{R}^3$ so that 
$$
\mathcal{k}_{I,\delta}(\theta)=F_0(\theta)\bb1\,+\,\underset{\ell =1,2,3}{\sum}F_\ell (\theta)\, \sigma_\ell,\qquad\forall\theta\in\mathbb{T}_*
$$
and coinciding on $B_R(\theta_0)\subset\Sigma_I$ with the one defined in Subsection \ref{SS-LH}. While on $B_R(\theta_0)$ they do not depend on $\delta$ being equal to the functions defined in Subsection \ref{SS-LH}, outside this ball they will depend on our choices: the value of $\delta\in(0,\delta_0)$ and the extension of the smooth sections in Subsection \ref{SS-Gl-Sections.}. We shall have the same relations:
$$
F_0(\theta)=(1/2)\big(\kappa_{-,\delta}(\theta)+\kappa_{+,\delta}(\theta)\big),\quad\big|\mathbf{F}(\theta)\big|=(1/2)\big|\kappa_{-,\delta}(\theta)-\kappa_{+,\delta}(\theta)\big|,\qquad\forall\theta\in\mathbb{T}_*.
$$

\begin{definition}\label{D-symbol-POeffH}
	We call $\Op^\epsilon(\mathcal{k}_{I,\delta})$ the Peierls-Onsager effective Hamiltonian for the spectral region $I$, considered as a $\Gamma_*$-periodic symbol in $S^0_0(\X\times\X^*)_{2\times2}$ constant along the $\X$ directions.
\end{definition}

Once defined a Peierls-Onsager effective Hamiltonian (Definition \ref{D-symbol-POeffH}) with a symbol that is close of order $\epsilon$ to the operator $\Op^\epsilon(\mathcal{k}^\epsilon_{I,\delta})$ (Proposition \ref{P-kepsilon-k}) that is unitary equivalent with the magnetic 'quasi-band' Hamiltonian (Proposition \ref{P-magn-matrix-ke}), the proof of our main Theorem \ref{T-Main-2} may be obtained by using the magnetic pseudo-differential calculus to construct a 'quasi-resolvent' as in our papers \cite{CHP-1} and \cite{CHP-2}. Nevertheless, taking advantage of the fact that we deal with a constant magnetic field $\epsilon B^\circ$, we have chosen to avoid these rather complicated computations that where necessary in the case of a non-constant magnetic field and transform our problem in a semi-classical problem for a Weyl type pseudo-differential operator (as in \cite{HS1} and \cite{HS2})} and use the arguments developed  in these references. The next section is devoted to this analysis.

\section{Semi-classical analysis for the Peierls-Onsager effective Hamiltonian.}\label{S-red-ModelH}
The aim of this section is to present results which were developed in \cite{HS1} and \cite{HS2} in the context  of the Harper model and apply them to the Peierls-Onsager effective Hamiltonian $\Op^\epsilon(\mathcal{k}_{I,\delta})$ associated above to our magnetic quasi-band Hamiltonian for the spectral interval $I\subset\mathbb{R}$. Although most of the analysis there is rather general, the case of \enquote{touching bands} is only solved in 
a particular situation and there is a need to detail some new aspects. We include here the analogue of \cite{HS2} in our more general context. Throughout this section we shall consider the smooth map $\mathbb{T}_*\ni\theta\mapsto\mathcal{k}_{I,\delta}(\theta)\in\mathcal{M}_{2\times2}(\mathbb{C})$ as a smooth $\Gamma_*$-periodic function $\mathbb{R}^2\ni\xi\mapsto\mathcal{k}_{I,\delta}(\xi)\in\mathcal{M}_{2\times2}(\mathbb{C})$.

\subsection{Formulation of the one-dimensional \textit{semi-classical} problem.}\label{ss7.1new}

A few elements associated to the symplectic structure will be used in this subsection. Let us recall that $\Xi=\X\times\X^*$ has a canonical symplectic structure given by the symplectic form 
$$
\sigma\big((x,\xi),(y,\eta)\big):=<\xi,y>-<\eta,x>\equiv\langle(x,\xi)\,,\,\mathfrak{J}(y,\eta)\rangle_\Xi,\;\forall\big((x,\xi),(y,\eta)\big)\in\Xi\times\Xi,
$$
\beq\label{DF-J}
\ \mathfrak{J}:=\left(\begin{array}{cc}
	0&-\bb1_2\\\bb1_2&0
\end{array}\right).
\eeq
Let us write down the explicit form of our Peierls-Onsager effective Hamiltonian. 
\begin{align*}\begin{split}
&\forall f\in\mathscr{S}(\X)\oplus\mathscr{S}(\X)\subset L^2(\X)\otimes\mathbb{C}^2: \\
&\hspace*{1cm}\big(\Op^\epsilon(\mathcal{k}_{I,\delta})f\big)_j(x)=(2\pi)^{-2}\int_{\X}dy\,e^{-(i\epsilon B^\circ/2)x\wedge y}\int_{\X^*}d\xi\,e^{i<\xi,x-y>}\underset{k=\pm}{\sum}\,\mathcal{k}_{I,\delta}(\xi)_{jk}\,f_k(y).
\end{split}\end{align*}
If we consider the orthogonal transformation:
$$
\mathbb{R}^2\ni x\mapsto\mathfrak{J}x:=(-x_2,x_1)\in\mathbb{R}^2
$$
and notice that $-(\epsilon B^\circ/2)x\wedge y\,+<\xi,x-y>=\big<(\xi+\epsilon B^\circ(\mathfrak{J}x)/2,(x-y)\big>$ we conclude after a change of variables $\xi\mapsto\xi-\epsilon B^\circ(\mathfrak{J}x) /2$, that
\beq\label{F-5.4}
\big(\Op^\epsilon(\mathcal{k}_{I,\delta})f\big)_j(x)=(2\pi)^{-2}\int_{\X}dy\,\int_{\X^*}d\xi\,e^{i<\xi,x-y>}\underset{k=\pm}{\sum}\,\mathcal{k}_{I,\delta}\big(\xi-\epsilon B^\circ(\mathfrak{J}x) /2\big)_{jk}\,f_k(y)
\eeq
We notice that we can view the above change of variables as a linear map $\Xi\ni X\mapsto\widetilde{Z}_1(X):=\big(\widetilde{z}_1(X),\widetilde{\zeta}_1(X)\big)\in\Xi$ with
\begin{align*}\left\{\begin{array}{l}
\widetilde{z}_1(X):=\xi_2-(\epsilon B^\circ/2)x_1,\\
\widetilde{\zeta}_1(X):=\xi_1+(\epsilon B^\circ/2)x_2
\end{array}\right.
\end{align*}
and we remark that
$$
\sigma\big(\widetilde{z}_1,\widetilde{\zeta}_1\big)=\sigma\big(\xi_2-(\epsilon B^\circ/2)x_1,\xi_1+(\epsilon B^\circ/2)x_2\big)=2(\epsilon B^\circ/2)=\epsilon B^\circ.
$$
Let us denote by $\h:=\epsilon B^\circ$ and extend this rank 2 linear map to a symplectic linear map $\Xi\ni X\mapsto\widetilde{X}\in\Xi$ with determinant 1:
\begin{align*}\left\{\begin{array}{l}
\widetilde{x}_1:=\xi_2-(\h/2)x_1,\\
\widetilde{x}_2:=-\xi_2-(\h/2)x_1,\\
\widetilde{\xi}_1:=\h^{-1}\xi_1+(1/2)x_2, \\
\widetilde{\xi}_2:=\h^{-1}\xi_1-(1/2)x_2.
\end{array}\right.
\end{align*}
having the determinant
\begin{align*}\left|\begin{array}{cccc}
-\h/2&0&0&1\\
-\h/2&0&0&-1\\
0&1/2&\h^{-1}&0\\
0&-1/2&\h^{-1}&0
\end{array}\right|=1.
\end{align*}
We shall denote by $\widetilde{\X}$ the subspace generated by $\big(\widetilde{x}_1(X),\widetilde{x}_2(X)\big)$ when $X\in\Xi$ and by $\widetilde{\X}^*$ the one generated by $\big(\widetilde{\xi}_1(X),\widetilde{\xi}_2(X)\big)$ when $X\in\Xi$. The abstract theory of symplectic maps implies (see chapter 7 in \cite{Go}) the existence of a unitary operator (element of the metaplectic group, determined up to a sign) $\mathfrak{U}^\epsilon:L^2(\X)\overset{\sim}{\rightarrow}L^2(\widetilde{\X})$ such that, if we denote by $\Op$ the Weyl calculus associated to the decomposition $\Xi=\X\times\X^*$ and by $\widetilde{\Op}$ the Weyl calculus associated to the decomposition $\Xi=\widetilde{\X}\times\widetilde{\X}^*$, related by the symplectic transform $S^\epsilon:\Xi\ni X\mapsto\widetilde{X}(X)\in\Xi$, we have the identity:
\beq\label{F-Egor}
\mathfrak{U}^\epsilon\Op\big(\Phi\circ S^\epsilon\big)\big[\mathfrak{U}^\epsilon\big]^{-1}=\widetilde{\Op}(\Phi),\quad\forall\Phi\in\mathscr{S}^\prime(\Xi).
\eeq
We apply this remark to our symbol $\mathcal{k}_{I,\delta}$ and write that
$$
\mathcal{k}_{I,\delta}\big(\xi-\epsilon B^\circ(\mathfrak{J}x)\big)=
\mathcal{k}_{I,\delta}(\widetilde{x}_1,\h\widetilde{\xi}_1)\otimes\bb1(\widetilde{x}_2,\h\widetilde{\xi}_2)=:\widetilde{\mathcal{K}}_{I,\delta}\big(S^\epsilon(X)\big)
$$
so that \eqref{F-5.4} and \eqref{F-Egor} become
$$
\Op^\epsilon\big(\mathcal{k}_{I,\delta}\big)=\Op\big(\widetilde{\mathcal{K}}_{I,\delta}\circ S^\epsilon\big)=
\big[\mathfrak{U}^\epsilon\big]^{-1}\widetilde{\Op}\big(\widetilde{\mathcal{K}}_{I,\delta}\big)\mathfrak{U}^\epsilon.
$$
This formulas suggest to work in the following decomposition of $\Xi$:
$$
\Xi=\Xi_1\times\Xi_2=\Big(\widetilde{\X}_1\times\widetilde{\X}^*_1\Big)\times\Big(\widetilde{\X}_2\times\widetilde{\X}^*_2\Big),\quad X:=\big(\widetilde{X}_1,\widetilde{X}_2\big)=\big((\widetilde{x}_1,\widetilde{\xi}_1),(\widetilde{x}_2,\widetilde{\xi}_2)\big)
$$
so that for any $\varphi\in\mathscr{S}(\widetilde{\X})$: 
\beq
\Big(\widetilde{\Op}\big(\widetilde{\mathcal{K}}_{I,\delta}\big)\varphi\Big)(\widetilde{x}_1,\widetilde{x}_2)=(2\pi)^{-1}\int_{\widetilde{X}_1}d\widetilde{y}_1\int_{\widetilde{X}^*_1}d\widetilde{\xi}_1\,e^{i<\widetilde{\xi}_1,(\widetilde{x}_1-\widetilde{y}_1)>}\mathcal{k}_{I,\delta}\big((\widetilde{x}_1-\widetilde{y}_1)/2,\h\widetilde{\xi}_1\big)\big)\varphi(\widetilde{y}_1,\widetilde{x}_2).
\eeq
Thus we are reduced to a 1-dimensional semi-classical problem for a  $2\times 2$ matrix valued symbol, with semi-classical parameter $\h=\epsilon B^\circ$ controlled by the intensity of the magnetic field. During this section we shall denote simply $\widetilde{\Xi}_1\cong\mathbb{R}^2$, $\widetilde{\X}_1\cong\mathbb{R}$, $\widetilde{\X}^*_1\cong\mathbb{R}$ and the variables $\widetilde{x}_1=t$, $\widetilde{y}_1=s$, $\widetilde{\xi}_1=\tau$ and $T=(t,\tau)\in\mathbb{R}^2$. We shall also use the notation
\beq\label{DF-kIh}
\mathcal{k}_{I,\delta,\h}(t,\tau):=\mathcal{k}_{I,\delta}(t,\h\tau).
\eeq
We shall consider the usual Weyl quantization of symbols on $\mathbb{R}^2$:
\beq\label{DF-Op1}
\big(\Op_1(f)\varphi\big)(t):=(2\pi)^{-1}\int_{\mathbb{R}}ds\int_{\mathbb{R}}d\tau\,e^{i<\tau,(t-s)>}f\big((t+s)/2,\tau\big)\varphi(s)
\eeq
and state the conclusion of the above construction.
\begin{proposition}
With $\h:=\epsilon B^\circ>0$ and the notations introduced in \eqref{F-Egor}, \eqref{DF-Op1} and \eqref{DF-kIh} we have the identities:
$$
\Op^\epsilon\big(\mathcal{k}_{I,\delta}\big)=\big[\mathfrak{U}^\epsilon\big]^{-1}\left(\Op_1(\mathcal{k}_{I,\delta,\h}\big)\otimes\bb1_{L^2(\mathbb{R})}\right)\mathfrak{U}^\epsilon.
$$
\end{proposition}

\begin{corollary}\label{C-equi-scl}
The self-adjoint operator $\Op^\epsilon\big(\mathcal{k}_{I,\delta}\big)$ acting in $L^2(\X)$ and $\Op_1(\mathcal{k}_{I,\delta,\h}\big)$ acting in $L^2(\mathbb{R})$ have equal spectra.
\end{corollary}

In fact we shall prefer to work with a \enquote{symmetric parametrization} for $\mathbb{R}^2$ and consider the  unitary dilation operators 
$$
L^2(\mathbb{R})\ni f \mapsto \mathfrak{d}_r f\in L^2(\mathbb{R})  \mbox{  with } \big(\mathfrak{d}_r f\big) (t) =r^\frac{1}{2} f(rt),\quad\mbox{for } r>0.
$$
Then 
\beq\label{F-sym-repr}
\mathfrak{d}_{\h^{1/2}}\Op_1\big(\mathcal{k}_{I,\delta,h}\big)\mathfrak{d}_{\h^{1/2}}^{-1}=\Op_1\big(\widetilde{\mathcal{k}}_{I,\delta,h}\big),\quad\mbox{where }\widetilde{\mathcal{k}}_{I,\delta,h}(t,\tau):=\mathcal{k}_{I,\delta}(\h^{1/2}t,\h^{1/2}\tau).
\eeq

\begin{corollary}\label{C-isosp-scl}
	The self-adjoint operator $\Op^\epsilon\big(\mathcal{k}_{I,\delta}\big)$ acting in $L^2(\X)$ and $\Op_1(\widetilde{\mathcal{k}}_{I,\delta,\h}\big)$ acting in $L^2(\mathbb{R})$ have equal spectra.
\end{corollary}

Thus, from now on we consider $\mathcal{k}_{I,\delta}$ as symbol on the symplectic linear space $\mathbb{R}^2=\mathbb{R}\times\mathbb{R}$ with the symplectic structure $\sigma_1\big((t,\tau),(t',\tau')\big):=\tau t'-\tau' t$. We recall that for any $\xi\in\mathbb{R}^2$ we have that $\mathcal{k}_{I,\delta}(\xi)\in\mathcal{M}_{2\times2}(\mathbb{C})$ and from Proposition \ref{P-kI} we see that its eigenvalues $\kappa_{\pm,\delta}(\theta)$ satisfy:
\begin{itemize}
	\item $\kappa_{-,\delta}(\xi)<-L<0<L<\kappa_{+,\delta}(\xi)$ for any $\xi\in\mathbb{R}^2\setminus\left\{\underset{\gamma^*\in\Gamma_*}{\bigcup} B_R(\theta_0+\gamma^*)\right\}$, for some $L>0$.
	%that does not depend on $\delta>0$.
	\item $\kappa_{-,\delta}(\xi)=\kappa_{+,\delta}(\xi)$ for $\xi=\theta_0+\gamma^*$ for any $\gamma^*\in\Gamma_*$.
\end{itemize}

\subsection{Decoupling the lattice of crossing-points.}

The aim of this section is the spectral analysis of the spectrum of $\Op_1\big(\mathcal{k}_{I,\delta,\h}\big)$ in a neighborhood of $0\in\mathbb{R}$. Thus, the interesting region is the neighborhood of the points $\big\{\theta_0+\gamma^*\big\}_{\gamma^*\in\Gamma_*}\subset\mathbb{R}^2$.  We shall continue our analysis following \S2 in \cite{HS1}, thus we shall locally perturb our symbol $\mathcal{k}_{I,\delta}$ in order to generate a $\Gamma_*$-indexed family of symbols  $\big\{\mathcal{k}_{I,\delta}^{\gamma_*}\big\}_{\gamma^*\in\Gamma_*}$ that have \textit{only one crossing-point} in $\theta_0+\gamma^*$ for some $\gamma_*\in\Gamma_*$, being uniformly elliptic outside any neighborhood of this fixed point $\theta_0+\gamma^*\in\mathbb{R}^2$. 

 In order to obtain this family of \enquote{one crossing-point} symbols we apply the procedure developed in Subsection  \ref{SSS-eigval-sep} and locally introduce a gap near all the the crossing points but the one in $\theta_0+\gamma^*$ in order to obtain:
\begin{proposition}
Let $c>0$ be the lower bound in \eqref{aprilie2} and $g\in  C^\infty_0(\mathbb{R}^2)$ defined at the beginning of Subsection \ref{SSS-eigval-sep}. There exist some $d_0>0$, $d_1>0$ both small enough and $C>0$ only depending on the cut-off function $g$ and the form of $\mathcal{k}_{I,\delta}$ close to the crossing point $\theta_0$ such that the symbol 
\beq\label{DF-kcircIdelta}
\mathcal{k}^{\gamma^*}_{I,\delta}(T)\ :=\ \mathcal{k}_{I,\delta}(T)\,+\,\underset{\alpha^*\in\Gamma_*\setminus\{\gamma^*\}}{\sum}(d_0/8) g ((T-\alpha^*-\theta_0)c/d_0)\sigma_{v^{(3)}}
\eeq
has a spectral gap $(-d_0/C,d_0/C)$ for any $T\in\mathbb{R}^2$ outside a ball of radius of order $d_1$ centered at $\theta_0+\gamma^*$.
\end{proposition}
We shall denote by $\mathcal{k}^\circ_{I,\delta}$ the \enquote{one crossing-point} symbol with $\gamma^*=0\in\Gamma_*$.

\subsection{Spectral analysis of the \enquote{one crossing point} extended symbol} \label{SS-1-p-scls}

The defining formula \eqref{DF-kcircIdelta} shows that for $T\in B_R(\theta_0)$ the 'one crossing point' modified symbol $\mathcal{k}^\circ_{I,\delta}(T)$ coincides with the '$\Gamma_*$-periodic' one $\mathcal{k}_{I,\delta}(T)$ and this one, when projected on the torus $\mathbb{T}_*$ does not depend on $\delta>0$ and coincides with $M_I(\theta)$ in \eqref{F-HIcirc}. In fact, the two important properties of $\mathcal{k}^\circ_{I,\delta}(T)$ that will be used in the following analysis are:
\begin{itemize}
	\item its restriction to $B_R(\theta_0)$
	\item the minimal width of its spectral gap outside $B_R(\theta_0)$
\end{itemize}
and both these elements do not depend on $\delta>0$. Having this in mind and in order to simplify the notations in the coming arguments, in this subsection we shall use the short-hand notations 
\begin{align}\label{F-not-simpl-k}
\forall T=(t,\tau)\in\mathbb{R}^2:\quad& 
\mathcal{k}(T)\equiv\mathcal{k}^\circ_{I,\delta}(T+\theta_0),\\ &\mathcal{k}_\h(T)\equiv\mathcal{k}^\circ_{I,\delta,\h}(T+\theta_0)=\mathcal{k}^\circ_{I,\delta}(t+\theta_{0,1},\h(\tau+\theta_{0,2}))\nonumber \\ \label{F-not-simpl-tildek} &\widetilde{\mathcal{k}}_\h(T)\equiv\widetilde{\mathcal{k}}^\circ_{I,\delta,\h}(T+\theta_0)=\mathcal{k}^\circ_{I,\delta}(\h^{1/2}(T+\theta_0)).
\end{align}
As stated from the beginning of our paper, our spectral analysis is based on the approximation of the spectrum in a small neighborhood of 0 contained in $I\subset\mathbb{R}$ by the spectrum in the given neighborhood of 0 of the linearisation of the Peierls-Onsager effective Hamiltonian $\Op^\epsilon(\mathcal{k}_{I,\delta})$, and this has the same spectrum with $\Op_1(\widetilde{\mathcal{k}}_\h)$ (with $\h=\epsilon B^\circ$), as seen in Corollary \ref{C-isosp-scl}. Thus let us consider the formal Taylor expansion of the matrix valued symbol $\widetilde{\mathcal{k}}_\h$ near $T=0$ and recast this expansion in powers of $\h^{1/2}$:
\begin{align}\label{F-T-dev-tildek}
\widetilde{\mathcal{k}}_{\h}(t,\tau)\sim &\h^{1/2}\mathfrak{l} (t,\tau) + \sum_{\ell \in \mathbb N, \ell > 1}\h^{\ell/2}\, \hat{\mathcal{k}}^\ell (t,\tau),\\\nonumber 
&\mathfrak{l}(t,\tau)):=t\Big[\left.\partial_t\Big(F_0\big(\theta_0+(t,\tau)\big)\bb1_2\,+\,\underset{\ell =1,2,3}{\sum}\mathbf{F}_\ell \big(\theta_0+(t,\tau)\big)\, \sigma_\ell\Big)\Big]\right|_{(t,\tau)=(0,0)}\\\nonumber 
&\hspace*{36pt}+
\tau\Big[\left.\partial_\tau\Big(F_0\big(\theta_0+(t,\tau)\big)\bb1_2\,+\,\underset{\ell =1,2,3}{\sum}\mathbf{F}_\ell \big(\theta_0+(t,\tau)\big)\, \sigma_\ell\Big)\Big]\right|_{(t,\tau)=(0,0)} \\\nonumber 
&\hspace*{32pt}=\big(f_1\bb1_2+\sigma_{v^{(1)}}\big)t\,+\,\big(f_2\bb1_2+\sigma_{v^{(2)}}\big)\tau
\end{align}
where  $\hat{\mathcal{k}}^\ell (t,\tau)$ is a homogeneous polynomial of degree $\ell$ in $(t,\tau)\in\mathbb{R}^2$ and we have used the notations in \eqref{F-desc-F-v}.
Writing the formula for the remainder in the $k$-th order Taylor formula we have the following control of this formal series as an asymptotic expansion.
\begin{definition}\label{D-Tdev-tildek}
Let us define: 
\begin{align*}
r_1(t,\tau,\h) &:=\widetilde{\mathcal{k}}_{\h}(t,\tau)- \h^{1/2}\mathfrak{l}(t,\tau)), \\
r_k(t,\tau,\h) &:=\widetilde{\mathcal{k}}_{\h}(t,\tau)- \h^{1/2}\mathfrak{l}(t,\tau)) - \sum_{k-1\geq \ell > 1} \hat{\mathcal{k}}^\ell (t,\tau)\, \h^{\ell/2},\qquad\forall k\geq2.
\end{align*}
\end{definition}
\begin{proposition}\label{P-Tdev-tildek}
For $k\in\mathbb{N}\setminus\{0\}$ we have that $r_k(t,\tau,\h)\in C^\infty_{\text{\sf pol}}\big(\mathbb{R}^2\times\mathbb{R}_+;\mathcal{M}_{2\times2}(\mathbb{C})\big)$ and we have the estimations:
\begin{equation}\label{estsurr}
||\big(\partial_t^\ell \partial_\tau^m r_k\big)(t,\tau,\h) || \leq C_{\ell,m}\, \h^{(k+1)/2} (t^2 + \tau^2)^{[(k+1-\ell - m)/2]_+} \,.
\end{equation}
\end{proposition}

\subsubsection{The linearised 1-dimensional differential operator.}

We are interested in the spectral analysis of the differential operator in $L^2(\mathbb{R})$ associated to the linear symbol $\mathfrak{l}(t,\tau)$ and given explicitly by
\begin{align*}
\big(\Op_1\big(\mathfrak{l}\big)\varphi\big)(t)&=\frac{1}{2\pi}\left(\int_\mathbb{R}ds\int_{\mathbb{R}}d\tau\,e^{i<\tau,(t-s)>}\big(f_1\bb1_2+\sigma_{v^{(1)}}\big)\frac{t+s}{2}\varphi(s)\right .\\
&\quad \quad \left .\,+\,\int_\mathbb{R}ds\int_{\mathbb{R}}
d\tau\,e^{i<\tau,(t-s)>}\big(f_2\bb1_2+
\sigma_{v^{(2)}}\big)\tau\varphi(s)\right) \\
&=\big(f_1\bb1_2+\sigma_{v^{(1)}}\big)t\varphi(t)-\big(f_2\bb1_2+\sigma_{v^{(2)}}\big)\big(i\partial_t\varphi\big)(t).
\end{align*}
By construction $\mathfrak{l}(t,\tau)$ is a linear function $\mathbb{R}^2\ni(t,\tau)\rightarrow\mathcal{M}_{2\times2}(\mathbb{C})$ with hermitian values, so that its determinant defines a quadratic function $\mathbb{R}^2\ni(t,\tau)\rightarrow\mathbb{R}$. Thus there exists a matrix $\mathfrak{D}_0\in\mathcal{M}_{2\times2}(\mathbb{R})$ such that
$$
\det\big(\mathfrak{l}(t,\tau)\big)=\left\langle\big(t,\tau\big)\,,\,\mathfrak{D}_0\left(\begin{array}{c}
t\\ \tau
\end{array}\right)\right\rangle_{\mathbb{R}^2}.
$$
Moreover, Proposition \ref{P-kI} tells us that $\mathcal{k}((t,\tau))=M_I(\theta_0+(t,\tau))$
so that for $\sqrt{|t|^2+|\tau|^2}$ small enough
\begin{align*}
\det\big(M_I(\theta_0+(t,\tau))\big)&=\det\big(\mathfrak{l}(t,\tau)\big)\,+\,\mathcal{O}(|t|^2+|\tau|^2) \\
&=\lambda_-(\theta_0+(t,\tau))\lambda_+(\theta_0+(t,\tau)),
\end{align*}
and Hypothesis \ref{H-2} tells that its Hessian is negative definite near its minimum 0 at $(t,\tau)=(0,0)$. But it is a quadratic function of $(t,\tau)$ and thus it must be a negative definite one, i.e.
\beq\label{F-est-det-k}
\exists a_0>0,\quad\det\big(\mathfrak{l}(t,\tau)\big)=\left\langle\big(t,\tau\big)\,,\,\mathfrak{D}_0\left(\begin{array}{c}
	t\\ \tau
\end{array}\right)\right\rangle_{\mathbb{R}^2}\,\leq\,(-a_0)\big(|t|^2+|\tau|^2\big).
\eeq

We conclude that $\mathfrak{l}$ is a globally elliptic symbol (in the sense of \cite{He-84}). 
The theory of globally elliptic operators (as presented in \cite{He-84} mainly on the basis of results with D. Robert, see \cite{He-84} and references therein) allows us to get a series of important conclusions.
\begin{theorem}\label{T-HR} (see \cite{He-84})
	\begin{enumerate}
		\item There exists $\mathfrak{p}(t,\tau) \in C^\infty (\mathbb R^2;\mathbb C)$ such that for any non negative integers  $\ell, m$ there exists $C_{\ell m}$ such that
		$$ |\partial_t^\ell\partial_\tau^m\mathfrak{p}(t,\tau)| \leq C_{\ell m} (1+ t^2 +\tau^2)^{-(1+\ell + m)/2}  \mbox{ for } (t,\tau)\in \mathbb R^2\,.
		$$
		\begin{equation*}
		\mathfrak{p}(t,\tau) = \mathfrak{l}(t,\tau)^{-1} \mbox{ for } t^2 +\tau^2 \geq 1\,.
		\end{equation*}
\item $\Op_1(\mathfrak{l})$ is an essentially self-adjoint operator on the domain $\mathscr S (\mathbb R;\mathbb C^2)$ and its closure  $\mathfrak{L} =\overline{\Op_1(\mathfrak{l})}$ as an unbounded self-adjoint operator in $L^2(\mathbb R;\mathbb C^2)$ has a compact resolvent and hence a discrete sequence of eigenvalues with finite multiplicities.
		\item The domain of $\mathfrak{L}$ is $B^1(\mathbb R,\mathbb C^2) = \{ u \in L^2(\mathbb R;\mathbb C^2), x u  \in L^2, u' \in L^2\}.$
		\item All the eigenfunctions are in $\mathscr S(\mathbb R;\mathbb C^2)$. 
		\item The eigenvalues of $\mathfrak{L}$ have multiplicity one.
		\item For any $\lambda \notin \sigma(\mathfrak{L})$, $(\mathfrak{L}-\lambda)^{-1}$ is a pseudodifferential operator with a symbol $\mathfrak{r}(\lambda)\in C^\infty (\mathbb R^2;\mathbb C)$ such that for any non negative integers  $\ell, m$ there exists $C_{\ell m}$ such that
		$$ |\partial_t^\ell\partial_\tau^m\mathfrak{r}(t,\tau)| \leq C_{\ell m} (1+ t^2 +\tau^2)^{-(1+\ell + m)/2}  \mbox{ for } (t,\tau)\in \mathbb R^2\,
		$$ which depends holomorphically on $\lambda$ in each domain of holomorphy of the resolvent of $\mathfrak{L}$.
		\item  If  $\lambda$ is an eigenvalue of $\mathfrak{L}$, then $-\lambda$ is also an eigenvalue of $\mathfrak{L}$.
	\end{enumerate}
\end{theorem} 
Nothing is related to the dimension one, except the fifth item which is a consequence of the  Cauchy uniqueness theorem 
together with \eqref{F-est-det-k}.  The last item is a consequence of the property that the symbol is antisymmetric
\begin{equation*}
\mathfrak{l}(t,\tau) =- \mathfrak{l}(-t,-\tau)\,.
\end{equation*}
We then observe that if $(\lambda,u)$ is a spectral pair for $\mathfrak{L}$ then $(-\lambda, \mathfrak{I} u)$ is also a spectral pair
with $\big(\mathfrak{I}u\big)(x):= u(-x)$.

\begin{remark}\label{R-first-order}
If we  now consider an eigenvalue $\lambda_1$ of $\mathfrak{L}$ and let $v_0$ be a normalized eigenfunction, we  observe (Fredholm alternative)  that the equation
$$
(\mathfrak{L}-\lambda_1) v= f\,, f\in \mathscr S(\mathbb R), v \in \mathscr S(\mathbb R)\,,
$$
has a solution if and only if $\int_\mathbb R f(t) \bar v_0(t)\, dt =0$. Moreover the solution is unique if we impose $\int_\mathbb R v(t) \bar v_0(t)\, dt =0$.
\end{remark}

\subsubsection{Quasi-modes near $(0,0)$}\label{SS-qmodes00}
We come back to the complete symbol  $\mathcal{k}^\circ_{I,\delta}$ in \eqref{DF-kcircIdelta} and its associated symmetric form (as in \eqref{F-sym-repr}) and  use the simplified notation in \eqref{F-not-simpl-k} and \eqref{F-not-simpl-tildek}.
Recalling Definition \ref{D-Tdev-tildek} and Proposition \ref{P-Tdev-tildek} we have the following asymptotic expansions of the spectral elements of $\Op_1(\mathcal{k}_\h)$ associated to the spectral elements of the self-adjoint operator $\mathfrak{L}$.

\begin{proposition}
	There exist  two infinite sequences $\{\lambda_\ell(\h)\}_{\ell\geq1}\subset\mathbb{R}$ and $\{v_\ell\}_{\ell\in\mathbb{N}} \subset \mathscr S(\mathbb R)$ 
	such that, for any $k\geq1$, we have 
	\begin{equation}\label{eq:8.55}
	\Big(\big(\Op_1(\mathcal{k}_{\h})-\lambda^{(k)} (\h))  u^{(k-1)}\Big)(t,\h)  = \mathcal{O} (h^{(k+1)/2})\,, 
	\end{equation}
	in $L^2(\mathbb R)$ with 
	$$
	\lambda^{(k)}(h) = \sum_{1 \leq \ell \leq k} h^{\ell /2} \lambda_\ell\,,\, u^{(k-1)}(t,\h) = \h^{-1/4} \sum_{0 \leq \ell \leq k-1} \h^{\ell /2} v_\ell (h^{-1/2} t)\,.
	$$
\end{proposition}
\begin{proof}
	Recalling the unitary equivalence of \eqref{DF-kcircIdelta} and \eqref{F-sym-repr}, it is enough to expand $\widetilde{\mathcal{k}}_{\h}$ in powers of $\h^{1/2}$, as in \eqref{F-T-dev-tildek} and write formal power series with respect to $\h^{1/2}$ for the spectral elements $\lambda\in\mathbb{R}$ as eigenvalue and $v (t,\h):=\mathfrak{d}_{\h^{1/2}}u\in L^2(\mathbb{R})$:
	$$v (t,\h):=\mathfrak{d}_{\h^{1/2}}u \sim \sum_{\ell \geq 0}  v_\ell (t) \h^{\ell/2}\,,\, \lambda(h) \sim  \sum_{1 \leq \ell } \h^{\ell /2} \lambda_\ell,
	$$
	\begin{equation}\label{eq:8.56}
	\Big(\big(\Op_1((\widetilde{\mathcal{k}}_{\h})-\lambda (h)\big)  v\Big)(t,\h) \sim 0.
	\end{equation}
	This allows us to write down the sequence of equations describing the annulation of the coefficients of the powers of $\h^{1/2}$.
	The first equation, for the coefficient of $\h^{1/2}$, reads
	$$
	\big(\Op_1(\mathfrak{l}) - \lambda_1) v_0 =0\,.
	$$
	This is satisfied by our choice of $(\lambda_1,v_0)$ in Remark \ref{R-first-order}.
	
	We now consider the coefficient of $\h$ and get
	\beq\label{F-1-st-order}
	\big(\Op_1(\mathfrak{l}) - \lambda_1\big) v_1 = - \Op_1(\hat{\mathcal{k}}^2) v_0+ \lambda_2  v_0\,.
	\eeq
	As noticed in Remark \ref{R-first-order} we can solve this equation if and only if
	$$
	\lambda_2= <v_0, \Op_1(\hat{\mathcal{k}}^2) v_0>\,,
	$$ 
	and this defines $\lambda_2$. We can then choose $v_1 \in \mathscr S(\mathbb R)$ the unique solution of \eqref{F-1-st-order}
	orthogonal to $v_0$.
	
	Let us go one step further and consider the equation obtained by the vanishing of the coefficient of $\h^{3/2}$:
	\beq\label{F-2-nd-order}
	\big(\Op_1(\mathfrak{l}) - \lambda_1\big) v_2 = - \Op_1(\hat{\mathcal{k}}^2) v_1 - \Op_1(\hat{\mathcal{k}}^3) v_0 + \lambda_3  v_0+\lambda_2v_1.
	\eeq
	
	It is clear that for the formal expansion we can solve term by term determining at each step $\lambda_{k+1}$ and $u_{k}$ such that
	\beq\label{F-iter-kIh}
	\Op_1(\mathfrak{l})v_k + \sum_{k+1\geq \ell > 1} \Op_1(\hat{\mathcal{k}}^\ell)v_{k+1-\ell}\,=\,\underset{1\leq\ell\leq k+1}{\sum}\lambda_\ell v_{k+1-\ell} ,\quad u^{(k)}=\mathfrak{d}_{\h^{-1/2}}\left(\underset{0\leq\ell\leq k}{\sum}\h^{\ell/2}v_{\ell}\right).
	\eeq
	
	It remains to show how we control the remainder  estimates. We treat for simplification the case $k=1$. Let us show that
	$$
	\Big(\big(\Op_1(\mathcal{k}_{\h}) -\h^{1/2} \lambda_1\big)u^{(0)}\Big) (t ,\h)) = \mathscr{O} (\h),\quad\mbox{in }L^2(\mathbb{R}).
	$$
	This will clearly follow if we prove that
	$$
	\big\|\big(\Op_1(\widetilde{\mathcal{k}}_{\h})-\h^{1/2}\Op_1(\mathcal{l})\big) v_0\big\| =\mathscr{O} (\h).
	$$
	This last estimate is obtained by looking at the symbol $\widetilde{\mathcal{k}}_{\h} - \h^{1/2}\mathcal{l}$. Using \eqref{estsurr},     
	this implies, using a global  Beals like pseudo-differential  calculus
	and the continuity properties of the associated operators, 
	$$
	\big\|\big(\Op_1(\widetilde{\mathcal{k}}_{\h})-\h^{1/2}\Op_1(\mathcal{l})\big) v_0\big\| \leq C \h  \sum_{\ell + m \leq 2} || t^\ell \partial_t^m v_0|| \,.
	$$
	The last quantity is finite since $v_0 \in \mathscr S(\mathbb R)$.
	
	Suppose we have proven \eqref{eq:8.55} for $k\leq m$ and consider $k=m+1$ and the corresponding expression:
	$$
	\left(\Op_1(\mathcal{k}_{\h})-\Big(\underset{1\leq\ell\leq m}{\sum}\h^{\ell/2}\lambda_\ell\Big)\right)\Big(\underset{0\leq\ell\leq k-1}{\sum}\h^{\ell/2}(\mathfrak{d}_{\h^{1/2}}v_\ell)\Big).
	$$
	Using \eqref{F-iter-kIh},
	the same arguments as before, based on \eqref{estsurr}, imply \eqref{eq:8.55} for $k=m+1$.
\end{proof}

\begin{corollary} 
There exists an infinite sequence $\{\lambda_\ell\}_{\ell\geq1}\subset\mathbb{R}$ such that for any $k\geq 1$, there exists $\h_k>0$ and $C_k>0$ such that
\begin{equation*}
\dist\left(\sigma\big(\Op_1(\mathcal{k}_{\h})\big), \underset{1\leq\ell\leq k}{\sum}\h^{\ell/2}\lambda_\ell\right) \leq C_k \h^{(k+1)/2}\,,\, \forall \h\in (0,\h_k]\,.
\end{equation*}
\end{corollary}

\subsubsection{The quasiresolvents in the one crossing point situation}

We now consider the question of the resolvent of the linearised operator. We consider a compact $K \subset \mathbb C$ and a point $\lambda \in  K$ such that $\dist(\lambda, \sigma (\mathfrak{L})) \geq \alpha_0>0
$.

\begin{proposition}\label{P-sp-kIdelta}
Given $\alpha_0>0$, there exists $\h_0>0$ such that given any compact $K\subset\mathbb{C}$ with $\dist\big(K,\sigma(\mathfrak{L})\big)=\alpha_0>0$, the set $\h^{1/2}K$ is contained in the resolvent set of $\Op_1\big(\mathcal{k}_{\h}\big)$ for any $\h\in(0,\h_0]$.
\end{proposition}
\begin{proof}
Under the hypothesis of the Proposition, we can construct an inverse of $(\mathfrak{L} -\lambda)$ satisfying the properties described in point 6 of Theorem \ref{T-HR} (as explained in \cite{He-84}). We consider the operator
\beq\label{F-03}
R_1(\lambda,\h) := \h^{-1/2} (\mathfrak{L} -\lambda)^{-1},
\eeq
and look at $ R_1(\lambda,\h)\circ \Op_1(\widetilde{\mathcal{k}}_{\h})$.  From point (6) in Theorem \ref{T-HR} we know that $  R_1(\lambda,\h)=\Op_1\big(\h^{-1/2}\mathfrak{r}(\lambda)\big)$ with $\mathfrak{r}(\lambda)\in C^\infty (\mathbb R^2;\mathbb C)$ such that for any non negative integers  $\ell, m$ there exists $C_{\ell m}$ such that
$$ |\partial_t^\ell\partial_\tau^m\mathfrak{r}(t,\tau)| \leq C_{\ell m} (1+ t^2 +\tau^2)^{-(1+\ell + m)/2}  \mbox{ for } (t,\tau)\in \mathbb R^2.
$$
Thus $R_1(\lambda,\h)\circ \Op_1(\widetilde{\mathcal{k}}_{\h})=\Op_1\big(\h^{-1/2}\mathfrak{r}(\lambda)\sharp\widetilde{\mathcal{k}}_{\h}\big)$ and using the estimates above and the pseudo-differential calculus, we get
\begin{align*}
\big(\h^{-1/2}\mathfrak{r}(\lambda)\sharp(\widetilde{\mathcal{k}}_{\h}-\h^{1/2}\lambda)\big)(t,\tau)&=\big(\mathfrak{r}(\lambda)\sharp(\mathcal{l}-\lambda)\big)(t,\tau)+\h^{-1/2}\big(\mathfrak{r}(\lambda)\sharp r_1(\cdot,\h)\big)(t,\tau)\\
&=1+\h^{-1/2}\big(\mathfrak{r}(\lambda)\sharp r_1(\cdot,\h)\big)(t,\tau)
\end{align*}
\begin{align*}
\h^{-1/2}\big(\mathfrak{r}(\lambda)\sharp r_1(\cdot,\h)\big)(t,\tau)&=\pi^{-2}\h^{-1/2}\int_{\mathbb{R}^2}dT'\int_{\mathbb{R}^2}dT"\,e^{-2i\sigma(T-T',T-T")}\mathfrak{r}(\lambda;T')r_1(T",\h) \\
&=\pi^{-2}\h^{1/2}\mathfrak{s}_1(T,\h))
\end{align*}
where $\mathfrak{s}_{1}\in C^\infty(\mathbb{R}^2\times\mathbb{R}_+;\mathcal{M}_{2\times2}(\mathbb{C})\big)$ satisfies
\begin{equation}\label{F-est-s1}
\big|\big(\partial_t^\ell \partial_\tau^m\mathfrak{s}_{1}\big)(t,\tau,h,\lambda) \big| \leq C_{\ell,m} (|t|^2+|\tau|^2)^{[(1-m-\ell)/2]_+}\,.
\end{equation}
Thus we can write
\begin{equation}\label{8.?}
R_1(\lambda,\h) \circ (\Op_1(\widetilde{\mathcal{k}}_{\h})-\h^{1/2} \lambda)  = \bb1 + \h^{1/2}\Op_1(\mathfrak{s}_{1}) \,,
\end{equation} 
  We notice now that the divergence at infinity of the remainder terms in Definition \ref{D-Tdev-tildek} goes up with $k\in\mathbb{N}$. In order to control this growing of the first residue term $r_1$ we shall take advantage of the fact that the symbol $\widetilde{k}_\h$ is uniformly invertible outside any neighbourhood of $T=(0,0)$ and use a cut-off in order to control the local behaviour close to $T=(0,0)$.
We consider $\chi_1$ with compact support and equal to $1$ in a neighbourhood of $0$ and for any $r>0$ denote by $\chi_{1,r}(t,\tau):=\chi_1(t,r\tau)$ and  $\widetilde{\chi}_{1,r}(t,\tau):=\chi_1(r^{1/2}t,r^{1/2}\tau)$. With a parameter $\varrho>0$ to be fixed later, we compose equation \eqref{8.?} on the left with $ \Op_1(\widetilde{\chi}_{1,\varrho^{-2}\h})$.
Using the pseudo differential calculus, we obtain for the right member:
\begin{align}
&\big(\h^{1/2}\widetilde{\chi}_{1,\varrho^{-2}\h}\sharp\mathfrak{s}_{1}\big)(T,\h)=\pi^{-2}\h^{1/2}\int_{\mathbb{R}^2}dT'\int_{\mathbb{R}^2}dT"\,e^{-2i\sigma(T',T")}\chi_1\big(\varrho^{-1}\h^{1/2}(T-T')\big)\,\mathfrak{s}_1(T-T",\h) \nonumber \\ \label{F-rest-chi1s1}
&\hspace*{14pt}=\h^{1/2}\chi_1\big(\varrho^{-1}\h^{1/2}T\big)\,\mathfrak{s}_1(T;\h)\,+\\ \nonumber
&\hspace*{28pt}+\h^{1/2}(2i\pi)^{-2}\frac{\h^{1/2}}{\varrho}\int_{\mathbb{R}^2}dT'\int_{\mathbb{R}^2}dT"\int_0^1ds\,e^{-2i<T',\mathfrak{J}T">}\big(\mathfrak{J}^t\nabla_T\chi_1\big)\big(\varrho^{-1}\h^{1/2}(T-T')\big)\cdot\\ \nonumber &\hspace*{8cm}\cdot \big(\nabla_T\mathfrak{s}_1\big)(T-sT",\h),
\end{align}
(with the notation in \eqref{DF-J}). Using estimation \eqref{F-est-s1} for the first term above and the estimation in the Calderon-Vaillancourt theorem we notice that there exists some $p\in\mathbb{N}$ and a constant $C>0$ such that
\begin{align*}
\big\|\h^{1/2}\Op_1(\widetilde{\chi}_{1,\varrho^{-2}\h})
\Op_1(\mathfrak{s}_{1})\big\|_{\mathbb{B}(L^2(\mathbb{R}^2;\mathbb{C}^2))}&\leq
\h^{1/2}\underset{m+\ell\leq p}{\sum}\big\|\partial_t^m\partial_\tau^\ell\big(\widetilde{\chi}_{1,\varrho^{-2}\h}\sharp\mathfrak{s}_{1}\big)\big\|_\infty\\
&\leq C\h^{1/2}\frac{\varrho}{h^{1/2}}\left[\underset{0\leq q\leq p}{\sum}\left(\frac{h^{1/2}}{\varrho}\right)^q\right].
\end{align*}
For the second term in \eqref{F-rest-chi1s1} usual oscillatory integrals techniques show that it is bounded as operator on $L^2(\mathbb{R}^2;\mathbb{C}^2)$ with a norm of order $h^{1/2}\left(\frac{h^{1/2}}{\varrho}\right)^3$.
We shall choose $\varrho\in(0,1]$ small enough but strictly positive and then we shall choose $\h_\varrho\in(0,\varrho^2]$ such that for $\h\in(0,\h_\varrho]$ we have that the parameter $\h^{1/2}/\varrho\in(0,1]$. Then all the semi-norms of a symbol $\widetilde{F}_{\varrho^{-2}\h}(T):=F(\varrho^{-1}\h^{1/2}T)$ are uniformly bounded by the corresponding semi-norms of the symbol $F$.
In particular, for any $\eta_0 >0$,  there exists $\varrho_0>0$  such that for $\varrho\leq \varrho_0$, there exists $\h_\varrho>0$ such that for $\h\in (0,\h_\varrho]$
we have 
\beq\label{F-eta0}
||\h^{1/2}\Op_1(\widetilde{\chi}_{1,\varrho^{-2}\h})\Op_1(\mathfrak{s}_{1})||_{\mathbb B (L^2(\mathbb R,\mathbb C^2))} \leq \eta_0\,.
\eeq
We now fix $\varrho>0$ with the above property.  

The second claim is that there exists a pseudo-differential operator $R_2(\lambda):=\Op_1\big(\widetilde{r}_{2,\lambda,\h}\big)$, with $\widetilde{r}_{2,\lambda,\h}(t,\tau):=r_{2,\lambda}(\h^{1/2}t,\h^{1/2}\tau)$ for some $r_{2,\lambda} \in S^0(\mathbb R^2;\mathcal M_{2\times 2}(\mathbb C))$ uniformly with respect to $\lambda\in K$  for a compact $K\subset\mathbb{C}$ at distance greater then $\alpha_0>0$ from $\sigma(\mathfrak{L})\subset\mathbb{R}$, and such that
\begin{equation}\label{F-R2}
R_2(\lambda) \circ \big(\Op_1(\widetilde{\mathcal{k}}_{\h})-h^{1/2}\lambda) = \bb1-\Op_1\big(\widetilde{\chi}_{1,\varrho^{-2}\h}\big) + h^{1/2}  \Op_1\big(\mathfrak s_2\big) \,,
\end{equation}
where $\mathfrak s_2 \in S^0$.
This second property is obtained by simply assuming that $ \mathcal{k}^\circ_{I,\delta}(t,\tau)$ is uniformly  invertible outside any neighbourhood of $(0,0)\in\mathbb{R}^2$ and defining
$$
r_{2,\lambda}(\varrho^{-1}\h^{1/2}t,\varrho^{-1}\h^{1/2}\tau):= \big(1-\chi_1(\varrho^{-1}\h^{1/2}t,\varrho^{-1}\h^{1/2}\tau)\big) \big(\mathcal{k}_{I,\delta}(\h^{1/2}t,\h^{1/2}\tau)\big)^{-1}\,.
$$
By usual Weyl calculus with the parameter $\varrho^{-1}\h^{1/2}\in(0,1]$ we notice that 
 $$
\big(1-\widetilde{\chi}_{1,\varrho^{-2}\h}\big)\widetilde{\mathcal{k}}_{\h}^{-1}=\big(1-\widetilde{\chi}_{1,\varrho^{-2}\h}\big)\sharp\widetilde{\mathcal{k}}_{\h}^{-1}+\h^{1/2}\mathfrak{x}_1(\varrho,\h)
$$
where $\mathfrak{x}_1(\varrho,\h)$ is a symbol of class $S^0(\mathbb{R}^2)$ uniformly for $\h\in(0,1]$. Similarly, we obtain that 
$$ 
\widetilde{r}_{2,\lambda,\h}\sharp\big(\widetilde{\mathcal{k}}_{\h}-\h^{1/2}\lambda\big)=1+\h^{1/2}\mathfrak{x}_2(\lambda,\varrho,\h)
$$
with $\mathfrak{x}_2(\lambda,\varrho,\h)$ a symbol of class $S^0(\mathbb{R}^2)$ uniformly for $\h\in(0,1]$ and $\lambda\in K$  for a compact $K\subset\mathbb{C}$ at distance greater then $\alpha_0>0$ from $\sigma(\mathfrak{L})\subset\mathbb{R}$. Let us emphasize that the dependence of $\mathfrak{x}_1$ and $\mathfrak{x}_2$ on $\varrho\in(0,1]$ may be singular when $\varrho$ goes to 0 but we shall fix some $\varrho\in(0,1]$ small enough, when fixing $\eta_0$ in \eqref{F-eta0} and keep it fixed for all our analysis. This finishes the proof of \eqref{F-R2}.

Finally  the operator
\beq\label{F-02}
R_{app}(\lambda):=\Op_1\big(\widetilde{\chi}_{1,\varrho^{-2}\h}\big)R_1(\lambda) + R_2(\lambda)
\eeq 
has the property that
\beq\label{F-01}
R_{app}(\lambda) \circ  \big(\Op_1(\widetilde{\mathcal{k}}_{\h})-h^\frac 12 \lambda)  = I + \Op_1(\mathfrak r_{\h,\lambda})\,,
\eeq
with $\mathfrak r_{\h,\lambda}$ a symbol of class $S^0$ satisfying for $\h$ small enough the estimation: 
\beq\label{F-04}
||\Op_1(\mathfrak r_{\h,\lambda})||_{\mathbb B (L^2(\mathbb R^2,\mathbb C^2))} < 2 \eta_0 \,.
\eeq
Hence, as soon as $2\eta_0 < 1$, we have for $\h$ small enough
$$\h^{1/2} \lambda \not\in \sigma \big(\Op_1(\widetilde{\mathcal{k}}_{\h})\big)\,.$$
\end{proof}

\begin{corollary}
For any $\alpha_0 >0$, there exists $\h_0$ such that for $h\in (0,h_0)$
$$ \Big(\sigma\big(h^{-1/2} \Op_1(\mathcal{k}_{\h})\big)\cap K\Big) \subset  \Big(\sigma(\mathfrak{L}) \cap K\Big)\,+\, (-\alpha_0, \alpha_0) 
$$
\end{corollary}

Let us choose $K$ as a disk whose boundary avoids $\sigma(\mathfrak{L})$.
For $\alpha_0 >0$ small enough the right hand side of the equality in the Corollary above consists of disjoint intervals each interval centered on one (and only one) eigenvalue of $\mathfrak{L}$ (all having multiplicity 1). Due to \eqref{F-sym-repr} we know that $\Op_1(\mathcal{k}_{I,\delta,\h})$ and $\Op_1(\widetilde{\mathcal{k}}_{I,\delta,\h})$ are unitary equivalent and thus have the same spectrum. For simplicity, in the statement and proof of the following Proposition, we shall use the notation
$$
\mathfrak{L}_1:=h^{-1/2} \Op_1(\widetilde{\mathcal{k}}_{\h}),
$$
for $\h\in(0,\h_0]$ with $\h_0$ as in the statement of Proposition \ref{P-sp-kIdelta}.
Given any $\lambda_1\in\sigma(\mathfrak{L})$ we shall compare the projections $ \Pi_1$ and $ \Pi$ relative to the spectral interval $(\lambda_1-\alpha_0,\lambda_1+\alpha_0)$
of $\mathfrak{L}_1$ and $\mathfrak{L}$. 
\begin{proposition}
Suppose fixed $\lambda_1\in\sigma(\mathfrak{L})$ and some $\alpha_0 >0$ small enough. Let us denote by $\Pi_1$ and resp. by $\Pi$ the spectral orthogonal projections of $\mathfrak{L}_1$ and resp. of $\mathfrak{L}$ on the spectral interval $L:=(\lambda_1-\alpha_0,\lambda_1+\alpha_0)$. Then, for any $\eta_1>0$ there exists $\h_1\in(0,\h_0]$ such that for $\h\in(0,\h_1]$ we have that $\big\|\Pi_1-\Pi\big\|_{\mathbb{B}(L^2(\mathbb{R};\mathbb{C}^2))}\,<\,\eta_1$.
\end{proposition}
\begin{proof}
We know that these projections are obtained 
by integrating the resolvent around a small circle in $\mathbb C$ centered at $\lambda_1$. For these resolvents  we have found an explicit formula.
Let us look at
\begin{equation*}
\frac{1}{2\pi}  \int_{|\lambda|=\alpha_0} (\mathfrak{L}_1 -\lambda)^{-1} d\lambda\,.
\end{equation*}
 We shall use the notations and elements defined in the proof of Proposition \ref{P-sp-kIdelta}. Using \eqref{F-01}, \eqref{F-02} and \eqref{F-03}, we can write the equation
$$
(\mathfrak{L}_1-\lambda)^{-1} = \Op_1\big(\widetilde{\chi}_{1,\varrho^{-2}\h}\big)(\mathfrak{L}-\lambda)^{-1} +
h^{-\frac 12}  R_2(\lambda) - \Op_1\big(\mathfrak r_{h,\lambda}\big) (\mathfrak{L}_1-\lambda)^{-1}
$$
By integration, observing that $ R_2(\lambda)$ depends holomorphically of $\lambda\in K$ and using the estimation \eqref{F-04} we get 
$$
\Pi_1 = \Op_1\big(\widetilde{\chi}_{1,\varrho^{-2}\h}\big)\Pi  + \mathcal O (\eta_0)\,.
$$
We can choose $\eta_0$ and $\varrho$ such that for $\h$ small enough
$$
||  \Pi_1 - \Op_1\big(\widetilde{\chi}_{1,\varrho^{-2}\h}\big)\Pi  ||_{\mathbb B (L^2(\mathbb R,\mathbb C^2))} \leq \frac{1}{2}\,.
$$
We end the proof by showing that:
$$
|| ( 1- \Op_1\big(\widetilde{\chi}_{1,\varrho^{-2}\h}\big)\Pi ||_{\mathbb B (L^2(\mathbb R,\mathbb C^2))} = \mathcal O (h^\frac 12)\,.
$$
This results of the Calderon-Vaillancourt theorem noting that $\Pi$ has a symbol in $\mathscr S(\mathbb R^2,\mathbb C^2)$ and that the support of the symbol $(t,\tau) \mapsto ( 1-  \chi (\varrho^{-1}\h^{1/2}t,\varrho^{-1}\h^{1/2}\tau)))$ is contained in 
$t^2 + \tau^2 \geq C\h^{-1/2}$. Actually this norm is $\mathcal O (\h^\infty)$. Finally, for any $\eta_1>0$,  we have found  $\h_1>0$ such that for $\h\in (0,\h_1]$ we have the estimation
$$ ||  \Pi-\Pi_1 ||_{\mathbb B (L^2(\mathbb R,\mathbb C^2))} < \eta_1\,.
$$
\end{proof}
As a consequence, choosing $\eta_1 < 1$, we can see that the rank of $\Pi_1$ is exactly one. Hence the interval $(\lambda_1-\alpha_0,\lambda_1 + \alpha_0)$ contains for $\h$ small enough a unique eigenvalue of $\h^{-1/2}\Op_1( \mathcal{k}_{I,\delta,\h})$.
Moreover using the results in Subsection \ref{SS-qmodes00} we have a complete expansion in powers of $h^\frac 12$ of this eigenvalue.

\subsection{The periodic semi-classical symbol.}

We remark that we are in the situation treated in \S2 ("cas de plusieurs puits microlocaux") in \cite{HS1} with a family $\big\{P_{\alpha^*}:=\tau_{\alpha^*}\mathcal{k}^\circ_{I,\delta}\big\}_{\alpha^*\in\Gamma_*}$ of symbols on $\mathbb{R}^2$, uniformly elliptic outside any neighborhood of the form $B_{d_1}(\theta_0+\alpha^*)$ and a symbol $P=\mathcal{k}_{I,\delta}$ elliptic outside the union of the above countable family of balls that coincides with $P_{\alpha^*}$ on $B_{d_1}(\theta_0+\alpha^*)$. Moreover, in our case all the operators $\Op_1(P_\alpha^*)$ are unitary equivalent. We have to study the spectrum of each of the operators $\Op_1(P_\alpha^*)$ and then (2.25) and results in \S2 of \cite{HS1} tell us that the spectrum of $\Op_1(\widetilde{\mathcal{k}}_{I,\delta,\h})$ in a neighborhood of 0, is contained in a $\mathcal O(h^\infty)$-neighborhood of the spectrum of the 'one crossing point' Hamiltonian $\Op_1(\widetilde{\mathcal{k}}^\circ_{I,\delta,\h})$. More precisely, we obtain the following statement.

\begin{proposition}\label{P-sp-location}~
\begin{enumerate}
\item	For any $L >0$ as in the statement of Theorem \ref{T-Main-2}, with $\text{\tt d}_H$ defined in \eqref{apr-1}, there exist $\h_0$ and $C_0$ such that:
	\begin{equation*}
	\text{\tt d}_H\Big(\sigma\big(\Op_1(\mathcal{k}_{I,\delta,\h})\big)\cap (-L\h^{1/2}, L\h^{1/2}) , \h^{1/2} \sigma (\mathfrak{L})\cap (-L\h^{1/2}, L\h^{1/2})\Big)\leq C_0 \h \,,\, \forall \h \in (0,\h_0]\,.
	\end{equation*}
\item	Moreover, for each eigenvalue $\lambda_n$ of $\sigma(\mathfrak{L})$, there exists an asymptotic series $$\lambda(\h)\sim \h^\frac 12\left( 
	\lambda_n +\sum_{j>0}\lambda_{n,j} \h^{\frac j2}\right)$$ such that if we denote by
	$$
	\lambda_{n,N}(\h):=\lambda_n +\sum_{1\leq l\leq N}\lambda_{n,j} \h^{\frac j2},\quad\forall N\in\mathbb{N}\setminus\{0\}
	$$
	then $\forall N\in\mathbb{N}\setminus\{0\}$ there exists $C_N>0$ and $\h_N\in(0,\h_{N-1}]$ such that for any $\h\in(0,\h_N]$ we have that
	\begin{align*}
	\sigma \big(\Op_1(\mathcal{k}_{I,\delta,\h})\big) &\cap  \Big (\h^{1/2} (\lambda_n - C_0 \h^{1/2}) , \h^{1/2} (\lambda_n + C_0 \h^{1/2})\Big )
	= \\
	&= \sigma\big(\Op_1(\mathcal{k}_{I,\delta,\h})\big) \cap  \Big (\lambda_{n,N}(\h) - C_N \h^{N/2}\, , \, \lambda_{n,N}(\h) + C_N \h^{N/2}\Big )\,.
	\end{align*}
	\end{enumerate}
\end{proposition}

\section{End of the proof of Theorem \ref{T-Main-2}.}\label{S-final}

As already stated in the Remark \ref{P-CP1}, due to Theorem 3.1 in \cite{CP-1}, Theorem \ref{T-Main-2} follows once we have proved the case $\kappa=0$ of the constant field $\epsilon B^\circ$. The main idea of the proof is to compare the spectrum of $H^\epsilon_\Gamma$ in the neighborhood of $0\in\mathbb{R}$ with the spectrum of the operator $\h^{1/2}\mathcal{L}$ in the same neighborhood

First let us fix some neighborhood $B_R(\theta_0)\subset\Sigma_I$ of $\theta_0$ needed for the construction in \eqref{DF-mItheta} and \eqref{DF-Phi-pm} and let us redefine our spectral window as  $I_R:=\big(-\Lambda_R\,,\,\Lambda_R\big)\cap \mathring{I}$ with $\Lambda_R>0$ from Proposition \ref{P-restr-sp-window}.

We have to take into account that there exists some $\epsilon_F>0$, as required by Definition \ref{D-F-epsilon} such that the projection $P^\epsilon_I$ exists for any $\epsilon\in[0,\epsilon_F]$. Then let us fix some compact interval $I_\circ\subset I_R$ containing $0$ in its interior, denote by $\epsilon_1>0$ the upper bound for the magnetic field intensity $\epsilon\geq0$ associated to $I_\circ\subset \mathring{I}$ by Remark \ref{R-control-2}.

Let us fix some $N\in\mathbb{N}\setminus\{0\}$ and having in view Theorem \ref{T-HR}, let us consider the first $N+1$ strictly positive eigenvalues $\{\lambda_1,\ldots\lambda_{N+1}\}$ of $\mathcal{L}$. They have multiplicity 1 and we may define $\mu_N:=\lambda_N+(\lambda_{N+1}-\lambda_N)/2>\lambda_N>0$. Let us choose some $\widetilde{\epsilon}_0\leq\min\big\{\epsilon_F,\epsilon_1\big\}>0$ such that 
$$
\big(-\mu_N\sqrt{\widetilde{\epsilon}_0B^\circ}\,,\,\mu_N\sqrt{\widetilde{\epsilon}_0B^\circ}\big)\subset\mathring{I}_\circ.
$$
In the following, for any $a>0$ we shall denote by 
$$
J^\epsilon_N:=\big(-\mu_N\sqrt{\epsilon B^\circ}\,,\,\mu_N\sqrt{\epsilon B^\circ}\big)\subset\mathbb{R}.
$$
Using point (1) of Proposition \ref{P-sp-location} with $L:=\mu_N$ in conjunction with Proposition \ref{C-equi-scl}, we conclude that there exists some $\epsilon_0:=(B^\circ)^{-1}\h_0\in(0,\widetilde{\epsilon}_0)$ and some $C_0>0$ such that 
\begin{equation}\label{C-1}
	\text{\tt d}_H\Big(\sigma\big(\Op^\epsilon(\mathcal{k}_{I,\delta})\big)\cap J^{\epsilon}_{N}, \sqrt{\epsilon B^\circ} \sigma (\mathfrak{L})\cap J^{\epsilon}_{N}\Big)\leq C_0B^\circ \epsilon \,,\, \forall \epsilon \in (0,\epsilon_0]\,.
	\end{equation}

Considering now the Propositions \ref{P-magn-matrix-ke} and \ref{P-kepsilon-k}, the above conclusion may be rephrased as: there exists some $C_1>0$ such that  
\begin{equation}\label{C-2}
	\text{\tt d}_H\Big(\sigma\big(\widetilde{H}^\epsilon_{I,\delta}\big)\cap J^{\epsilon}_{N}, \sqrt{\epsilon B^\circ} \sigma (\mathfrak{L})\cap J^{\epsilon}_{N}\Big)\leq C_1B^\circ \epsilon \,,\, \forall \epsilon \in (0,\epsilon_0]\,.
	\end{equation}

Now, for any $\epsilon\in(0,\epsilon_0]$ let us use Proposition \ref{P-magn-FS-est} with $I_\circ=J^\epsilon_N$ and obtain the following Corollary.

\begin{corollary}\label{C-magn-FS-est}
With $N\in\mathbb{N}\setminus\{0\}$ and $\epsilon_0>0$ fixed as above, for any $\epsilon\in[0,\epsilon_0]$ we have the following estimation:
$$\max\left \{ \sup_{e\in\sigma(H^{\epsilon}_\Gamma)\cap J^\epsilon_N}\dist\Big(e,\sigma(\widetilde{H}^{\epsilon}_{I,\delta})\Big),\; \sup_{e\in \sigma(\widetilde{H}^{\epsilon}_{I,\delta})\cap J^\epsilon_N}\dist\Big(e,\sigma(H^{\epsilon}_\Gamma)\Big)\right \}\leq C(N)\; \epsilon.
$$
\end{corollary}
Let us recall that the arguments leading to Corollary \ref{C-magn-FS-est} depend on the parameter $\delta\in(0,\delta_0)$ and the same thing is true for the constant $C_N>0$ and for $\epsilon_0>0$.
Putting together Corollary \ref{C-magn-FS-est} and the above conclusion \eqref{C-2} we obtain the Theorem \ref{T-Main-2} with $\kappa=0$.

\appendix

\section*{Appendices.}

\section{Construction of a global smooth section}\label{A-BBdl}
In order to simplify notation, we will work with $\mathbb{Z}^2$ instead of $(2\pi \mathbb{Z})^2$ periodicity in $\theta$. 
\begin{lemma}\label{L-A1}
	Let us consider a smooth family of orthogonal projections $P(\theta)$ living in some separable complex Hilbert space. Assume that the family is $\mathbb{Z}^2$-periodic in $\theta\in \R^2$ and the rank of $P(\theta)$ is $d\geq 2$. Let $B:=B_r(0)$ with $r<1/2$ be an open disk such that $\overline{B}\subset (-1/2,1/2)^2$. Let us assume that 
$ \{u_j(\theta)\}_{j=1}^d$ form a smooth and orthonormal basis of ${\rm Ran\,}P(\theta)$ when $\theta\in B$. Then  one can construct a continuous and $\mathbb{Z}^2$-periodic unit vector $\psi(\theta)\in {\rm Ran\,}P(\theta)$  such that $\psi(\theta)=u_1(\theta)$ on $B$. 
\end{lemma}

\begin{proof}
Our proof will be done for a finite rank $d$ but all estimates are uniform in $d$. Most of the ideas are borrowed from \cite{CHN, CM} and \cite{FMP}. We start off by choosing any orthonormal basis $\{v_j\}_{j=1}^d\in {\rm Ran\,}P\big ((-1/2,-1/2)\big )$. By parallel transport we can smoothly extend it to the segment $\{(-1/2,t)\in \R^2:\; -1/2\leq t\leq 1/2\}$. The transported basis $\{v_j(-1/2,t)\}_{j=1}^d$ might differ at $t=1/2$ from the one we started with at $t=-1/2$, but they are related through a unitary because they span the same projection since $P\big ((-1/2,-1/2)\big )=P\big ((-1/2,1/2)\big )$. This unitary can be \enquote{straightened} out \cite{CHN, FMP} and we obtain the same basis at the corner $(-1/2,1/2)$. Repeating this construction for the other three boundary segments of $[-1/2,1/2]^2$ we may assume that we have a continuous basis $\{v_j(\theta)\}_{j=1}^d\in {\rm Ran\,}P(\theta)$ (when $\theta$ is restricted to the boundary of $[-1/2,1/2]^2$) which is the same on the opposite sides and can be \enquote{periodized} to the boundary of any unit square with center at some point in $\mathbb{Z}^2$. The remaining problem is to \enquote{continuously propagate} this basis inside the unit square and to match it with whatever comes out of $B$. 

We extend by parallel transport $\{v_j(\theta)\}_{j=1}^d$ from the boundary inside $(-1/2,1/2)^2\setminus \overline{B}$ along rays linking points of the boundary with the origin. In this way, on the boundary of $B$ we will have two bases: one from outside (the $v_j$'s) and one from inside (the $u_j$'s). Let us use polar coordinates $\theta:=(\rho,\varphi)$ around the boundary of $B$, which corresponds to $\rho=r$. 

There exists a continuous family of $d\times d$ unitary matrices $W(\varphi)$ which is $2\pi$-periodic such that 
\begin{align}\label{may9}
v_j(r,\varphi)=\sum_{k=1}^d W_{kj}(\varphi) u_k(r,\varphi),\quad W_{kj}(\varphi)=\langle v_j(r,\varphi),u_k(r,\varphi)\rangle.
\end{align}
If the first column of this matrix would equal $[1,0,...,0]^t$ for all $\varphi$, then $v_1$ and $u_1$ would coincide on the boundary of $B$ and $v_1$ would be our continuous extension of $u_1$ to the whole square. Now assume that the first column of $W$  is the $d$ dimensional complex unit vector $c_1(\varphi)$. This vector can be seen as a map from the unit circle to the real sphere $\mathbb{S}^{2d-1}$. Since $2d-1\geq 2$, using Lemma \ref{PropB} we obtain some $c_1(\varphi;s)$ with $\varphi\in \mathbb{S}^1$ and $0\leq s\leq 1$ which is continuous in both variables and 
$$c_1(\varphi;0) =c_1(\varphi),\quad c_1(\varphi;1)=[1,0,0,...,0]^t.$$
The orthogonal projection $p(\varphi;s):=|c_1(\varphi;s)\rangle \langle c_1(\varphi;s)|$ living in $\mathbb{C}^d$ is continuous in $s$ and by a repeated Nagy unitary construction we may find a continuous intertwining unitary $U(\varphi;s)$ such that 
$$p(\varphi;s)U(\varphi;s) =U(\varphi;s)p(\varphi;0),\quad 0\leq s\leq 1.$$
This implies that the vectors $c_1(\varphi;s)$ and $U(\varphi;s)c_j(\theta)$ where $2\leq j\leq d$ form an orthonormal basis in $\mathbb{C}^d$. We put them together as the columns of a unitary $W(\varphi;s)$. This unitary is continuous in both variables and continuously interpolates between $W(\varphi)$ and the matrix $W(\varphi;1)$ whose first column equals $[1,0,...,0]^t$ for all $\varphi$. 

Let us get back to our \enquote{inner} basis  $u_j(\rho,\varphi)$, $\rho\leq r$. It can be further extended by parallel transport up to $\rho\leq r+\epsilon$ when $\epsilon$ is small enough. Let us consider the following \enquote{rotated} vector:
\begin{align}\label{may11}
\tilde{u}_1(\rho,\varphi):= \sum_{k=1}^d \big [W(\varphi)W^{-1}\big (\varphi; \frac{\rho-r}{\epsilon}\big )\big ]_{k1} u_k(\rho,\varphi),\quad r\leq \rho\leq r+\epsilon.
\end{align}

Denote by 
\begin{align}\label{may12}
A(\varphi):=W(\varphi)W^{-1} (\varphi; 1),\quad \sum_{k\geq 1}|A_{k1}|^2=1.
\end{align}

Using \eqref{may11} and \eqref{may12} we have:
\begin{align}\label{may13}
\langle v_1(r+\epsilon,\varphi),\tilde{u}_1(r+\epsilon,\varphi)\rangle =\sum_{k=1}^d A_{k1}(\varphi)\; \langle v_1(r+\epsilon,\varphi),{u}_k(r+\epsilon,\varphi)\rangle. 
\end{align}

We have
\begin{align*}
\langle v_1(r+\epsilon,\varphi),{u}_k(r+\epsilon,\varphi)\rangle &=
\langle v_1(r+\epsilon,\varphi)-v_1(r,\varphi),{u}_k(r+\epsilon,\varphi)\rangle\\
&+\langle v_1(r,\varphi),{u}_k(r+\epsilon,\varphi)-u_k(r,\varphi)\rangle+
\langle v_1(r,\varphi),{u}_k(r,\varphi)\rangle.
\end{align*}
Using \eqref{may9} we have
\begin{align*}
 \Big (\langle v_1(r+\epsilon,\varphi),{u}_k(r+\epsilon,\varphi)\rangle    &-W_{k1}(\varphi)\Big )A_{k1}(\varphi)\\
 &=\langle v_1(r+\epsilon,\varphi)-v_1(r,\varphi),{u}_k(r+\epsilon,\varphi)\rangle\; A_{k1}(\varphi)\\
 &+\langle v_1(r,\varphi),{u}_k(r+\epsilon,\varphi)-u_k(r,\varphi)\rangle\; A_{k1}(\varphi).
\end{align*}
Using the Bessel inequality for the coefficients of an orthogonal system, the fact that $\sum_{k\geq 1}|A_{k1}|^2=1$,  the continuity of $v_1$ and $u_k$, together with the Cauchy-Schwartz inequality lead to
$$\lim_{\epsilon\to 0}\sum_{k\geq 1}\left |\Big (\langle v_1(r+\epsilon,\varphi),{u}_k(r+\epsilon,\varphi)\rangle    -W_{k1}(\varphi)\Big )A_{k1}(\varphi)\right |=0.$$

Inserting this in \eqref{may13} leads to:
$$\lim_{\epsilon\to 0}\langle v_1(r+\epsilon,\varphi),\tilde{u}_1(r+\epsilon,\varphi)\rangle= W_{11}(\varphi;1)=1.$$
This implies that $\tilde{u}_1(r+\epsilon,\varphi)$ is almost equal to the vector $v_1(r+\epsilon,\varphi)$ if $\epsilon$ is small enough. They will still be almost equal if we further extend  $\tilde{u}_1(r+\epsilon,\varphi)$ by parallel transport to the region $r+\epsilon\leq \rho\leq r+2\epsilon$. On this narrow region we may take a convex combination of the type $$\frac{r+2\epsilon-\rho }{\epsilon}\tilde{u}_1(\rho,\varphi)+\frac{\rho -r-\epsilon }{\epsilon}v_1(\rho,\varphi),\quad r+\epsilon\leq \rho\leq r+2\epsilon$$ which cannot vanish and can be normalized. In this way we switch $\tilde{u}_1$ with $v_1$. 

To conclude, our continuous and periodic section $\psi(\theta)$ coincides with the vector which equals $u_1(\rho,\varphi)$ inside $B$, then equals $\tilde{u}_1(\rho,\varphi)$ on $r\leq \rho\leq r+\epsilon$, then equals a  convex combination between $v_1(\rho,\varphi)$ and the extended (by parallel transport) $\tilde{u}_1(r+\epsilon,\varphi)$   to the interval $r+\epsilon\leq \rho \leq r+2\epsilon$, and finally equals $v_1(\rho,\varphi)$ on the rest of $[-1/2,1/2]^2$. 
\end{proof}

Now we show how to make the section globally smooth. 

 \begin{lemma}\label{L-A2} 
 	Let $P(\theta)$ be a smooth and $\mathbb{Z}^2$-periodic  family of orthogonal projections living in a complex Hilbert space. Let $u(\theta)\in {\rm Ran\,}P(\theta)$ be a continuous and periodic section of norm one which is smooth when restricted to an open set $B$ such that $\overline{B}\subset (-1/2,1/2)^2$. Then given any compact set $K\subset B$ there exists a globally smooth and periodic section $v(\theta)$ which equals $u$ on $K$. 
 \end{lemma}
 
 \begin{proof}Assume that $0\leq g\leq 1$ is smooth and compactly supported with $\int_{\R^2} g(t)dt=1$. Let $\epsilon>0$ and define $g_\epsilon(t)=\epsilon^{-2} g(t/\epsilon)$. Consider 
$\tilde{u}_\epsilon(\theta)=\int_{\R^2}g_\epsilon(\theta- t)u(t)dt $. Since $u$ is uniformly continuous there exists $\epsilon_0>0$ such that for all $\epsilon<\epsilon_0$ we have $\Vert P(\theta)\tilde{u}_\epsilon(\theta)\Vert\geq 1/2$ and the map $$u_\epsilon (\theta):=\frac{P(\theta)\tilde{u}_\epsilon(\theta)}{\Vert P(\theta)\tilde{u}_\epsilon(\theta)\Vert}\in {\rm Ran\,}P(\theta) $$
is smooth and $\mathbb{Z}^2$-periodic. Moreover, 
$$\lim_{\epsilon\to 0}\sup_{\theta\in\R^2 }\Vert u(\theta)-u_\epsilon(\theta)\Vert =0.$$
Now assume that $u(\theta)$ is smooth in a certain open set $B\subset (-1/2,1/2)^2$. Let $K$ be any compact subset of $B$ and let $0\leq f\leq 1$ be equal to $1$ on $K$ and zero outside $B$. If $\epsilon$ is small enough then $$\tilde{v}_\epsilon(\theta):=f(\theta) u(\theta)+(1-f(\theta))u_\epsilon(\theta)\in {\rm Ran\,}P(\theta)$$
is smooth and never zero, hence $v(\theta):=\frac{\tilde{v}_\epsilon(\theta)}{\Vert \tilde{v}_\epsilon(\theta)\Vert}$ is a smooth and periodic map which coincides with $u(\theta)$ on $K$.
\end{proof}

The following lemma is a textbook result in Algebraic Geometry (see for example Proposition 1.14 in \cite{Hat}) but for completeness we give  a short and intuitive proof based on Sard's Lemma. 
\begin{lemma}\label{PropB}
	Let $2\leq d\leq \infty$. Then any continuous loop from $\mathbb{S}^1$ to $\mathbb{S}^{d}$ can be continuously contracted to a constant map whose value is the "north pole". 
\end{lemma}
\begin{proof}
Assume first that  $2\leq d<\infty$ and let $u:\mathbb{S}^1\rightarrow\mathbb{S}^d$ be continuous. We can see $u$ as defined on $\R$ and being $\mathbb{Z}$-periodic. Assume that $0\leq g\leq 1$ is compactly supported with $\int_{\R} g(t)dt=1$. Let $\epsilon>0$ and define $g_\epsilon(t)=\epsilon^{-1} g(t/\epsilon)$. Let 
$\tilde{u}_\epsilon(\varphi)=\int_{\R}g_\epsilon(\varphi- t)u(t)dt $. Due to the uniform continuity of $u$, there exists $\epsilon_0>0$ such that for all $\epsilon<\epsilon_0$ we have $\vert\tilde{u}_\epsilon(\varphi)\vert_{\R^d}\geq 1/2$ and the map $$u_\epsilon (\varphi):=\frac{\tilde{u}_\epsilon(\varphi)}{\vert \tilde{u}_\epsilon(\varphi)\vert _{\R^d}}\in\mathbb{S}^d  $$
is smooth and $\mathbb{Z}$-periodic. Moreover, 
$$\lim_{\epsilon\to 0}\sup_{\varphi\in\R }\vert u(\varphi)-u_\epsilon(\varphi)\vert_{\R^d} =0.$$
Thus if $\epsilon$ is small enough, the formula 
$$\frac{s\; u(\varphi)+(1-s) \; u_\epsilon(\varphi)}{\vert s\;  u(\varphi)+(1-s)\; u_\epsilon(\varphi) \vert_{\R^d}},\quad 0\leq s\leq 1$$
is a homotopy which links $u$ and $u_\epsilon$. Now since the range of $u_\epsilon$ defines a one dimensional smooth curve, it cannot cover $\mathbb{S}^d$ due to Sard's Lemma and there must exist some $N\in \mathbb{S}^d$ such that $-N$ is not in the range of $u_\epsilon$. Thus we may contract $u_\epsilon$ to the constant map $v(\varphi)=N$ by
$$F(\varphi;s)=\frac{(1-s)u_\epsilon(\varphi)+sN}{\vert (1-s)u_\epsilon(\varphi)+sN\vert_{\R^d}},\quad 0\leq s\leq 1.$$
After a third homotopy we may connect $N$ to the "north pole" and we are done. 

The case $d=\infty$ corresponds to a situation in which $u(\varphi)$ takes values on the unit sphere of an infinitely dimensional separable Hilbert space. We need to first reduce the problem to a finite $d$. This can be done using the uniform continuity of $u(\varphi)$ and choosing a high dimensional subspace where \enquote{most of $u$ belongs}. In other words, we may project $u$ on a large but finite dimensional unit sphere, the difference between $u$ and its projection being smaller than some $\epsilon$. Then the previous argument applies. 
\end{proof}

\section{A Feshbach-Schur type argument.}\label{A-FS-arg}
Once that we have defined a \enquote{quasi-band} as a subspace associated to the spectral window we are interested in, a second step consists in defining a \enquote{quasi-band} Hamiltonian by projecting the initial Hamiltonian on the quasi-band space. In \cite{CHP-2} we have elaborated a procedure to compare some spectral regions of these two Hamiltonians in situations when the quasi-band is \enquote{close} to a spectral projection of the initial Hamiltonian only in some small interval. Our procedure was based on a modification of the usual Schur-Feschbach method and was suitable for the study of spectral regions near the bottom of the spectrum. In this appendix we adjust the arguments given in Section 4 of \cite{CHP-2} in order to apply them when we replace  Hypothesis~4.1 in that paper with the following one which is suitable for our new context.
\begin{definition}\label{H-band-red}
	Given a Hilbert space $\mathcal{H}$ and a triple $(H,\Pi,\mathring{I}$) where $H$ is a  self-adjoint operator, $\Pi$ is an orthogonal projection and $\mathring{I}\subset\mathbb{R}$ is an open interval, 
	 we call it admissible if  $H\Pi$ (and thus also $\Pi H$) are bounded,  and $\mathring{I}$ is contained in the resolvent set of $\Pi^\bot H\Pi^\bot$ considered as a self-adjoint operator in $\Pi^\bot\mathcal{H}$.
\end{definition}

This hypothesis implies that $\Pi^\bot(H-e\bb1)\Pi^\bot$ is invertible in $\Pi^\bot\mathcal{H}$ for any $e\in \mathring{I}$ and we shall denote by $R_\bot(e)$ its inverse 
(either considered as bounded operator in $\Pi^\bot\mathcal{H}$ or as bounded operator in $\mathcal{H}$ extended by $0$ on $\Pi\mathcal{H}$). We define:
$$
Y\ :=\ \Pi\,+\,\Pi H\Pi^\bot R_\bot(0)^2\Pi^\bot H\Pi,
$$
and notice that $Y\geq\Pi$ so that it is invertible on $\Pi\mathcal{H}$. Thus we may define the following \textit{\enquote{dressed projected Hamiltonian}} acting in $\Pi\mathcal{H}$:
\beq\label{F-Htilde}
\widetilde{H}\ :=\ Y^{-1/2}\Big (\Pi H\Pi\,-\,\Pi H\Pi^\bot R_\bot(0)\Pi^\bot H\Pi\Big )Y^{-1/2}\ \in\ \mathbb{B}\big(\Pi\mathcal{H}\big).
\eeq
\begin{remark}
	Applying the Feshbach-Schur reduction (see \cite{GS}) we conclude that for any $e\in \mathring{I}$ the operator $H-e\bb1$ is invertible in $\mathcal{H}$ if and only if the operator 
	\beq\label{F-Se}
	S(e)\ :=\ \Pi(H-e\bb1)\Pi\,-\,\Pi H\Pi^\bot R_\bot(e)\Pi^\bot H\Pi
	\eeq
	is invertible in $\Pi\mathcal{H}$ and in this case we have that
	\beq\label{F-inv-Se}
	\Pi(H-e\bb1)^{-1}\Pi\ =\ S(e)^{-1}.
	\eeq
\end{remark}
In the usual Feshbach-Schur reduction one assumes that the commutator $[\Pi,H]$ is small and thus the product $\Pi H\Pi^\bot$ appearing in $S(e)$ is small and one can make a perturbative argument for the invertibility. In our situation this is not the case, but we are interested in the spectral analysis in a narrow window around a specific value ($\lambda_-(\theta_0)=\lambda_+(\theta_0)=0$) and thus we will use as small parameter  \textit{the width of this spectral window}, so that we can make an expansion of the resolvent with respect to this small parameter around $0$. In this sense we notice that $0\in \mathring{I}\subset\rho(\Pi^\bot H\Pi^\bot)$ and for $e\in\mathring{I}$ we use the resolvent equation in order to write
$$
R_\bot(e)\,=\,R_\bot(0)\,+\,eR_\bot(0)^2\,+\,e^2R_\bot(0)^2R_\bot(e).
$$
We insert this expansion in \eqref{F-Se} and get:
\begin{equation*}
S(e)=\,\Pi\Big(H\,-\,\Pi H\Pi^\bot R_\bot(0)\Pi^\bot H\Pi\Big)\Pi\,-\,eY\,-\,e^2\Pi H\Pi^\bot R_\bot(0)^2R_\bot(e)\Pi^\bot H\Pi\,.
\end{equation*} 
Thus $S(e)$ is invertible in $\Pi\mathcal{H}$ if and only if the following operator is invertible in $\Pi\mathcal{H}$:
$$
Y^{-1/2}S(e)Y^{-1/2}\ =\ \big(\widetilde{H}-e\Pi\big)\,-\,e^2Y^{-1/2}\Pi H\Pi^\bot R_\bot(0)^2R_\bot(e)\Pi^\bot H\Pi Y^{-1/2}.
$$
Let us introduce for later use $$X(e):=e^2R_\bot(0)^2R_\bot(e)\,.$$
In order to obtain some uniform estimates for our type of Hamiltonians with a spectral gap (like in Definition \ref{H-band-red}), we shall consider a smaller compact interval $I^\prime\subset \mathring{I}$ containing $0$ in its interior. The important parameter to consider is then $d_{I'}:=\dist(\R\setminus \mathring{I},I^\prime)=\underset{x\in\mathbb{R}\setminus\mathring{I},y\in I^\prime}{\inf}|x-y|>0$. We also have
$$
\mathring{I}:=(\Lambda_-,\Lambda_+),\quad \|R_\bot(0)\|\,\leq\,\Big(\min\{\Lambda_-,\Lambda_+\}\Big)^{-1}\,:= \,\Lambda_I^{-1}.
$$
Moreover we shall use the following parameter
$$
\ell_{I^\prime}:=\sup\{|t|\,,\,t\in  I^\prime\}.
$$

\begin{proposition}\label{P-FS-arg}
	 Given any admissible triple $(H,\Pi,\mathring{I})$ (see Definition \ref{H-band-red})   and with $\widetilde{H}$ defined by \eqref{F-Htilde}, for any  compact interval $I^\prime\subset \mathring{I}$ containing $0$ in its interior,  we have the estimate:
	$$
	\max\Big \{\sup_{e\in \sigma(H)\cap I'}\dist \Big(e\,,\,\sigma(\widetilde{H})\Big)\,,\sup_{e\in \sigma(\widetilde{H})\cap I'}\dist \Big(e\,,\,\sigma(H)\Big) \Big\}\leq\,\|H\Pi\|^2\left(\frac{\ell_{I^\prime}}{\Lambda_I}\right)^2\frac{1}{d_{I'}}.
	$$
\end{proposition}
\begin{proof}
	Let us first estimate the operator norm of $R_\bot(e)$ for some $e\in I^\prime$ using standard calculus with self-adjoint operators.
	$$
	\big\|R_\bot(e)\big\|\,\leq\,\Big(\underset{t\in I^\prime}{\min}\,\dist\big(t,\sigma(\Pi^\bot H\Pi^\bot\big)\Big)^{-1}\,\leq\,d_{I'}^{-1}\,.
	$$
	Thus:
	$$
	\big\|X(e)\big\|\,\leq\,\frac{e^2}{\Lambda_I^2d_{I^\prime}}\,\leq\,\left(\frac{\ell_{I^\prime}}{\Lambda_I}\right)^2\frac{1}{d_{I'}}
	$$
	
	Let us suppose that some $e\in I^\prime$ is in the resolvent set of $\widetilde{H}$, as self-adjoint operator in $\Pi\mathcal{H}$. From \eqref{F-Se} we obtain that
	$$
	Y^{-1/2}S(e)Y^{-1/2}\,=\,\Big(\Pi\,-\,\Pi H\Pi^\bot X(e)\Pi^\bot H\Pi(\widetilde{H}-e\Pi)^{-1}\Big)(\widetilde{H}-e\Pi).
	$$
	From the above results and the usual calculus with self-adjoint operators it is easy to obtain that
	$$
	\left\|\Pi H\Pi^\bot X(e)\Pi^\bot H\Pi(\widetilde{H}-e\Pi)^{-1}\right\|\leq\|H\Pi\|^2\big\|X(e)\big\|\left\|(\widetilde{H}-e\Pi)^{-1}\right\|\,\leq\,\|H\Pi\|^2\left(\frac{\ell_{I^\prime}}{\Lambda_I}\right)^2\frac{1}{d_{I'}}\frac{1}{\dist(e,\sigma(\widetilde{H}))}.
	$$
	We conclude that for 
	$$
	\dist\big(e,\sigma(\widetilde{H})\big)\,>\,\|H\Pi\|^2\left(\frac{\ell_{I^\prime}}{\Lambda_I}\right)^2\frac{1}{d_{I'}}
	$$
	the operator $Y^{-1/2}S(e)Y^{-1/2}$ and thus also $S(e)$ are invertible in $\Pi\mathcal{H}$. The Feshbach-Schur argument implies then that $H-e\bb1$ is also invertible in $\mathcal{H}$. This means that 
	$$
	e\in\sigma(H)\cap I^\prime\ \Longrightarrow\ \dist\big(e,\sigma(\widetilde{H})\big)\,\leq\,\|H\Pi\|^2\left(\frac{\ell_{I^\prime}}{\Lambda_I}\right)^2\frac{1}{d_{I'}}.
	$$
	
	Let us now suppose that some $e\in I^\prime$ is in the resolvent set of $H$, as self-adjoint operator in $\mathcal{H}$. From the Feshbach-Schur argument discussed before this proposition we know then that $S(e)$ is invertible in $\Pi\mathcal{H}$ and \eqref{F-inv-Se} is true.
	The same equality \eqref{F-Se} tells us that $\widetilde{H}-e\Pi$ is invertible in $\Pi\mathcal{H}$ if and only if the following operator is invertible in $\Pi\mathcal{H}$
	$$
	S(e)+\Pi H\Pi^\bot X(e)\Pi^\bot H\Pi=\Big(\Pi+\Pi H\Pi^\bot X(e)\Pi^\bot H\Pi (H-e\bb1)^{-1}\Pi\Big)S(e).
	$$
	Thus let us estimate
	$$
	\left\|\Pi H\Pi^\bot X(e)\Pi^\bot H\Pi (H-e\bb1)^{-1}\Pi\right\|\,\leq\,\|H\Pi\|^2\big\|X(e)\big\|\left\|(H-e\bb1)^{-1}\right\|\leq\|H\Pi\|^2\left(\frac{\ell_{I^\prime}}{\Lambda_I}\right)^2\frac{1}{d_I}\frac{1}{\dist(e,\sigma(H))}.
	$$
	We conclude that for 
	\begin{align}\label{mai5}
	\dist\big(e,\sigma(H)\big)\,>\,\|H\Pi\|^2\left(\frac{\ell_{I^\prime}}{\Lambda_I}\right)^2\frac{1}{d_{I'}}\,,
	\end{align}
	the operator $\widetilde{H}-e\Pi$ is invertible in $\Pi\mathcal{H}$.  This means that 
	$$
	e\in\sigma(\widetilde{H})\cap I^\prime\ \Longrightarrow\ \dist\big(e,\sigma(H)\big)\,\leq\,\|H\Pi\|^2\left(\frac{\ell_{I^\prime}}{\Lambda_I}\right)^2\frac{1}{d_{I'}}.
	$$
	This concludes the proof.
\end{proof}

\bibliographystyle{amsplain}

\end{document}